\documentclass[10pt,reqno,a4paper]{amsart}
\usepackage[utf8]{inputenc}
\usepackage{amsmath}
\usepackage[british]{babel}
\usepackage{amssymb,amstext,amsthm,eucal,bbm,mathrsfs,amscd,bbold}
\usepackage{amsfonts}
\usepackage[svgnames,table]{xcolor}
\usepackage{booktabs} 
\usepackage{graphicx}
\usepackage[T1]{fontenc}
\usepackage{enumitem}
\usepackage[small]{eulervm}
\usepackage{tgpagella}
\usepackage{verbatim}
\usepackage{mathtools}
\usepackage{booktabs}
\usepackage{tensor}
\usepackage{array}
\usepackage[left=2cm,right=2cm,top=2cm,bottom=2cm]{geometry}
\usepackage{enumitem}
\usepackage{caption}
\usepackage{bm}
\usepackage[hyperindex,breaklinks]{hyperref}

\hypersetup{%
  pdftitle   = {Generalised Bargmann Superalgebras},
  pdfkeywords = {Lie superalgebra, kinematical, aristotelian, superspace},
  pdfauthor  = {Ross Grassie},
  pdfcreator = {\LaTeX\ with package \flqq hyperref\frqq},
  linkcolor=Black,
  citecolor=NavyBlue,
  urlcolor=OrangeRed,
  anchorcolor=OrangeRed,
  colorlinks=true
}

\usepackage{pifont}

\usepackage[sorting=none, backend=biber, block=ragged, citestyle=numeric-comp]{biblatex}
\numberwithin{equation}{subsection}

\theoremstyle{plain}
\newcounter{dummy} \numberwithin{dummy}{section}
\newtheorem{lemma}[dummy]{Lemma}

\theoremstyle{definition}

\newcommand{\N}{\mathcal{N}}

\newcommand{\bH}{\mathrm{H}}
\newcommand{\bZ}{\mathrm{Z}}
\newcommand{\bJ}{\mathrm{J}}
\newcommand{\bB}{\mathrm{B}}
\newcommand{\bP}{\mathrm{P}}
\newcommand{\bQ}{\mathrm{Q}}

\newcommand{\J}{\boldsymbol{\mathrm{J}}}
\newcommand{\B}{\boldsymbol{\mathrm{B}}}
\renewcommand{\P}{\boldsymbol{\mathrm{P}}}
\newcommand{\Q}{\boldsymbol{\mathrm{Q}}}

\newcommand{\sH}{\mathsf{H}}
\newcommand{\sZ}{\mathsf{Z}}
\newcommand{\sJ}{\mathsf{J}}
\newcommand{\sB}{\mathsf{B}}
\newcommand{\sP}{\mathsf{P}}
\newcommand{\sQ}{\mathsf{Q}}

\newcommand{\s}{\mathfrak{s}}
\renewcommand{\r}{\mathfrak{r}}
\newcommand{\e}{\mathfrak{e}}
\newcommand{\p}{\mathfrak{p}}
\newcommand{\co}{\mathfrak{co}}
\renewcommand{\k}{\mathfrak{k}}
\newcommand{\g}{\mathfrak{g}}
\newcommand{\so}{\mathfrak{so}}
\renewcommand{\sp}{\mathfrak{sp}}
\newcommand{\n}{\mathfrak{n}}
\newcommand{\osp}{\mathfrak{osp}}
\renewcommand{\a}{\mathfrak{a}}

\newcommand{\mh}{H}
\newcommand{\mz}{Z}
\newcommand{\mb}{B}
\renewcommand{\mp}{P}

\newcommand{\hh}{\mathbb{h}}
\newcommand{\zz}{\mathbb{z}}
\newcommand{\bb}{\mathbb{b}}
\newcommand{\pp}{\mathbb{p}}

\newcommand{\vt}{\bm{\theta}}
\newcommand{\JJ}{\mathbb{J}}
\newcommand{\BB}{\mathbb{B}}
\newcommand{\PP}{\mathbb{P}}

\newcommand{\cc}{\mathbb{c}}
\newcommand{\dd}{\mathbb{d}}
\newcommand{\ee}{\mathbb{e}}
\newcommand{\ff}{\mathbb{f}}
\renewcommand{\gg}{\mathbb{g}}
\renewcommand{\ll}{\mathbb{l}}
\newcommand{\mm}{\mathbb{m}}
\newcommand{\nn}{\mathbb{n}}
\newcommand{\qq}{\mathbb{q}}
\newcommand{\rr}{\mathbb{r}}
\renewcommand{\ss}{\mathbb{s}}
\newcommand{\uu}{\mathbb{u}}
\newcommand{\vv}{\mathbb{v}}

\newcommand{\ii}{\mathbb{i}}
\newcommand{\jj}{\mathbb{j}}
\newcommand{\kk}{\mathbb{k}}

\newcommand{\overbar}[1]{\mkern 1.5mu\overline{\mkern-1.5mu#1\mkern-1.5mu}\mkern 1.5mu}

\newcommand{\Ad}{\operatorname{Ad}}
\renewcommand{\Re}{\operatorname{Re}}
\renewcommand{\Im}{\operatorname{Im}}

\newcommand{\Aut}{\operatorname{Aut}}
\newcommand{\Mat}{\operatorname{Mat}}
\newcommand{\ad}{\operatorname{ad}}
\newcommand{\G}{\operatorname{G}}
\newcommand{\GL}{\operatorname{GL}}
\newcommand{\SO}{\operatorname{SO}}

\newcommand{\Sp}{\operatorname{Sp}}

\newcommand{\Hom}{\operatorname{Hom}}

\newcommand{\cV}{\mathscr{V}}
\newcommand{\cJ}{\mathscr{J}}
\newcommand{\cS}{\mathscr{S}}

\definecolor{gris}{rgb}{0.5,0.5,0.5}

\begin{document}

\title{Generalised Bargmann Superalgebras}
\author[Grassie]{Ross Grassie}
\address{Maxwell Institute and School of Mathematics, The University
  of Edinburgh, James Clerk Maxwell Building, Peter Guthrie Tait Road,
  Edinburgh EH9 3FD, Scotland, United Kingdom}
\email[RG]{\href{mailto:ross.grassie@ed.ac.uk}{ross.grassie@ed.ac.uk}}
\begin{abstract}
The Bargmann algebra and centrally-extended Newton-Hooke algebras describe the non-relativistic symmetries of massive particles in flat and curved spacetimes,  respectively.  These three algebras all arise as deformations of the universal central-extension of the static kinematical Lie algebra.  In this paper, we classify the $\N=1$ super-extensions for each of these algebras in $(3+1)$-dimensions, up to isomorphism.  We then identify the non-empty branches of the algebraic variety describing the $\N=2$ super-extensions of these algebras. We find 9 isomorphism classes in the $\N=1$ case and 22 branches in the $\N=2$ case.  We then give a brief discussion on some applications of these Lie superalgebras, including their possible uses for non-relativistic supergravity and holography. 
\end{abstract}
\dedicatory{In memory of Winifred Whittle}
\maketitle
\tableofcontents
\section{Introduction}
\subsection{Background}
The idea of modelling spacetime geometrically was pioneered by Albert Einstein in his theory of general relativity.  The equivalence principle underlying this theory requires the spacetime geometry to be pseudo-Riemannian.  Thus, this model builds upon special relativity; namely, it has local Lorentz symmetry at its core.  However, not all physical phenomena of interest appear to us in an intrinsically Lorentzian relativistic way.  A large body of literature is emerging, which aims to generalise Einstein's idea to allow other kinematical symmetries to dictate the geometry of spacetime.  
\\ \\
The question of which symmetries may describe physical systems has lead to the recent classification of kinematical Lie algebras \cite{Andrzejewski:2018gmz, Figueroa-OFarrill:2017sfs, Figueroa-OFarrill:2017tcy, Figueroa-OFarrill:2017ycu}.  This classification divides into four distinct types: Riemannian, pseudo-Riemannian, Galilean and Carrollian.  The former two are those traditionally posited to describe kinematical systems, but it is the latter two that are receiving renewed interest.  
\\ \\
One reason for this interest has its roots in the search for a quantum theory of gravity. Conventional approaches try to access the quantum gravity corner of the Bronstein cube via general relativity or relativistic quantum field theory.  However, a possible third approach to relativistic quantum gravity may lie in first producing a non-relativistic quantum gravity theory.
\\ \\
 As is often the case when exploring new territory, it is best to start from somewhere familiar.  Beginning from a Lorentzian system in $(D+2)$-dimensions, a null reduction down to $(D+1)$-dimensions produces a non-relativistic counterpart.  This process has led not only to non-relativistic versions of gravity \cite{Grosvenor:2017dfs, Bergshoeff:2017dqq,JULIA1995291,10.1007/978-3-319-68445-1_43} but non-relativistic versions of string theory \cite{Kluson:2018egd,Harmark:2019upf}, holography \cite{Harmark:2017rpg,Christensen:2013lma} and supersymmetric quantum field theories \cite{Bergshoeff:2020baa, Auzzi:2019kdd}.  There are several other approaches to forming non-relativistic (super)gravity theories currently under investigation, including gauging known symmetry algebras, such as the Bargmann algebra \cite{Andringa:2010it,Hartong:2015zia,Bergshoeff_2019,Bergshoeff:2014gja, Bergshoeff:2017dqq} and super-Bargmann algebra \cite{Andringa:2013mma,Bergshoeff:2014gja}, producing new algebras to gauge and Lagrangians through expansion procedures \cite{Hansen:2018ofj, Hansen:2020pqs,Harmark:2019upf} and finding extensions that may allow for the production of an action \cite{Gomis:2018xmo,Gomis:2019nih,Bergshoeff:2020fiz,deAzcarraga:2019mdn,deAzcarraga:2002xi,Matulich:2019cdo,Papageorgiou:2009zc, Gomis:2019sqv, Bergshoeff:2019ctr, Romano:2019ulw}.  Applying these approaches to Maxwellian extensions of Newtonian gravity is also an active field of research \cite{Concha:2020sjt,Concha:2019lhn,Aviles:2018jzw, Concha:2020ebl}, and non-relativistic symmetries and geometries have even found their place in double field theory \cite{Blair:2019qwi,Morand:2017fnv,Ko:2015rha,Berman:2019izh}.
\\ \\
It may be thought that since the Newton-Cartan (NC) geometry one obtains by gauging the Bargmann algebra may be found via a null reduction of general relativity in flat space \cite{Bergshoeff:2017dqq}, a similar story may be true for the centrally-extended Newton-Hooke algebras, which are the curved equivalents.  Namely, that the non-relativistic geometries obtained through the gauging of the centrally-extended Newton-Hooke algebras may be obtained via null reductions of anti-de Sitter (AdS) and de Sitter (dS) space.  However, this is not the case.  Seeking a unified description of these algebras, we note that the Galilean and Newton-Hooke spacetimes are described by NC structures $(\tau, h)$ in which the clock one-form is closed, $d\tau = 0$ \cite{Figueroa-OFarrill:2019sex}.  In a recent work by José Figueroa-O'Farrill, it is shown that these NC structures arise via a null reduction of a particular type of Bargmann space \cite{Figueroa-OFarrill:2020gpr}.\footnote{These results are also shown in \cite{Gibbons:2003rv, PhysRevD.31.1841, Duval:1990hj}.}  More explicitly, these Bargmann spaces are described by a triple $(B, g, \xi)$ consisting of a $(D+2)$-dimensional Lorentzian manifold $B$, a Lorentzian metric $g$, and a null, nowhere-vanishing vector field $\xi$, such that $\nabla \xi = 0$ \cite{Duval:2014uoa}.  In other words, $g$ is a Brinkmann metric, or $(B, g)$ is a $pp$-wave \cite{Duval:1990hj}.\footnote{Thank you to José Figueroa-O'Farrill for this interpretation.}  Using light-cone coordinates, we may write
\begin{equation}
	g = 2 du dv + A(u, x) du^2 + B_i (u, x) du dx^i + g_{ij} (u, x) dx^i dx^j \quad \text{and} \quad \xi = \frac{\partial}{\partial v},
\end{equation}
where $u$ and $v$ are our light-cone coordinates, and $i, j = 0, 1, ..., D-1$.  Now, taking the quotient of $B$ with respect to the isometry subgroup generated by $\xi$, one obtains an NC space, $N = B/\langle \xi \rangle$.  Using the musical isomorphism
\begin{equation}
	\xi^\flat = g(\xi, -) = du
\end{equation}
produces the clock one-form on the base space, while the spatial metric $h$ is defined 
\begin{equation}
	h (\alpha, \beta) = g^{-1}( \pi^*(\alpha), \pi^*(\beta)),
\end{equation}
where $\alpha, \beta \in \Omega^1(N)$, and $\pi: B \rightarrow N$ is the projection of this trivial principle $\mathbb{R}$ bundle. The Galilean and Newton-Hooke spacetimes then arise through different choices of $h$ for our NC base space $(N, \tau, h)$.  In this picture, the central-extended Galilean and Newton-Hooke algebras may be thought of as the Lie algebras for the isometry groups of the total spaces, or, equivalently, in the following way.  Taking the flat case as our example, the Galilean algebra $\g$ arises as the centraliser of the null Killing vector $\xi$ in the Lie algebra of Killing vector fields $\mathcal{K}$:
\begin{equation}
	\begin{split}
		\mathcal{K} &= \{ X \in \mathfrak{X}(B) \, |\, \mathcal{L}_X g = 0 \} \\
		\g &= \{ X \in \mathcal{K} \,  | \,  [X, \xi] = 0 \}.
	\end{split}
\end{equation}
This Lie algebra may admit a central-extension $\langle \xi \rangle$ fitting into the short exact sequence 
\begin{equation}
	0  \rightarrow \langle \xi \rangle \rightarrow \g \rightarrow \g/ \langle \xi \rangle \rightarrow 0.
\end{equation}
The Bargmann algebra is then the non-trivial central extension, i.e. the extension that does not admit a splitting such that
$\g = \g/ \langle \xi \rangle \oplus \langle \xi \rangle$ as a Lie algebra.  The centrally-extended Newton-Hooke algebras then appear in an identical manner when $\g$ is replaced in this story by $\n_{\pm}$.
\\ \\
The primary aim of the present paper is to provide a set of super-extensions for the Bargmann and centrally-extended Newton-Hooke algebras in $(3+1)$-dimensions to facilitate future investigations into non-relativistic physics.  It may be read as a case study for the development of a formalism for classifying kinematical Lie superalgebras (KLSAs) which began in \cite{Figueroa-OFarrill:2019ucc}.  We will first address the issue of introducing a central extension, $Z$, before describing how to generalise the construction for extended supersymmetry.
\\ \\
As far as we are aware, this is the first classification of super-extensions for these algebras.  Several papers have asked what the possible \textit{"super-kinematics"} are, and discussed either $\N=1$ \cite{Rembielinski:1984xy} or $\N=2$ \cite{Kosinski:1986ts} super-extensions of the kinematical Lie algebras first identified by Bacry and Lévy-LeBlond in their pioneering paper \cite{doi:10.1063/1.1664490}.  Others consider contractions of the anti-de Sitter (AdS) superalgebra $\osp(1|4)$ \cite{Hussin_1999,CampoamorStursberg:2008hm,Huang:2014ega}.  However, we believe this is the first attempt to classify these centrally-extended kinematical Lie algebras.
\\ \\
In the $\N=1$ case, we will provide a full classification, analogous to that of the non-centrally extended kinematical Lie superalgebras in \cite{Figueroa-OFarrill:2019ucc}.  However, the $\N=2$ case is considerably more involved.  Therefore, in this instance, we will stop short of explicitly identifying each super-extension of the algebras and instead highlight which brackets are allowed to be non-vanishing.  Thus we are able to identify important features, such as which generators appear in the bracket $[\Q, \Q]$ and which act on the supercharges, without burdening the reader with excessive detail.  However, should the reader wish to investigate a particular type of super-extension further, we will provide the required tools. 
\\ \\
For the present paper, we will use the classification of non-trivial central extensions of kinematical Lie algebras in $(3+1)$-dimensions presented in \cite{Figueroa-OFarrill:2017ycu} as our starting point (see Table \ref{tab:algebras}).  This list includes the Bargmann $\hat{\g}$ and centrally-extended Newton-Hooke algebras $\hat{\n}_{\pm}$, as well as the centrally-extended static kinematical Lie algebra $\hat{\a}$, of which the former three are all deformations.\footnote{In this paper, we will not use the standard notation for the static kinematical Lie algebra $\s$; instead, we will use $\a$ for Aristotelian, since this is the sole example of an Aristotelian algebra in this paper. This notational change allows us to use $\s$ when referring to kinematical Lie superalgebras.}  From now on, we will refer to these algebras collectively as \textit{generalised Bargmann algebras}.
\begin{table}[h!]
  \centering
  \caption{Generalised Bargmann algebras in $D=3$}
  \label{tab:algebras}
  \rowcolors{2}{blue!10}{white}
    \begin{tabular}{l|*{5}{>{$}l<{$}}|l}\toprule
      \multicolumn{1}{c|}{GBA} & \multicolumn{5}{c|}{Nonzero Lie brackets in addition to $[\J,\J] = \J$, $[\J, \B] = \B$, $[\J,\P] = \P$} & \multicolumn{1}{c}{Comments}\\\midrule
      \hypertarget{a}{1} & [\B,\P] = \bZ & &  &  & & $\hat{\a}$ \\
      \hypertarget{n-}{2} & [\B,\P] = \bZ & [\bH,\B] = \B & [\bH,\P]= -\P &  & & $\hat{\n}_-$\\
      \hypertarget{n+}{3} & [\B,\P] = \bZ & [\bH,\B] = \P & [\bH,\P]= -\B & & & $\hat{\n}_+$\\
      \hypertarget{g}{4} & [\B,\P] = \bZ & [\bH,\B]=-\P & & & & $\hat{\g}$\\
      \bottomrule
    \end{tabular}
\end{table}
\subsection{Outline of Paper}
In section \ref{sec:formalism_intro}, we will introduce the objects of interest and describe the classification problem we wish to solve.  In particular, section \ref{subsec:FI_KLAs} will introduce kinematical Lie algebras and their central extensions, and setup the \textit{universal} generalised Bargmann algebra, which plays a crucial role in ensuring we do not duplicate our calculations later on.  In section \ref{subsec:FI_KLSAs}, we give a brief definition of kinematical Lie superalgebras and define the classification problem. The exact details of the quaternionic formalism in the $\N=1$ and $\N=2$ cases are left until sections \ref{subsec:N1_setup} and \ref{subsec:N2_setup}, respectively.  
\\ \\
In section \ref{sec:N1_classification}, we classify the $\N=1$ super-extensions of the generalised Bargmann algebras.  As mentioned above, section \ref{subsec:N1_setup} sets up the quaternionic formalism for the $\N=1$ super-extensions.  There are also some useful preliminary results presented here, and a discussion on the possible basis transformations $\G \subset \GL(\s_{\bar{0}}) \times \GL(\s_{\bar{1}})$ for the $\N=1$ kinematical Lie superalgebras.  In section \ref{subsec:N1_class}, we classify the $\N=1$ super-extensions of the four algebras before summarising our findings in section \ref{subsec:N1_summary}.  This final section includes a discussion on unpacking the quaternionic formalism to a notation that may be more familiar to the reader. 
\\ \\
Section \ref{sec:N2_classification} investigates the $\N=2$ super-extensions of the generalised Bargmann algebras.  Following the same pattern as section \ref{sec:N1_classification}, section \ref{subsec:N2_setup} introduces the quaternionic formalism for the $\N=2$ case.  It includes some preliminary results and the definition of the group of basis transformations $\G$ for these superalgebras.  Section \ref{subsec:N2_branches} identifies four different branches of possible super-extension for the investigation in section \ref{subsec:N2_class}.  The results are summarised in section \ref{subsec:N2_summary}, where there is also a brief discussion on unpacking the $\N=2$ quaternionic formalism.  Finally, in section \ref{sec:discussion}, we give some concluding remarks and discuss possible future directions for this work. 
\\ \\
The classifications in this paper adopt an unconventional formalism, and the details of how we arrive at our results can be quite involved, particularly for the $\N=2$ case.  For those pressed for time or interested solely in the superalgebras in the classification, the following sections present the required reading to access the results.  For the quaternionic formalism in the $\N=1$ case, see section \ref{subsec:N1_setup}, and, in the $\N=2$ case, see section \ref{subsec:N2_setup}.  Sections \ref{subsec:N1_summary} and \ref{subsec:N2_summary} include discussions on unpacking the notation to something more conventional in the $\N=1$ and $\N=2$ cases, respectively.  Table \ref{tab:N1_classification} contains the results of the $\N=1$ classification, and Table \ref{tab:N2_classification} contains the results of the $\N=2$ analysis. 
\section{Introduction to Formalism} \label{sec:formalism_intro}
In this section, we will introduce the concept of a kinematical Lie (super)algebra and set up the classification problem we wish to solve in this paper.  Of particular importance is the definition of our universal generalised Bargmann algebra, which will be used throughout both the $\N=1$ and $\N=2$ classification problems.  All the statements made in this section assume that we are working in three spatial dimensions, $D=3$. 
\subsection{Kinematical Lie Algebras} \label{subsec:FI_KLAs}
A kinematical Lie algebra (KLA) $\k$ is a 10-dimensional real Lie algebra containing a rotational subalgebra $\r$ isomorphic to $\so(3)$ such that, under the adjoint action of $\r$, it decomposes as $\k = \r \,\oplus\, 2V \,\oplus\, \mathbb{R}$, where $V$ is a three-dimensional $\so(3)$ vector module and $\mathbb{R}$ is a one-dimensional $\so(3)$ scalar module.  We will denote the real basis for this Lie algebra as $\{\bJ_i, \bB_i, \bP_i, \bH\}$, where $\bJ_i$ is the generator for the subalgebra $\r$, $\bB_i$ and $\bP_i$ span our two copies of $V$, and $\bH$ spans the $\so(3)$ scalar module.  The Lie brackets common to all kinematical Lie algebras are
\begin{equation} \label{eq:kinematical_brackets}
	[\bJ_i, \bJ_j] = \epsilon_{ijk} \bJ_k \quad [\bJ_i, \bB_j] = \epsilon_{ijk} \bB_k \quad [\bJ_i, \bP_j] = \epsilon_{ijk} \bP_k
	\quad [\bJ_i, \bH] = 0.
\end{equation}
The classification of these algebras was first presented in \cite{doi:10.1063/1.527306}, completing the work of \cite{doi:10.1063/1.1664490}, in which the kinematical Lie algebras admitting parity and time-reversal automorphisms were
classified. Throughout this paper, we will frequently use the following abbreviated notation for the Lie brackets of these algebras.
\begin{gather}
		[\bJ_i, \bB_j] = \epsilon_{ijk} \bB_k \quad \text{is equivalent to} \quad  [\J, \B] = \B \nonumber \\
		[\bH, \bB_i] = \bP_i \quad \text{is equivalent to} \quad [\bH, \B] = \P \\
		[\bB_i, \bP_j] = \delta_{ij} \bH \quad \text{is equivalent to} \quad [\B, \P] = \bH \nonumber ,
\end{gather}
\textit{et cetera}.  The static kinematical Lie algebra, of which all other KLAs are deformations, admits a universal central extension \cite{Figueroa-OFarrill:2017ycu}.  This enhancement of the algebra requires an additional $\so(3)$ scalar module in the underlying vector space.  Let $\bZ$ span this extra copy of $\mathbb{R}$.  Our centrally-extended static kinematical Lie algebra is spanned by $\J, \B, \P, \bH$, and $\bZ$, with non-vanishing brackets 
\begin{equation}
	[\J, \J] = \J \quad [\J, \B] = \B \quad [\J, \P] = \P \quad [\B, \P] = \bZ.
\end{equation}
All other centrally-extended kinematical Lie algebras are deformations of this algebra; therefore, these brackets are common to all such algebras. For a complete classification of the centrally-extended kinematical Lie algebras see Table \ref{tab:ce_algebras}, taken from \cite{Figueroa-OFarrill:2017ycu}.  The three sections of this table, starting from the top, are the non-trivial central extensions, the trivial central extensions, and, finally, the non-central extensions of kinematical Lie algebras. 
\begin{table}[h!]
  \centering
  \caption{Centrally-extended kinematical Lie algebras in $D=3$}
  \label{tab:ce_algebras}
  \rowcolors{2}{blue!10}{white}
  \resizebox{\textwidth}{!}{
    \begin{tabular}{l|*{5}{>{$}l<{$}}|l}\toprule
      \multicolumn{1}{c|}{KLA} & \multicolumn{5}{c|}{Nonzero Lie brackets in addition to $[\J ,\J] = \J$, $[\J , \B] = \B$, $[\J ,\P] = \P$} & \multicolumn{1}{c}{Comments}\\\midrule
      1 & [\B ,\P] = \bZ & &  &  & & $\hat{\a}$ \\
      2 & [\B ,\P] = \bZ & [\bH,\B] = \B & [\bH,\P]= -\P &  & & $\hat{\n}_-$\\
      3 & [\B ,\P] = \bZ & [\bH,\B] = \P & [\bH,\P]= -\B & & & $\hat{\n}_+$\\
      4 & [\B ,\P] = \bZ & [\bH,\B]=-\P & & & & $\hat{\g}$\\
	  \hline
      5 & [\B ,\P] = \bH & [\bH,\B] = \P &  &  & [\B , \B] = \J & $\e \oplus \mathbb{R}Z$ \\
      6 & [\B ,\P] = \bH & [\bH,\B] = -\P &  &  & [\B , \B] = -\J & $\p \oplus \mathbb{R}Z$\\
      7 & [\B ,\P] = \bH + \J & [\bH,\B] = -\B & [\bH,\P]= \P & & & $\so(4,1) \oplus \mathbb{R}Z$\\
      8 & [\B ,\P] = \bH & [\bH,\B]= \P & [\bH , \P] = -\B & [\P, \P] = \J & [\B, \B] = \J & $\so(5) \oplus \mathbb{R}Z$\\
      9 & [\B ,\P] = \bH & [\bH,\B] = -\P &  [\bH, \P] = \B & [\P, \P] = -\J  & [\B, \B] = -\J & $\so(3, 2) \oplus \mathbb{R}Z$ \\
      \hline
      10 & [\B ,\P] = \bZ & [\bH,\B] = \B & [\bH,\P]= \P & [\bH, \bZ] = 2 \bZ & & \\
      11 & [\B ,\P] = \bZ & [\bH,\B] = \gamma \B & [\bH,\P]= \P & [\bH, \bZ] = (\gamma + 1) \bZ & & $\gamma \in (-1, 1)$\\
      12 & [\B ,\P] = \bZ & [\bH,\B]= \B+ \P & [\bH, \P] = \P & [\bH, \bZ] = 2 \bZ &  & \\
      13 & [\B ,\P] = \bZ & [\bH,\B] = \alpha \B + \P & [\bH,\P]= -\B + \alpha \P & [\bH, \bZ] = 2\alpha \bZ & & $\alpha > 0$ \\
      14 & [\B ,\P] = \bZ & [\bZ,\B] = \P & [\bH,\P]= \P & [\bH, \bZ] = \bZ & [\B, \B] = \J & $\co(4) \ltimes \mathbb{R}^4 $\\
      15 & [\B ,\P] = \bZ & [\bZ,\B]= -\P & [\bH, \P] = \P & [\bH, \bZ] = \bZ &  [\B, \B] = - \J & $\co(3,1) \ltimes \mathbb{R}^{3, 1}$ \\ 
      \bottomrule
    \end{tabular}
    }
\end{table} \\
In the present paper, we will focus solely on the first of these sections, and it shall be exclusively the algebras of this section that we are referring to when using the term \textit{generalised Bargmann algebras}. It is useful for our later calculations to define a \textit{universal} generalised Bargmann algebra.  In addition to the standard kinematical brackets given in \eqref{eq:kinematical_brackets}, this algebra has non-vanishing brackets
\begin{equation}
	[\B, \P] = \bZ \qquad [\bH, \B] = \lambda \B + \mu \P \qquad [\bH, \P] = \eta \B + \varepsilon \P,
\end{equation}
where $\lambda, \mu, \eta, \varepsilon \in \mathbb{R}$.\footnote{Note, the universal generalised Bargmann algebra is not a Lie algebra for arbitrary $\lambda, \mu, \eta, \varepsilon$.  It is used here simply as a computational tool.}  Setting these four parameters to certain values allows us to reduce to the four cases of interest.  For example, 
$\hat{\g}$ is given by setting $\lambda=\eta=\varepsilon=0$ and $\mu=-1$.  By beginning with the universal algebra, and only picking our parameters, and, thus, our algebra, when we can no longer make progress in the universal case, we reduce the amount of repetition in our calculations.
\subsection{Kinematical Lie Superalgebras} \label{subsec:FI_KLSAs}
An $\N$-extended kinematical Lie superalgebra (KLSA) $\s$ is a real Lie superalgebra $\s = \s_{\bar{0}} \oplus \s_{\bar{1}}$, such that $\s_{\bar{0}} = \k$ is a kinematical Lie algebra, and $\s_{\bar{1}}$ consists of $\N$ copies of $S$, the real four-dimensional spinor module of the rotational subalgebra $\r \cong \so(3)$.
\\ \\
Under the isomorphism $\r \cong \sp(1)$, $S$ may be described as a copy of the quaternions where $\r$ acts via left quaternion multiplication.  We will denote the quaternions by $\mathbb{H}$ and the quaternion units by $\ii, \jj, \kk$, where $\ii \jj = \kk$, and $\jj \ii = -\kk$.  Now, to incorporate the $\s_{\bar{0}}$ generators into this formalism, we introduce the following injective $\mathbb{R}$-linear maps:\footnote{Solely to keep the notation consistent, when referring to the $\so(3)$ scalar module basis elements $\bH$ and $\bZ$ in this formalism, we will use $\sH$ and $\sZ$.  Therefore, if we consider $\sJ(\omega) = \omega_i \bJ_i$ to be the map between $\sJ$ and $\bJ$, the map between $\sH$ and $\bH$ is $\sH = \bH$.  Similarly, $\sZ = \bZ$.}
\begin{equation}
	\begin{split}
	\sJ: \Im(\mathbb{H}) \rightarrow \s_{\bar{0}} \quad &\text{such that} \quad \sJ(\omega) = \omega_i \bJ_i 
		\quad \text{where} \quad \omega = \omega_1 \ii + \omega_2 \jj + \omega_3 \kk \in \Im(\mathbb{H}) \\
	\sB: \Im(\mathbb{H}) \rightarrow \s_{\bar{0}} \quad &\text{such that} \quad \sB(\beta) = \beta_i \bB_i 
		\quad \text{where} \quad \beta = \beta_1 \ii + \beta_2 \jj + \beta_3 \kk \in \Im(\mathbb{H}) \\
	\sP: \Im(\mathbb{H}) \rightarrow \s_{\bar{0}} \quad &\text{such that} \quad \sP(\pi) = \pi_i \bP_i 
		\quad \text{where} \quad \pi = \pi_1 \ii + \pi_2 \jj + \pi_3 \kk \in \Im(\mathbb{H}). \\	
	\end{split}			
\end{equation}
We may now write the kinematical brackets of \eqref{eq:kinematical_brackets} using these maps and the standard quaternion calculus as follows.
\begin{equation}
	\begin{split}
	[\J, \J] = \J \quad &\implies \quad [\sJ(\omega), \sJ(\omega')] = \tfrac12 \sJ([\omega, \omega']) \\
	[\J, \B] = \B \quad &\implies \quad [\sJ(\omega), \sB(\beta)] = \tfrac12 \sB([\omega, \beta]) \\
	[\J, \P] = \P \quad &\implies \quad [\sJ(\omega), \sP(\pi)] = \tfrac12 \sP([\omega, \pi]) \\
	[\J, \bH] = 0 \quad &\implies \quad [\sJ(\omega), \sH] = 0 ,
	\end{split}
\end{equation}
where $\omega, \beta, \pi \in Im(\mathbb{H})$, $[\omega, \beta] := \omega\beta - \beta\omega$, and $\omega\beta$ is given by quaternion multiplication.  The additional brackets for the universal generalised Bargmann algebra may also be written in this quaternionic notation as
\begin{equation}
	\begin{split}
		[\B, \P] = \bZ \quad &\implies \quad [\sB(\beta), \sP(\pi)] = Re(\bar{\beta}\pi) \sZ \\
		[\bH, \B] = \lambda \B + \mu \P \quad &\implies \quad [\sH, \sB(\beta)] = \lambda \sB(\beta) + \mu \sP(\beta) \\
		[\bH, \P] = \eta \B + \varepsilon \P \quad &\implies \quad [\sH, \sP(\pi)] = \eta \sB(\pi) + \varepsilon \sP(\pi).
	\end{split}
\end{equation}
To write the $\s_{\bar{1}}$ generators in this language, we will use the injective $\mathbb{R}$-linear map
\begin{equation}
	\sQ: \mathbb{H}^{\N} \rightarrow \s_{\bar{1}} \quad \text{such that} \quad \sQ(\theta) = \theta_a \bQ_a
	 \quad \text{where} \quad \theta \in \mathbb{H}^{\N},
\end{equation}
and $\{\bQ_a\}$ form a real basis for $\s_{\bar{1}}$.\footnote{It may be important to note at this stage that $\theta = \theta_4 + \theta_1 \ii + \theta_2 \jj + \theta_3 \kk$ is just a quaternion with real components $\theta_i$, there are no Grassmann variables here.}  The exact form of the brackets involving
$\Q$ will depend on whether we are considering the $\N=1$ or $\N=2$ case; therefore, we will leave this discussion to sections \ref{subsec:N1_setup} and \ref{subsec:N2_setup}, respectively.  The only important point, for now, is that the bracket $[\J, \Q]$ is fixed from the outset in both these instances.
\\ \\
A super-extension $\s$ of one of our generalised Bargmann algebras $\k$ will be a kinematical Lie superalgebra such that $\s_{\bar{0}}= \k$.  To determine the super-extensions of the generalised Bargmann algebras, we, therefore, begin by letting $\s_{\bar{0}}$ be our universal generalised Bargmann algebra.  We then need to specify the Lie brackets $[\bH, \Q]$, $[\bZ, \Q]$, $[\B, \Q]$, $[\P, \Q]$, and $[\Q, \Q]$.  Each of the $[\s_{\bar{0}}, \s_{\bar{1}}]$ components of the bracket must be an $\r$-equivariant endomorphism of $\s_{\bar{1}}$, while the $[\s_{\bar{1}}, \s_{\bar{1}}]$ component must be an $\r$-equivariant map $\bigodot^2 \s_{\bar{1}} \rightarrow \s_{\bar{0}}$.  The space of possible brackets will be a real vector space $\cV$.  We then use the super-Jacobi identities to cut out an algebraic variety $\cJ \subset \cV$.  Since we are exclusively interested in supersymmetric extensions, we restrict ourselves to those Lie superalgebras for which $[\Q, \Q]$ is non-vanishing, which define a sub-variety $\cS \subset \cJ$.  This sub-variety may be unique to each generalised Bargmann algebra; therefore, it is at this stage we start to set the parameters of the universal algebra, where applicable. The isomorphism classes of the remaining kinematical Lie superalgebras are then in one-to-one correspondence with the orbits of $\cS$ under the subgroup $\G \subset \GL(\s_{\bar{0}}) \times \GL(\s_{\bar{1}})$.  The group $\G$ contains the automorphisms of $\s_{\bar{0}} = \k$ and additional transformations which are induced by the endomorphism ring of $\s_{\bar{1}}$. The form of this subgroup will be discussed in the $\N=1$ and $\N=2$ cases in sections \ref{subsubsec:N1_S_autos} and \ref{subsubsec:N2_S_autos}, respectively.  
\\ \\
In the $\N=1$ case, we will identify each orbit of $\cS$ explicitly, giving a full classification of the generalised Bargmann superalgebras in this instance.  However, in the $\N=2$ case, we will only identify the non-empty branches of $\cS$.  Each branch will have a unique set of $[\s_{\bar{0}}, \s_{\bar{1}}]$ and $[\s_{\bar{1}}, \s_{\bar{1}}]$ brackets for the associated generalised Bargmann algebra.  Thus we are able to highlight the form of the possible super-extensions without spending too much time pinpointing exact coefficients.
\section{$\N=1$ Extensions of the Generalised Bargmann Algebras} \label{sec:N1_classification}
Our investigation into generalised Bargmann superalgebras begins with the simplest
case, $\N=1$.  Section \ref{subsec:N1_setup} will complete our set up for this 
case by specifying the precise form of the $[\s_{\bar{0}}, \s_{\bar{1}}]$ and $[\s_{\bar{1}}, \s_{\bar{1}}]$ brackets.  Section
\ref{subsubsec:N1_S_preliminaries} then gives some preliminary results that will be useful in the
classification of the $\N=1$ extensions, and section \ref{subsubsec:N1_S_autos} will
define the group of basis transformations $\G \subset \GL(\s_{\bar{0}}) \times \GL(\s_{\bar{1}})$, which will allow us to pick out a single representative for each isomorphism class.  In section \ref{subsec:N1_class}, the classification is given before we summarise the results in section \ref{subsec:N1_summary}.
\subsection{Setup for the $\N=1$ Calculation} \label{subsec:N1_setup}
For completeness, recall that, in addition to the standard kinematical Lie brackets, the brackets for the universal generalised Bargmann superalgebra are
\begin{equation}
	\begin{split}
		[\sB(\beta), \sP(\pi)] &= \Re(\bar{\beta}\pi) \sZ \\ 
		[\sH, \sB(\beta)] &= \lambda \sB(\beta) + \mu \sP(\beta) \\
		[\sH, \sP(\pi)] &= \eta \sB(\pi) + \varepsilon \sP(\pi) ,
	\end{split}
\end{equation}
where $\beta, \pi \in \Im(\mathbb{H})$ and $\lambda, \mu, \eta, \varepsilon \in \mathbb{R}$.  We now want to specify the possible $[\s_{\bar{0}}, \s_{\bar{1}}]$
and $[\s_{\bar{1}}, \s_{\bar{1}}]$ brackets.  From \cite{Figueroa-OFarrill:2019ucc},
we have
\begin{equation} \label{eq:N1_general_brackets}
	\begin{split}
		[\sJ(\omega), \sQ(\theta)] &= \tfrac12 \sQ(\omega \theta) \\
		[\sB(\beta), \sQ(\theta)] &= \sQ(\beta \theta \bb) \\
		[\sP(\pi), \sQ(\theta)] &= \sQ(\pi \theta \pp) \\
		[\sH, \sQ(\theta)] &= \sQ(\theta\hh), \\
	\end{split}
\end{equation}
where $\omega, \pi, \beta \in \Im(\mathbb{H})$, $\theta, \bb, \pp, \hh \in \mathbb{H}$.
Since $\bZ$ is just another $\so(3)$ scalar module, and, therefore, the analysis of the 
bracket $[\bZ, \Q]$ will be identical to that of $[\bH, \Q]$, we know we can write
\begin{equation}
	[\sZ, \sQ(\theta)] = \sQ(\theta\zz),
\end{equation}
where $\zz \in \mathbb{H}$.  Having added this additional generator, 
the possible $[\s_{\bar{1}}, \s_{\bar{1}}]$ brackets are now $\so(3)$-equivariant elements
of $Hom_{\mathbb{R}}(\bigodot^2 S, \s_{\bar{0}}) = Hom_{\mathbb{R}}(3V \oplus \mathbb{R}, 3V \oplus 2 \mathbb{R})
= 9\, Hom_{\mathbb{R}}(V, V) \oplus 2\, Hom_{\mathbb{R}}(\mathbb{R}, \mathbb{R})$.
As may have been expected, given that we did not alter the vectorial sector of the algebra, the number of $Hom_{\mathbb{R}}(V, V)$ elements does not change. However, now that we have an additional $\so(3)$ scalar module, we have an additional scalar map, so
\begin{equation} \label{eq:N1_QQ}
	[\sQ(\theta), \sQ(\theta)] = n_0 |\theta|^2 \sH + n_1 |\theta|^2 \sZ - \sJ(\theta\nn_2\bar{\theta}) - \sB(\theta\nn_3\bar{\theta}) - \sP(\theta\nn_4\bar{\theta}),
\end{equation}
where $n_0, n_1 \in \mathbb{R}$ and $\nn_2, \nn_3, \nn_4 \in \Im(\mathbb{H})$.  This expression polarises to
\begin{equation}
	[\sQ(\theta), \sQ(\theta')] = n_0 Re(\bar{\theta}\theta') \sH + n_1 Re(\bar{\theta}\theta') \sZ - \sJ(\theta'\nn_2\bar{\theta}+ \theta \nn_2\bar{\theta'}) 
	- \sB(\theta'\nn_3\bar{\theta}+ \theta \nn_3\bar{\theta'}) -\sP(\theta'\nn_4\bar{\theta}+ \theta \nn_4\bar{\theta'}).
\end{equation}
\subsubsection{Preliminary Results} \label{subsubsec:N1_S_preliminaries}
Following the example of \cite{Figueroa-OFarrill:2019ucc}, we will now consider the super-Jacobi identities
for this super-extension of the generalised Bargmann algebra and use them to define our variety $\cJ$.  We have three types of Jacobi identity to consider
\begin{enumerate}
	\item $(\s_{\bar{0}}, \s_{\bar{0}}, \s_{\bar{1}})$,
	\item $(\s_{\bar{0}}, \s_{\bar{1}}, \s_{\bar{1}})$, and 
	\item $(\s_{\bar{1}}, \s_{\bar{1}}, \s_{\bar{1}})$.
\end{enumerate}
We do not need to consider the $(\s_{\bar{0}}, \s_{\bar{0}}, \s_{\bar{0}})$ case as we already know
that these are satisfied by the generalised Bargmann algebras. Equally, we do not need to include $\J$ in our investigations as the identities involving the rotational subalgebra $\r$ impose the $\so(3)$-equivariance of the brackets, which we already have by construction.  Now, let us consider each type of identity in turn.  In the following
discussions, we will only write down explicitly those identities which are not trivially satisfied.
\\
\paragraph{$(\s_{\bar{0}}, \s_{\bar{0}}, \s_{\bar{1}})$} \label{par:N1_S_P_001} ~\\ \\
By imposing these super-Jacobi identities, we ensure that $\s_{\bar{1}}$ is an 
$\s_{\bar{0}}$ module, not just an $\so(3)$ module. 
The $(\s_{\bar{0}}, \s_{\bar{0}}, \s_{\bar{1}})$ identities can be summarised as follows.
\begin{lemma}\label{lem:N1_001}
  The following relations between $\hh, \zz,\bb,\pp \in \mathbb{H}$ are implied by the
  corresponding $\k$-brackets:
  \begin{equation}
    \begin{split}
    	  [\bH,\bZ] = \lambda \bH + \mu \bZ & \implies [\zz, \hh] = \lambda \hh + \mu \zz \\
      [\bH,\B] = \lambda \B + \mu \P & \implies [\bb,\hh] = \lambda \bb + \mu \pp\\
      [\bH,\P] = \lambda \B + \mu \P & \implies [\pp,\hh] = \lambda \bb + \mu \pp\\
      [\bZ,\B] = \lambda \B + \mu \P & \implies [\bb,\zz] = \lambda \bb + \mu \pp\\
      [\bZ,\P] = \lambda \B + \mu \P & \implies [\pp,\zz] = \lambda \bb + \mu \pp\\
      [\B,\B] = \lambda \B + \mu \P + \nu \J & \implies \bb^2 = \tfrac12 \lambda \bb + \tfrac12 \mu \pp + \tfrac14 \nu\\
      [\P,\P] = \lambda \B + \mu \P + \nu \J & \implies \pp^2 = \tfrac12 \lambda \bb + \tfrac12 \mu \pp + \tfrac14 \nu\\
      [\B,\P] = \lambda \bH + \mu \bZ & \implies \bb \pp + \pp\bb = 0\quad\text{and}\quad [\bb,\pp] = \lambda \hh  + \mu \zz.
    \end{split}
  \end{equation}
\end{lemma}
\begin{proof}
All the results excluding $\bZ$ are taken from Lemma 1 in \cite{Figueroa-OFarrill:2019ucc}, and 
the $[\bZ, \B]$ and $[\bZ, \P]$ results are the same \textit{mutatis mutandis} as $[\bH, \B]$ and $[\bH, \P]$.
Therefore, the only new results are those for $[\bH, \bZ]$ and $[\B, \P]$.  The $[\bH, \bZ, \Q]$ 
identity is written
\begin{equation}
	[\sH, [\sZ, \sQ(\theta)] = [[\sH, \sZ], \sQ(\theta)] + [\sZ, [\sH, \sQ(\theta)]].
\end{equation}
Substituting in the relevant brackets, we find
\begin{equation}
	\sQ(\theta\zz\hh) = \lambda \sQ(\theta\hh) + \mu \sQ(\theta\zz) + \sQ(\theta\hh\zz).
\end{equation}
Using the injectivity of $\sQ$, we arrive at
\begin{equation}
	[\zz, \hh] = \lambda \hh + \mu \zz.
\end{equation}
Finally, the $[\B, \P, \Q]$ identity is
\begin{equation}
	[\sB(\beta), [\sP(\pi), \sQ(\theta)]] = [[\sB(\beta), \sP(\pi)], \sQ(\theta)] + [ \sP(\pi), [\sB(\beta), \sQ(\theta)]].
\end{equation}
Substituting in the brackets from \eqref{eq:N1_general_brackets}, we arrive at
\begin{equation}
	\sQ(\beta\pi \theta \pp\bb) = \Re(\bar{\beta}\pi) (\lambda \sQ(\theta\hh) + \mu \sQ(\theta\zz)) + \sQ(\pi\beta \theta \bb\pp)
\end{equation}
Letting $\beta = \pi = \ii$, we find
\begin{equation}
	[\bb, \pp] = \lambda \hh + \mu \zz.
\end{equation}
Now, let $\beta = \ii$ and $\pi = \jj$.  In this case, the $\lambda$ and $\mu$ terms vanish and
we are left with
\begin{equation}
	\bb\pp + \pp\bb = 0.
\end{equation}
\end{proof}
\paragraph{$(\s_{\bar{0}}, \s_{\bar{1}}, \s_{\bar{1}})$} \label{par:N1_S_P_011} ~ \\ \\
By imposing these super-Jacobi identities, we ensure that the $[\Q, \Q]$ bracket is an $\s_{\bar{0}}$-equivariant map  $\bigodot \s_{\bar{1}} \rightarrow \s_{\bar{0}}$.  The $(\s_{\bar{0}}, \s_{\bar{1}}, \s_{\bar{1}})$ identities can be difficult to
study if we are trying to be completely general; however, we know that 
all four algebras in Table \ref{tab:algebras} can be written as specialisations
of the universal generalised Bargmann algebra:
\begin{equation}
	[\B, \P] = \bZ \qquad [\bH, \B] = \lambda \B + \mu \P \qquad [\bH, \P] = \eta \B + \varepsilon \P,
\end{equation}
where $\lambda, \mu, \eta, \varepsilon \in \mathbb{R}$.  Therefore, we may use the
brackets of this algebra to obtain the following result.
\begin{lemma} \label{lem:N2_011}
The $[\bH, \Q, \Q]$ identity produces the conditions
\begin{equation}
	\begin{split}
		0 &= n_i \Re(\hh) \quad \text{where} \quad i \in \{0, 1\} \\
		0 &= \hh \nn_2 + \nn_2 \bar{\hh} \\
		\lambda \nn_3 + \eta \nn_4 &= \hh \nn_3 + \nn_3 \bar{\hh} \\
		\mu \nn_3 + \varepsilon \nn_4 &= \hh \nn_4 + \nn_4 \bar{\hh}.
	\end{split}
\end{equation}
The $[\bZ, \Q, \Q]$ identity produces the conditions
\begin{equation}
	\begin{split}
		0 &= n_i \Re(\hh) \quad \text{where} \quad i \in \{0, 1\} \\
		0 &= \hh \nn_j + \nn_j \bar{\hh} \quad \text{where} \quad j \in \{2, 3, 4\}.
	\end{split}
\end{equation}
The $[\B, \Q, \Q]$ identity produces the conditions
\begin{equation}
	\begin{split}
		0 &= n_0 \Re(\bar{\theta}\beta \theta \bb) \\
		0 &= \Re(\bar{\beta} \theta (\nn_4 + 2 n_1 \bar{\bb}) \bar{\theta}) \\
		0 &= \theta \nn_2 \overbar{\beta \theta \bb} + \beta \theta \bb \nn_2 \bar{\theta} \\
		\lambda n_0 |\theta|^2 \beta + \tfrac12 [\beta, \theta\nn_2\bar{\theta}] &= \theta \nn_3 \overbar{\beta \theta \bb}+ 
		\beta \theta \bb \nn_3 \bar{\theta} \\
		\mu n_0 |\theta|^2 \beta &= \theta \nn_4 \overbar{\beta \theta \bb} + \beta \theta \bb \nn_4 \bar{\theta}.
	\end{split}
\end{equation}
The $[\P, \Q, \Q]$ identity produces the conditions
\begin{equation}
	\begin{split}
		0 &= n_0 \Re(\bar{\theta}\pi \theta \pp) \\
		0 &= \Re(\bar{\pi} \theta (\nn_3 - 2 n_1 \bar{\pp}) \bar{\theta}) \\
		0 &= \theta \nn_2 \overbar{\pi \theta \pp} + \pi \theta \pp \nn_2 \bar{\theta} \\
		\eta n_0 |\theta|^2 \pi &= \theta \nn_3 \overbar{\pi \theta \pp}+ \pi \theta \pp \nn_3 \bar{\theta} \\
		\varepsilon n_0 |\theta|^2 \pi + \tfrac12 [\pi, \theta\nn_2\bar{\theta}] &= \theta \nn_4 \overbar{\pi \theta \pp} + \pi \theta \pp \nn_4 \bar{\theta}.
	\end{split}
\end{equation}
\end{lemma}
\begin{proof}
The $[\bH, \Q, \Q]$ super-Jacobi identity is written
\begin{equation}
	[\sH, [\sQ(\theta), \sQ(\theta)]] = 2 [ [\sH, \sQ(\theta)], \sQ(\theta)].
\end{equation}
Using \eqref{eq:N1_general_brackets} and \eqref{eq:N1_QQ}, we find
\begin{equation}
	\begin{split}
		- \sB( \theta (\lambda \nn_3 + \eta \nn_4) \bar{\theta}) - \sP(\theta(\mu \nn_3 + \varepsilon \nn_4) \bar{\theta})
		= &\, 2 n_0 \Re(\overbar{\theta\hh}\theta) \sH + 2 n_1 \Re(\overbar{\theta\hh} \theta) \sZ\\ & - \sJ(\theta \nn_2 \overbar{\theta\hh} 
		+ \theta\hh \nn_2 \bar{\theta}) - \sB(\theta \nn_3 \overbar{\theta\hh} + \theta\hh \nn_3 \bar{\theta}) - \sP(\theta\nn_4 \overbar{\theta\hh}
		+ \theta\hh \nn_4 \bar{\theta}).
	\end{split}
\end{equation}
Comparing $\sH$, $\sZ$, $\sJ$, $\sB$, and $\sP$ coefficients, and using the injectivity and 
linearity of the maps $\sJ, \sB$, and $\sP$, we find
\begin{equation}
	\begin{split}
		0 &= n_i \Re(\hh) \quad \text{where} \quad i \in \{0, 1\} \\
		0 &= \hh \nn_2 + \nn_2 \bar{\hh} \\
		\lambda \nn_3 + \eta \nn_4 &= \hh \nn_3 + \nn_3 \bar{\hh} \\
		\mu \nn_3 + \varepsilon \nn_4 &= \hh \nn_4 + \nn_4 \bar{\hh}.	
	\end{split}
\end{equation}
The calculations for the $[\bZ, \Q, \Q]$ identity follows in an analogous manner.  The key difference
is this case is that the L.H.S.
vanishes in all instances since $\bZ$ commutes with all basis elements.  The $[\B, \Q, \Q]$ identity
is 
\begin{equation}
	[\sB(\beta), [\sQ(\theta), \sQ(\theta)]] = 2 [[\sB(\beta), \sQ(\theta)], \sQ(\theta)].
\end{equation}
Substituting in the relevant brackets, the L.H.S. becomes
\begin{equation}
	\text{L.H.S.} = -\lambda n_0 |\theta|^2 \sB(\beta) - \tfrac12 \sB([\beta, \theta \nn_2\bar{\theta}]) - Re(\bar{\beta} \theta\nn_4 \bar{\theta}) \sZ,
\end{equation}
and the R.H.S. becomes
\begin{equation}
	\begin{split}
		\text{R.H.S.} = & \, 2 n_0 \Re(\bar{\theta}\beta \theta\bb) \sH + 2 n_1 \Re(\bar{\theta}\beta \theta\bb) \sZ \\
		& - \sJ(\theta\nn_2 \overbar{\beta \theta\bb} + \beta \theta \bb \nn_2 \bar{\theta}) - \sB(\theta\nn_3 \overbar{\beta \theta\bb} 
		+ \beta \theta \bb \nn_3 \bar{\theta}) - \sP(\theta\nn_4 \overbar{\beta \theta\bb} + \beta \theta \bb \nn_4 \bar{\theta}).
	\end{split}
\end{equation}
Again, comparing coefficients and using the injectivity and linearity of our maps, we
get
\begin{equation}
	\begin{split}
		0 &= n_0 \Re(\bar{\theta}\beta \theta \bb) \\
		0 &= \Re(\bar{\beta} \theta (\nn_4 + 2 n_1 \bar{\bb}) \bar{\theta}) \\
		0 &= \theta \nn_2 \overbar{\beta \theta \bb} + \beta \theta \bb \nn_2 \bar{\theta} \\
		\lambda n_0 |\theta|^2 \beta + \tfrac12 [\beta, \theta\nn_2\bar{\theta}] &= \theta \nn_3 \overbar{\beta \theta \bb}+ 
		\beta \theta \bb \nn_3 \bar{\theta} \\
		\mu n_0 |\theta|^2 \beta &= \theta \nn_4 \overbar{\beta \theta \bb} + \beta \theta \bb \nn_4 \bar{\theta}.
	\end{split}
\end{equation}
The $[\P, \Q, \Q]$ results follow in identical fashion by replacing $\bb$ with $\pp$ and $\beta$
with $\pi$.
\end{proof}
\paragraph{$(\s_{\bar{1}}, \s_{\bar{1}}, \s_{\bar{1}})$} \label{par:N1_S_P_111} ~\\ \\
The last super-Jacobi identity to consider is the $(\s_{\bar{1}}, \s_{\bar{1}}, \s_{\bar{1}})$ case,
$[\Q, \Q, \Q]$.
\begin{lemma}\label{lem:N1_111} 
  The $[\Q,\Q,\Q]$ identity produces the condition
  \begin{equation}
    n_0 \hh + n_1 \zz = \tfrac12 \nn_2  + \nn_3 \bb + \nn_4 \pp.
  \end{equation}
\end{lemma}
\begin{proof}
The $[\Q, \Q, \Q]$ identity is 
\begin{equation}
	0 = [[\sQ(\theta), \sQ(\theta)], \sQ(\theta)].
\end{equation}
Using \eqref{eq:N1_QQ}, and the brackets in \eqref{eq:N1_general_brackets}, 
we find
\begin{equation}
	\begin{split}
		0 &= [n_0 |\theta|^2 \sH + n_1 |\theta|^2 \sZ - \sJ(\theta\nn_2\bar{\theta}) - \sB(\theta\nn_3\bar{\theta}) - \sP(\theta\nn_4
		\bar{\theta}), \sQ(\theta) ] \\
		&= n_0 |\theta|^2 \sQ(\theta\hh) + n_1 |\theta|^2 \sQ(\theta\zz) - \tfrac12 \sQ(\theta\nn_2\bar{\theta}\theta) - \sQ(\theta\nn_3\bar{\theta}\theta\bb) - \sQ(\theta\nn_4
		\bar{\theta}\theta\pp)
	\end{split}
\end{equation}
Since $\sQ$ is injective, this gives us
\begin{equation}
	n_0 \hh + n_1 \zz = \tfrac12 \nn_2  + \nn_3 \bb + \nn_4 \pp.
\end{equation}
\end{proof}
\subsubsection{Basis Transformations} \label{subsubsec:N1_S_autos}
As well as modifying the super-Jacobi identities presented in \cite{Figueroa-OFarrill:2019ucc}, the new $\so(3)$ scalar also impacts the subgroup $\G \subset \GL(\s_{\bar{0}}) \times \GL(\s_{\bar{1}})$ of basis transformation for kinematical Lie superalgebras.  All the automorphisms in $\G$ generated by $\so(3)$ 
remain the same for $\bb$, $\pp$, and $\hh$, but we may now add how $\zz$ transforms.  These automorphisms
act by rotating the three imaginary quaternionic bases $\ii$, $\jj$, and $\kk$ by an element
of $\SO(3)$.  In particular, we have a homomorphism $\Ad: \Sp(1) \rightarrow \Aut(\mathbb{H})$
defined such that for $\uu \in \Sp(1)$ and $\ss \in \mathbb{H}$, $\Ad_{\uu}(\ss) = \uu \ss \bar{\uu}$.  The map $\Ad_{\uu}$ then acts trivially on the real component of $\ss$ and via $\SO(3)$ rotations on
$\Im(\mathbb{H})$.  Therefore, $\tilde{\sB} = \sB\circ \Ad_{\uu}$, $\tilde{\sP} = \sP\circ \Ad_{\uu}$, $\tilde{\sH} = \sH$, $\tilde{\sQ} = \sQ\circ \Ad_{\uu}$. Since $\bZ$ is an $\so(3)$ scalar, $\tilde{\sZ} = \sZ$.  Substituting this with $\tilde{\sQ} = \sQ\circ \Ad_{\uu}$ into the
$[\bZ, \Q]$ bracket, we find that $\tilde{\zz} = \bar{\uu}\zz \uu$.  Additionally, substituting these transformations into the $[\Q, \Q]$ bracket, we see that $c_1$ remains inert.  The other type of transformations to consider
are the $\so(3)$-equivariant maps $\s \rightarrow \s$.  Since we now have two $\so(3)$ 
scalars, we can have
\begin{equation}
	\begin{split}
		\tilde{\sH} &= a \sH + b \sZ \\
		\tilde{\sZ} &= c \sH + d \sZ
	\end{split} \quad \text{where} \quad 
	\begin{pmatrix}
		a & b \\ c & d
	\end{pmatrix} \in \GL(2, \mathbb{R}).
\end{equation}
The $\so(3)$ vector and spinor maps remain unchanged from those given in \cite{Figueroa-OFarrill:2019ucc}.
In particular, $\tilde{\sQ}(s) = \sQ(s \qq)$ for $\qq \in \mathbb{H}^\times$.
Substituting $\tilde{\sH}$, $\tilde{\sZ}$ and $\tilde{\sQ}$ into the brackets
\begin{equation}
	\begin{split}
		[\tilde{\sH}, \tilde{\sQ}(\theta)] &= \tilde{\sQ}(\theta\tilde{\hh}) \\
		[\tilde{\sZ}, \tilde{\sQ}(\theta)] &= \tilde{\sQ}(\theta\tilde{\zz}) \\
	\end{split} \qquad 
		[\tilde{\sQ}(\theta), \tilde{\sQ}(\theta)] = \tilde{n_0} |\theta|^2 \tilde{\sH} + \tilde{n_1} |\theta|^2 \tilde{\sZ} - \tilde{\sJ}(\theta\tilde{\nn_2}\bar{\theta}) - \tilde{\sB}(\theta\tilde{\nn_3}\bar{\theta}) - \tilde{\sP}(\theta\tilde{\nn_4}\bar{\theta}),
\end{equation}
we find 
\begin{equation}
	\begin{split}
		\tilde{\hh} &= \qq (a \hh + b \zz) \qq^{-1} \\
		\tilde{\zz} &= \qq (c \hh + d \zz) \qq^{-1}
	\end{split} \quad 
	\begin{split}
		\tilde{n_0} &= \frac{|\qq|^2}{ad -bc} (d n_0 - c n_1) \\
		\tilde{n_1} &= \frac{|\qq|^2}{ad- bc} (a n_1 - b n_0).
	\end{split}
\end{equation}
These amendments mean that the transformation in $\G$ produce the following basis changes
\begin{equation}
\begin{split}
\sJ &\mapsto \sJ\circ \Ad_{\uu} \\
\sB &\mapsto e \sB\circ \Ad_{\uu} + f \sP\circ \Ad_{\uu} \\
\sP &\mapsto h \sB\circ \Ad_{\uu} + i \sP\circ \Ad_{\uu} \\
\sH &\mapsto a \sH + b \sZ \\
\sZ &\mapsto c \sH + d \sZ \\
\sQ &\mapsto \sQ\circ \Ad_{\uu} \circ R_\qq.
\end{split}
\end{equation}
These transformations may be summarised by $(A = \big( \begin{smallmatrix} a & b \\ c & d \end{smallmatrix} \big),
C = \big(\begin{smallmatrix} e & f \\ h & i \end{smallmatrix}\big), \qq, \uu) \in \GL(\mathbb{R}^2) \times \GL(\mathbb{R}^2)
\times \mathbb{H}^\times \times \mathbb{H}^\times$.
\subsection{Classification} \label{subsec:N1_class}
The calculations for classifying the super-extensions of $\hat{\n}_{\pm}$ and $\hat{\g}$ all follow identically.  It
will, therefore, only be stated once below.  However, the central extension of the static 
kinematical Lie algbera is a little different, so will be treated first.
\subsubsection{$\hat{\a}$}
Using the preliminary results from Lemma \ref{lem:N1_001} in section \ref{par:N1_S_P_001}, we find
$\bb=\pp=\zz=0$.  Substituting these quaternions into the $[\B, \Q, \Q]$ and $[\P, \Q, \Q]$
identities with the relevant brackets, we get $\nn_2 = 0$, $\nn_3=0$ and  $\nn_4 = 0$ .  
Then, wanting $[\Q, \Q] \neq 0$, the $[\bH, \Q, \Q]$ conditions tells us that $\Re(\hh)=0$.  Finally,
Lemma \ref{lem:N1_111} reduces to $n_0 \hh = 0$.  Therefore, we have two possible cases: one in 
which $n_0 = 0$ and $\hh \in \Im(\mathbb{H})$ and another in which $\hh = 0$ and
$n_0$ is unconstrained.  In the former case, the subgroup $\G \subset \GL(\s_{\bar{0}}) \times \GL(\s_{\bar{1}})$ can be used to 
set $\hh = \kk$ and $n_1=1$, such that the only non-vanishing brackets 
involving $\sQ$ are
\begin{equation}
	[\sH, \sQ(\theta)] = \sQ(\theta\kk) \quad \text{and} \quad [\sQ(\theta), \sQ(\theta)] = |\theta|^2 \sZ.
\end{equation}
Notice, however, that this case also allows for $\hh = 0$, leaving only
\begin{equation}
	[\sQ(\theta), \sQ(\theta)] = |\theta|^2 \sZ.
\end{equation}
In the latter case, we can use $\G$ to scale $n_0$ and 
$n_1$, so the non-vanishing brackets are
\begin{equation}
	[\sQ(\theta), \sQ(\theta)] = |\theta|^2 \sH + |\theta|^2 \sZ.
\end{equation}
\subsubsection{$\hat{\n}_{\pm}$ and $\hat{\g}$}
Using the preliminary results of Lemmas \ref{lem:N1_001} and \ref{lem:N1_111}, we
 instantly find $\bb = \pp = \zz = 0$, and, subsequently, $n_0 \hh = \tfrac12 \nn_2$.
The super-Jacobi identity $[\B, \Q, \Q]$ then tells us that $n_0 = \nn_2 = \nn_4 = 0$ and the identity
$[\P, \Q, \Q]$ gives us $\nn_3 = 0$.  Thus, the $(\s_{\bar{1}}, \s_{\bar{1}}, \s_{\bar{1}})$ 
condition is trivially satisfied.  The only remaining condition is from $[\bH, \Q, \Q]$, which
tells us $n_1\Re(\hh)=0$.  Since we want $[\Q, \Q] \neq 0$, we must have $n_1 \neq 0$,
therefore, $\hh \in \Im(\mathbb{H})$.  Using $\G$ to set
$\hh = \kk$ and $n_1 = 1$, we have non-vanishing brackets
\begin{equation}
	[\sH, \sQ(\theta)] = \sQ(\theta\kk) \quad \text{and} \quad [\sQ(\theta), \sQ(s)] = |\theta|^2 \sZ.
\end{equation}
Similar to the $\hat{\a}$ case, the restriction $\hh \in \Im(\mathbb{H})$ does not remove the choice
$\hh = 0$.  Therefore, we may also have
\begin{equation}
	[\sQ(\theta), \sQ(s)] = |\theta|^2 \sZ
\end{equation}
as the only non-vanishing bracket.
\subsection{Summary} \label{subsec:N1_summary}
Table \ref{tab:N1_classification} lists all the $\N=1$ generalised Bargmann superalgebras we have classified.  Each Lie superalgebra is an $\N=1$ super-extension of one of the generalised Bargmann algebras given in Table \ref{tab:algebras}, taken from \cite{Figueroa-OFarrill:2017ycu}.  It is interesting to compare this list of $\N=1$ super-extensions of centrally-extended kinematical Lie algebras to the list of centrally-extended $\N=1$ kinematical Lie superalgebras given in Table \ref{tab:central-ext}.  This table is a reduced and adapted version of one given in \cite{Figueroa-OFarrill:2019ucc}, where we have only kept those extensions built upon the static, Newton-Hooke, and Galilean algebras.
\\ \\
Notice that only one of the generalised Bargmann superalgebras is present in the classification of centrally-extended kinematical Lie superalgebras, \hyperlink{S2}{S2}.  Although it does not match exactly, we can use the basis transformations in $\G \subset \GL(\s_{\bar{0}}) \times \GL(\s_{\bar{1}})$ to bring it into the same form as the second superalgebra in Table \ref{tab:central-ext}.  This result shows us that, in general, super-extending and centrally-extending kinematical Lie algebras are not commutative processes. 
\begin{table}[h!]
  \centering
  \caption{$\N=1$ Generalised Bargmann superalgebras (with $[\Q,\Q]\neq 0$)}
  \label{tab:N1_classification}
  \setlength{\extrarowheight}{2pt}
  \rowcolors{2}{blue!10}{white}
    \begin{tabular}{l|l*{5}{|>{$}c<{$}}}\toprule
      \multicolumn{1}{c|}{S} & \multicolumn{1}{c|}{$\k$} & \multicolumn{1}{c|}{$\hh$}& \multicolumn{1}{c|}{$\zz$}& \multicolumn{1}{c|}{$\bb$} & \multicolumn{1}{c|}{$\pp$} & \multicolumn{1}{c}{$[\sQ(\theta),\sQ(\theta)]$}\\
      \toprule
      1 & \hyperlink{a}{$\hat{\a}$} & & & & & |\theta|^2 \sZ  \\
      2 & \hyperlink{a}{$\hat{\a}$} & \kk & & & & |\theta|^2 \sZ  \\ 
      \hypertarget{S3}{3} & \hyperlink{a}{$\hat{\a}$} & & &&  & |\theta|^2 \sH + |\theta|^2 \sZ \\
      4 & \hyperlink{n-}{$\hat{\n}_-$} & & & & & |\theta|^2 \sZ  \\
      5 & \hyperlink{n-}{$\hat{\n}_-$} & \kk & & & & |\theta|^2 \sZ  \\
      6 & \hyperlink{n+}{$\hat{\n}_+$} &  & & & & |\theta|^2 \sZ  \\      
      7 & \hyperlink{n+}{$\hat{\n}_+$} & \kk & & & & |\theta|^2 \sZ  \\
      8 & \hyperlink{g}{$\hat{\g}$} & & & & & |\theta|^2 \sZ  \\
      9 & \hyperlink{g}{$\hat{\g}$} & \kk & & & & |\theta|^2 \sZ  \\
      \bottomrule
    \end{tabular}
    \caption*{The first column gives each generalised Bargmann superalgebra $\s$ a unique identifier,
    and the second column tells us the underlying generalised Bargmann algebra $\k$.
    The next four columns tells us how the $\s_{\bar{0}}$ generators $\bH, \bZ, \B$, and $\P$ act on $\Q$.
    Recall, the action of $\J$ is fixed, so we do not need to state this explicitly.  The final column then specifies the $[\Q, \Q]$ bracket. 
    }
\end{table}
\begin{table}[h!]
  \centering
  \caption{Central extensions of $\N=1$ kinematical Lie superalgebras}
  \label{tab:central-ext}
  \setlength{\extrarowheight}{2pt}
  \rowcolors{2}{blue!10}{white}
  \begin{tabular}{l*{6}{|>{$}c<{$}}}\toprule
    \multicolumn{1}{c|}{$\k$} & \multicolumn{1}{c|}{$[\sB(\beta),\sP(\pi)]$} & \multicolumn{1}{c|}{$\hh$} & \multicolumn{1}{c|}{$\zz$} & \multicolumn{1}{c|}{$\bb$}  & \multicolumn{1}{c|}{$\pp$} & \multicolumn{1}{c}{$[\sQ(\theta),\sQ(\theta)]$} \\
    \toprule
    $\a$ & & \tfrac12 \kk & & & & |\theta|^2 \sZ - \sP(\theta\kk\bar{\theta}) \\
    $\a$ & \Re(\bar{\beta}\pi) \sZ & & & & & |\theta|^2 \sH \\
    $\a$ & & & & & & |\theta|^2 \sZ - \sB(\theta\jj\bar{\theta}) - \sP(\theta\kk\bar{\theta}) \\
    $\a$ & & & & & & |\theta|^2 \sZ -\sP(\theta\kk\bar{\theta}) \\
    $\g$ & & \kk & & & & |\theta|^2 \sZ -\sP(\theta\kk\bar{\theta}) \\
    $\g$ & & & & & & |\theta|^2 \sZ -\sP(\theta\kk\bar{\theta}) \\
    $\n_+$ & & \tfrac12 \jj & & & & |\theta|^2 \sZ - \sB(\theta\ii\bar{\theta}) - \sP(\theta\kk\bar{\theta}) \\
    \bottomrule
  \end{tabular}
  \caption*{The first column identifies the kinematical Lie algebra $\k$ underlying the extensions.
  The second column indicates whether the central extension has been introduced in the $[\B, \P]$
  bracket.  The next four columns show the $[s_{\bar{0}}, \s_{\bar{1}}]$ brackets for the KLSA.  As 
  we can see, the only non-vanishing case is $[\sH, \sQ(\theta)] = \sQ(\theta\hh)$, where $\theta \in \mathbb{H}$ and
  $\hh \in \Im(\mathbb{H})$.  The final column then tells us whether the central extension enters
  the $[\Q, \Q]$ bracket.}
\end{table}
\subsubsection{Unpacking the Notation}
Although the quaternionic formulation of these superalgebras is convenient for our purposes, it is perhaps unfamiliar to the reader. Therefore, in this section, we will convert the $\N=1$ super-extension for the Bargmann algebra into a more conventional format. The brackets for this algebra, excluding the $\s_{\bar{0}}$ brackets which have already been discussed in section \ref{subsec:FI_KLSAs}, take the form 
\begin{equation}
	[\sH, \sQ(\theta)] = \sQ(\theta\kk) \quad \text{and} \quad [\sQ(\theta), \sQ(\theta)] = |\theta|^2 \sZ.
\end{equation}
Letting $\sQ(\theta) = \sum_{a=1}^{4} \theta_a \bQ_a$, where $\theta = \theta_4 + \theta_1 \ii + \theta_2 \jj + \theta_3 \kk$, we can rewrite these brackets as
\begin{equation}
	[\bH, \bQ_a] = \sum_{a = 1}^{4} \bQ_b \tensor{\Sigma}{^b_a} \quad \text{and} \quad  [\bQ_a, \bQ_b] = \delta_{ab} \bZ,
\end{equation}
where, for $\sigma_2$ being the second Pauli matrix, 
\begin{equation}
	\Sigma = \begin{pmatrix}
		 0 & i \sigma_2 \\ -i \sigma_2 & 0
	\end{pmatrix}.
\end{equation}
\section{$\N=2$ Extensions of the Generalised Bargmann Algebras} \label{sec:N2_classification}
Having established the introduction of the central extension $\bZ$ into our classification problem in section \ref{sec:N1_classification}, we now look to introduce an additional $\so(3)$ spinor module.  Section \ref{subsec:N2_setup} will describe the setup up for the $\N=2$ generalised Bargmann superalgebras, before extending the preliminary results from \cite{Figueroa-OFarrill:2019ucc} to this case.  It is also in this section that the group of basis transformations $\G$ will be adapted for extended supersymmetry.  The number of additional parameters in this case means that there are several branches of super-extension for each generalised Bargmann algebra.  In section \ref{subsec:N2_branches}, we use the preliminary results from the $(\s_{\bar{0}}, \s_{\bar{0}}, \s_{\bar{1}})$ super-Jacobi identities to identify four possible branches.  Each branch is then explored in detail in section \ref{subsec:N2_class} where we identify the non-empty sub-branches, which are summarised in section \ref{subsec:N2_summary}.
\subsection{Setup for the $\N=2$ Calculation} \label{subsec:N2_setup}
For completeness, recall that, in addition to the standard kinematical Lie brackets, the brackets for the universal generalised Bargmann superalgebra are
\begin{equation}
	\begin{split}
		[\sB(\beta), \sP(\pi)] &= \Re(\bar{\beta}\pi) \sZ \\ 
		[\sH, \sB(\beta)] &= \lambda \sB(\beta) + \mu \sP(\beta) \\
		[\sH, \sP(\pi)] &= \eta \sB(\pi) + \varepsilon \sP(\pi),
	\end{split} 
\end{equation}
where $\beta, \pi \in \Im(\mathbb{H})$ and $\lambda, \mu, \eta, \varepsilon \in \mathbb{R}$.  Because we now have two spinor modules, the brackets involving
$\s_{\bar{1}}$ need to be adapted from those given in \cite{Figueroa-OFarrill:2019ucc}.  We will continue to use the map $\sQ$ for the odd dimensions; however, it now acts on $\vt$, a vector in $\mathbb{H}^2$.  We will choose $\mathbb{H}^2$ to
be a left quaternionic vector space such that $\mathbb{H}$ acts linearly from the 
left and all $2 \times 2$ $\mathbb{H}$ matrices act on the right.  Therefore,
writing $\vt$ out in its components, we have
\begin{equation}
\vt= \begin{pmatrix}
\theta_1 & \theta_2
\end{pmatrix},
\end{equation}
where $\theta_1, \theta_2 \in \mathbb{H}$.
The $[\s_{\bar{0}}, \s_{\bar{1}}]$ brackets are again the $\so(3)$-equivariant 
endomorphisms of $\s_{\bar{1}}$.  Since we choose $\so(3)$ to act via
left quaternionic multiplication, the commuting endomorphisms are all
those that may act on the right.  In the present case, these are elements 
of $\Mat_{2}(\mathbb{H})$.  Thus the brackets containing the $\so(3)$ 
scalars are 
\begin{equation} \label{eq:N2_S_general_scalar_brackets}
\begin{split}
[\sH, \sQ(\vt)] &= \sQ(\vt\mh) \\
[\sZ, \sQ(\vt)] &= \sQ(\vt\mz),
\end{split}
\end{equation}
where $\mh, \mz \in \Mat_2(\mathbb{H})$.
In \cite{Figueroa-OFarrill:2019ucc}, $[\J, \Q] $, $[\B, \Q]$ and $[\P, \Q]$ were described by Clifford multiplication $V \otimes S \rightarrow S$, which is given by left quaternionic multiplication by $\Im(\mathbb{H})$.  The space of such maps was isomorphic to the space of $\r$-equivariant maps of $S$, which is a copy of the quaternions.  Now with $\s_{\bar{1}} = S \oplus S$, we have four possible endomorphisms of this type.  Labelling the two spinor modules $S_1$ and $S_2$, we may use the Clifford action to map $S_1$ to $S_1$, $S_1$ to $S_2$, $S_2$ to $S_1$, or $S_2$ to $S_2$.  All of these maps may be summarised as follows
\begin{equation}  \label{eq:N2_S_general_vector_brackets}
\begin{split}
[\sJ(\omega), \sQ(\vt)] &= \tfrac12 \sQ(\omega \vt) \\
[\sB(\beta), \sQ(\vt)] &= \sQ(\beta \vt \mb) \\
[\sP(\pi), \sQ(\vt)] &= \sQ(\pi \vt \mp).
\end{split}
\end{equation}
Here, $\omega, \beta,\pi \in \Im(\mathbb{H})$ and $\mb, \mp \in \Mat_2(\mathbb{H})$.  
Finally, consider the $[\Q, \Q]$ bracket.  This will consist of the $\so(3)$-equivariant
$\mathbb{R}$-linear maps $\bigodot^2 \s_{\bar{1}} \rightarrow \s_{\bar{0}}$.  To write down these
maps, we make use of the $\so(3)$-invariant inner product on $\s_{\bar{1}}$
\begin{equation}
\langle \vt, \vt' \rangle = \Re(\vt \vt^\dagger) \quad \text{where} \quad \vt, \vt' \in \mathbb{H}^2 \quad 
\vt^\dagger = \bar{\vt}^T.
\end{equation}
This bracket's $\so(3)$-invariance is clear on considering left multiplication by
$\uu \in \sp(1)$ and noting $\sp(1) \cong \so(3)$.  We can now use this bracket 
to identify $\bigodot^2 \s_{\bar{1}}$ with the symmetric $\mathbb{R}$-linear endomorphisms of
$\s_{\bar{1}} \cong S^2 \cong \mathbb{H}^2$, i.e.\ the maps $\mu: \mathbb{H}^2 \rightarrow \mathbb{H}^2$ such
that $\langle \mu(\vt), \vt' \rangle = \langle \vt, \mu(\vt') \rangle$.  A general 
$\mathbb{R}$-linear map of $\mathbb{H}^2$ may be written
\begin{equation}
\mu(\vt) = \qq \vt M \quad \text{where} \quad \qq \in \mathbb{H} \; \;  \text{and} \; \;  M \in \Mat_2(\mathbb{H}).
\end{equation}
Now, inserting this definition into the condition for a symmetric endomorphism, we 
obtain the following two cases: 
\begin{enumerate}
	\item $\qq \in \mathbb{R}$ and $M = M^\dagger$
	\item $\qq \in \Im(\mathbb{H})$ and $M = -M^\dagger$.
\end{enumerate}
The first instance gives us our $\so(3)$ scalar modules in $\bigodot^2 S^2$; therefore, 
these will map to either $\bH$ or $\bZ$ in $\s_{\bar{0}}$ to ensure we have
$\so(3)$-equivariance.  The condition on $M$ states that it must be of the form
\begin{equation}
M = \begin{pmatrix}
a & b + \mm \\ b - \mm & c
\end{pmatrix} = 
a \begin{pmatrix}
1 & 0 \\ 0 & 0
\end{pmatrix} + 
b \begin{pmatrix}
 0 & 1 \\ 1 & 0
\end{pmatrix} + 
c \begin{pmatrix}
 0 & 0 \\ 0 & 1
\end{pmatrix} + 
\mm \begin{pmatrix}
0 & 1 \\ -1 &  0
\end{pmatrix},
\end{equation}
where $a, b, c \in \mathbb{R}$ and $\mm \in \Im(\mathbb{H})$.  We can make sense of
this result using the decomposition of the odd part of the superalgebra, 
$\s_{\bar{1}} \cong S^2 \cong S \otimes \mathbb{R}^2$, and our knowledge of
the maps in  $\Hom(\bigodot^2 S, \s_{\bar{0}})$ derived from \cite{Figueroa-OFarrill:2019ucc}.  
Symmetrising the decomposed $\s_{\bar{1}}$, we get $\bigodot^2 S^2 \cong \bigodot^2 S \otimes \bigodot^2 \mathbb{R}^2 \oplus \bigwedge^2 S \otimes \bigwedge^2 \mathbb{R}^2$.  Notice that the coefficients
for the $\bigodot^2 \mathbb{R}^2$ basis elements result in multiplying $\vt \in S^2$ by
$\mathbb{R}$ on both the left and right.  Thus we find three copies of the scalar map in $\bigodot^2 S 
\rightarrow \mathbb{R} \oplus 3V$, one for each of the $\bigodot^2 \mathbb{R}$ basis.  
Similarly, the coefficient for the $\bigwedge^2 \mathbb{R}^2$
basis element produces the three scalar maps in $\bigwedge^2 S \rightarrow
3 \mathbb{R} \oplus V$.  Thus the above matrix accounts for the $\so(3)$ scalar maps in both
the symmetric and anti-symmetric components of the decomposition.  
\\ \\
Now, the second case gives us our $\so(3)$ vector modules in 
$\bigodot^2 S^2$; therefore, these will map to $\B$, $\P$, or $\J$ to ensure 
$\so(3)$-equivariance.  The condition on $M$ in this case produces
\begin{equation}
M = \begin{pmatrix}
\nn & d + \ll \\ -d + \ll & \rr
\end{pmatrix} = 
\nn \begin{pmatrix}
1 & 0 \\ 0 & 0
\end{pmatrix} + 
\ll \begin{pmatrix}
 0 & 1 \\ 1 & 0
\end{pmatrix} + 
\rr \begin{pmatrix}
 0 & 0 \\ 0 & 1
\end{pmatrix} + 
d \begin{pmatrix}
0 & 1 \\ -1 &  0
\end{pmatrix},
\end{equation}
where $\nn, \ll, \rr \in \Im(\mathbb{H})$ and $d \in \mathbb{R}$.  Again, we can understand
this result through the decomposition of $\s_{\bar{1}}$.  From \cite{Figueroa-OFarrill:2019ucc},
we know that the $\so(3)$ vectors in $\Hom(\bigodot^2 S, \s_{\bar{0}})$ come from 
simultaneous left and right quaternionic multiplication by $\Im(\mathbb{H})$.  These
are precisely the maps we find in the coefficients for the $\bigodot^2 \mathbb{R}^2$ basis, 
as expected through the decomposition of $\bigodot^2 S^2$.  Finally, we recover
the $\so(3)$ vector in $\bigwedge^2 S \rightarrow 3 \mathbb{R} \oplus V$ as 
the coefficient to the $\bigwedge^2 \mathbb{R}^2$ basis.
\\ \\
Putting all this together, we may write
\begin{equation}
[\sQ(\vt), \sQ(\vt)] = \langle \vt, \vt N_0 \rangle \sH + \langle \vt, \vt N_1 \rangle \sZ + \langle \vt, \JJ \vt N_2 \rangle
+ \langle \vt, \BB \vt N_3 \rangle + \langle \vt, \PP \vt N_4 \rangle,
\end{equation}
where $N_0$, $N_1$ are quaternion Hermitian, and $N_2$, $N_3$, $N_4$ are quaternion
skew-Hermitian, as stated above, and 
\begin{equation}
\JJ = \bJ_1 \ii + \bJ_2 \jj + \bJ_3 \kk \quad \BB = \bB_1 \ii + \bB_2 \jj + \bB_3 \kk \quad 
\PP = \bP_1 \ii + \bP_2 \jj + \bP_3 \kk.
\end{equation}
Using the fact that $\Re(\bar{\omega}\JJ) = \sJ(\omega)$ and $N_i = N_i^\dagger$ for 
$i \in \{0, 1\}$, we can write
\begin{equation}
[\sQ(\vt), \sQ(\vt)] = \Re(\vt N_0 \vt^\dagger) \sH + \Re(\vt N_1 \vt^\dagger) \sZ - \sJ(\vt N_2 \vt^\dagger) 
- \sB(\vt N_3 \vt^\dagger) - \sP(\vt N_4 \vt^\dagger).
\end{equation}
This polarises to
\begin{equation}
\begin{split}
[\sQ(\vt), \sQ(\vt')] = \frac{1}{2} \big( & \Re(\vt N_0 \vt'^\dagger + \vt' N_0 \vt^\dagger) \sH 
+ \Re(\vt N_1 \vt'^\dagger + \vt' N_1 \vt^\dagger) \sZ \\ & - \sJ(\vt N_2 \vt'^\dagger + \vt' N_2 \vt^\dagger) 
- \sB(\vt N_3 \vt'^\dagger + \vt' N_3 \vt^\dagger) - \sP(\vt N_4 \vt'^\dagger + \vt' N_4 \vt^\dagger) \big) .
\end{split}
\end{equation}

\subsubsection{Preliminary Results} \label{subsubsec:N2_S_preliminaries}
As in the $\N=1$ case, we can form a number of universal results that will help us when 
investigating the super-extensions of the generalised Bargmann algebras.  The following 
subsections will cover the $(\s_{\bar{0}}, \s_{\bar{0}}, \s_{\bar{1}})$, 
$(\s_{\bar{0}}, \s_{\bar{1}}, \s_{\bar{1}})$, and $(\s_{\bar{1}}, \s_{\bar{1}}, \s_{\bar{1}})$
identities, respectively. \\
\paragraph{$(\s_{\bar{0}}, \s_{\bar{0}}, \s_{\bar{1}})$} \label{par:N2_S_P_011} ~\\ \\
In the $\N=1$ case, $\hh, \zz, \bb, \pp \in \mathbb{H}$,
and in the $\N=2$ case $\mh, \mz, \mb, \mp \in \Mat_2(\mathbb{H})$.  Since
$\mathbb{H}$ and $\Mat_2(\mathbb{H})$ are both associative, non-commutative
algebras, the algebraic manipulations are the same in both cases.  Therefore, the 
$\N=1$ results generalise to the $\N=2$ case; the only difference being that the
variables are $2 \times 2$ $\mathbb{H}$ matrices rather than $\mathbb{H}$ elements. 
\begin{lemma}\label{lem:N2_001}
  The following relations between $\mh, \mz, \mb, \mp \in \Mat_2(\mathbb{H})$ are implied by the
  corresponding $\k$-brackets:
  \begin{equation}
    \begin{split}
    	  [\bH,\bZ] = \lambda \bH + \mu \bZ & \implies [\mz, \mh] = \lambda \mh + \mu \mz \\
      [\bH,\B] = \lambda \B + \mu \P & \implies [\mb,\mh] = \lambda \mb + \mu \mp\\
      [\bH,\P] = \lambda \B + \mu \P & \implies [\mp,\mh] = \lambda \mb + \mu \mp\\
      [\bZ,\B] = \lambda \B + \mu \P & \implies [\mb,\mz] = \lambda \mb + \mu \mp\\
      [\bZ,\P] = \lambda \B + \mu \P & \implies [\mp,\mz] = \lambda \mb + \mu \mp\\
      [\B,\B] = \lambda \B + \mu \P + \nu \J & \implies \mb^2 = \tfrac12 \lambda \mb + \tfrac12 \mu \mp + \tfrac14 \nu\\
      [\P,\P] = \lambda \B + \mu \P + \nu \J & \implies \mp^2 = \tfrac12 \lambda \mb + \tfrac12 \mu \mp + \tfrac14 \nu\\
      [\B,\P] = \lambda \bH + \mu \bZ & \implies \mb \mp + \mp\mb = 0\quad\text{and}\quad [\mb,\mp] = \lambda \mh  + \mu \mz.
    \end{split}
  \end{equation}
\end{lemma}
\begin{proof}
See the proof of Lemma \ref{lem:N1_001} for the algebraic manipulations that produce the
above results.
\end{proof}
\paragraph{$(\s_{\bar{0}}, \s_{\bar{1}}, \s_{\bar{1}})$} \label{par:N2_S_P_011} ~\\ \\
As in the $\N=1$ case, we use the universal generalised Bargmann algebra to simplify
our analysis here.  Recall the brackets for this algebra are
\begin{equation}
	[\B, \P] = \bZ \qquad [\bH, \B] = \lambda \B + \mu \P \qquad [\bH, \P] = \eta \B + \varepsilon \P,
\end{equation}
where $\lambda, \mu, \eta, \varepsilon \in \mathbb{R}$.  Using these brackets, we obtain the 
following result.
\begin{lemma} \label{lem:N2_011}
\begin{equation}
	\begin{split}
		\text{The $[\bH, \Q, \Q]$ identity produces the conditions} \\
		0 &= \mh N_i + N_i \mh^\dagger \quad \text{where} \quad i \in \{0, 1, 2\} \\
		\lambda N_3 + \eta N_4 &= \mh N_3 + N_3 \mh^\dagger \\
		\mu N_3 + \varepsilon N_4 &= \mh N_4 + N_4 \mh^\dagger. \\
	\text{The $[\bZ, \Q, \Q]$ identity produces the conditions} \\
	0 &= \mz N_i + N_i \mz^\dagger \quad \text{where} \quad i \in \{0, 1, 2, 3, 4\}. \\
	\text{The $[\B, \Q, \Q]$ identity produces the conditions} \\
		0 &= \mb N_0 - N_0 \mb^\dagger \\
		N_4 &= \mb N_1 - N_1 \mb^\dagger \\
		0 &= \beta \vt \mb N_2 \vt^\dagger + \vt N_2 (\beta \vt \mb)^\dagger \\
		\lambda \Re(\vt N_0 \vt^\dagger) \beta + \tfrac12 [\beta, \vt N_2 \vt^\dagger] &= \beta \vt \mb N_3 \vt^\dagger
		+ \vt N_3 (\beta \vt \mb)^\dagger \\
		\mu \Re(\vt N_0\vt^\dagger) \beta &= \beta \vt \mb N_4 \vt^\dagger + \vt N_4 (\beta \vt \mb)^\dagger. \\
	\text{The $[\P, \Q, \Q]$ identity produces the conditions} \\
		0 &= \mp N_0 - N_0 \mp^\dagger \\
		-N_3 &= \mp N_1 - N_1 \mp^\dagger \\
		0 &= \pi \vt \mp N_2 \vt^\dagger + \vt N_2 (\pi \vt \mp)^\dagger \\
		\eta \Re(\vt N_0\vt^\dagger) \pi &= \pi \vt \mp N_3 \vt^\dagger + \vt N_3 (\pi \vt \mp)^\dagger \\
		\varepsilon \Re(\vt N_0 \vt^\dagger) \pi + \tfrac12 [\pi, \vt N_2 \vt^\dagger] &= \pi \vt \mp N_4 \vt^\dagger
		+ \vt N_4 (\pi \vt \mp)^\dagger , \\
	\end{split}		
\end{equation}
where $\beta, \pi \in \Im(\mathbb{H})$ and $\vt \in \mathbb{H}^2$.
\end{lemma}
\begin{proof}
Beginning with the $[\bH, \Q, \Q]$ identity, we have
\begin{equation}
	[\sH, [\sQ(s), \sQ(s)]] = 2 [[\sH, \sQ(s)], \sQ(s)].
\end{equation}
Focussing on the L.H.S., note the general form of the $[\Q, \Q]$ bracket
has components along each of the $\s_{\bar{0}}$ basis elements; however,
$\bH$ only commutes with $\B$ and $\P$.  Therefore, 
\begin{equation}
	\begin{split}
		L.H.S. &= -[\sH, \sB(\vt N_3\vt^\dagger)] - [\sH, \sP(\vt N_4\vt^\dagger)] \\
		&= - \sB( \vt (\lambda N_3 + \eta N_4) \vt^\dagger) - \sP( \vt (\mu N_3 + \varepsilon N_4) \vt^\dagger ).
	\end{split}
\end{equation}
Substituting $[\sH, \sQ(\vt)] = \sQ(\vt\mh)$ into the R.H.S. and using the polarised form of the $[\Q, \Q]$ bracket,
we find
\begin{equation}
	\begin{split}
		R.H.S. = & \Re( \vt (\mh N_0  + N_0 \mh^\dagger) \vt^\dagger) \sH + \Re( \vt (\mh N_1  + N_1 \mh^\dagger) \vt^\dagger) \sZ \\
		& - \sJ( \vt (\mh N_2 + N_2 \mh^\dagger) \vt^\dagger ) - \sB( \vt (\mh N_3 + N_3 \mh^\dagger) \vt^\dagger )
		- \sP( \vt (\mh N_4 + N_4 \mh^\dagger) \vt^\dagger ).
	\end{split}
\end{equation}
Comparing coefficients and using the injectivity and linearity of the maps $\sJ, \sB$ and $\sP$, we get the
desired conditions.  The $[\bZ, \Q, \Q]$ result follows in an analogous manner.  Consider the 
$[\B, \Q, \Q]$ Jacobi identity
\begin{equation}
	[ \sB(\beta), [\sQ(\vt), \sQ(\vt)]] = 2 [[\sB(\beta), \sQ(\vt)], \sQ(\vt)].
\end{equation}
Since $\B$ commutes with $\bZ$ and $\B$, the L.H.S. takes the following form
\begin{equation}
	\begin{split}
		L.H.S. &= [\sB(\beta), \Re(\vt N_0\vt^\dagger) \sH - \sJ(\vt N_2\vt^\dagger) - \sP(\vt N_4 \vt^\dagger)] \\
		&= - \Re(\vt N_0 \vt^\dagger) (\lambda \sB(\beta) + \mu \sP(\beta) ) +\tfrac12 \sB([\vt N_2\vt^\dagger, \beta])
		- \Re(\bar{\beta} \vt N_4 \vt^\dagger) \sZ.
	\end{split}
\end{equation}
Turning attention to the R.H.S., we find
\begin{equation}
	\begin{split}
		R.H.S. = & \Re( \beta \vt \mb N_0 \vt^\dagger + \vt N_0 \mb^\dagger \vt^\dagger \bar{\beta}) \sH
		+ \Re( \beta \vt \mb N_1 \vt^\dagger + \vt N_1 \mb^\dagger \vt^\dagger \bar{\beta}) \sZ \\
		& - \sJ(\beta \vt \mb N_2 \vt^\dagger + \vt N_2 \mb^\dagger \vt^\dagger \bar{\beta}) 
		- \sB(\beta \vt \mb N_3 \vt^\dagger + \vt N_3 \mb^\dagger \vt^\dagger \bar{\beta})
		- \sP(\beta \vt \mb N_4 \vt^\dagger + \vt N_4 \mb^\dagger \vt^\dagger \bar{\beta}).
	\end{split}
\end{equation}
Using the property $\bar{\beta} = -\beta$, since $\beta \in \Im(\mathbb{H})$, and the cyclic
property of $\Re$, the first two terms can have their coefficients written in the form
\begin{equation}
	\Re( \beta \vt \mb N_i \vt^\dagger + \vt N_i \mb^\dagger \vt^\dagger \bar{\beta})
	= \Re( \beta \vt (\mb N_i - N_i \mb^\dagger ) \vt^\dagger ),
\end{equation} 
for $i \in \{0, 1\}$.  Again, comparing coefficients we obtain the desired results.  The 
$[\P, \Q, \Q]$ case follows identically.
\end{proof}
\paragraph{$(\s_{\bar{1}}, \s_{\bar{1}}, \s_{\bar{1}})$} \label{par:N2_S_P_111} ~\\ \\
The last super-Jacobi identity to consider is the $(\s_{\bar{1}}, \s_{\bar{1}}, \s_{\bar{1}})$ case,
$[\Q, \Q, \Q]$.
\begin{lemma}\label{lem:N2_111}
The $[\Q,\Q,\Q]$ identity produces the condition
\begin{equation}
\Re(\vt N_0 \vt^\dagger) \vt \mh + \Re(\vt N_1 \vt^\dagger) \vt \mz = \tfrac12 \vt N_2 \vt^\dagger \vt + \vt N_3 \vt^\dagger \vt \mb 
+ \vt N_4 \vt^\dagger \vt \mp.
\end{equation}
\end{lemma}
\begin{proof}
The $[\Q, \Q, \Q]$ identity is written
\begin{equation}
	0 = [[\sQ(\vt), \sQ(\vt)], \sQ(\vt)].
\end{equation}
Substituting in the $[\Q, \Q]$ bracket, this becomes
\begin{equation}
	0 = [ \Re(\vt N_0 \vt^\dagger) \sH + \Re(\vt N_1 \vt^\dagger) \sZ - \sJ(\vt N_2 \vt^\dagger) 
- \sB(\vt N_3 \vt^\dagger) - \sP(\vt N_4 \vt^\dagger), \sQ(\vt)]. 
\end{equation}
Finally, using the brackets in \eqref{eq:N2_S_general_scalar_brackets} and \eqref{eq:N2_S_general_vector_brackets} and the injectivity of $\sQ$, we obtain the desired result.
\end{proof}
\subsubsection{Basis Transformations} \label{subsubsec:N2_S_autos}
We will investigate the subgroup $\G \subset \GL(\s_{\bar{0}}) \times 
\GL(\s_{\bar{1}})$ by first looking at the transformations induced by the adjoint
action of the rotational subalgebra $\r \cong \so(3)$.  
We will then look at the $\so(3)$-equivariant maps transforming
the basis of the underlying vector space.  These will act via Lie algebra automorphisms in $\s_{\bar{0}}$
and endomorphisms of the $\so(3)$ module $S^2$ in $\s_{\bar{1}}$. 
Note, in the former case, where the automorphism is induced by
$\ad_{\bJ_i}$, each $\so(3)$ module will transform into itself, while,
in the latter case, when the transformation is some $\so(3)$-equivariant map,
the modules transform into one another.  For completeness, at the end
of the section, we determine the automorphisms of each generalised Bargmann algebra.
\\ \\
Recall that $\Sp(1)$ is the double-cover of $\Aut(\mathbb{H})$, and $\Aut(\mathbb{H}) 
\cong \SO(3)$.  We, therefore, write $\lambda \in \Aut(\mathbb{H})$ as $\lambda(\ss) =
\uu\ss\bar{\uu}$ for some $\uu \in \Sp(1)$, which will act trivially on the real component 
of $\ss$ and rotate the imaginary components.  Using this result, we can represent the action
of $\Aut(\mathbb{H})$ on the $\so(3)$ vector modules in $\s_{\bar{0}}$ by pre-composing the linear maps
$\sJ$, $\sB$, and $\sP$ with $\Ad_{\uu}$, for $\uu \in \Sp(1)$.  To preserve the kinematical
brackets in $[\s_{\bar{0}}, \s_{\bar{0}}]$, we must pre-compose with the same $\uu$
for each map. Note, $\so(3)$ acts trivially on $\bH$ and $\bZ$, so
these basis elements will be left invariant under these automorphisms. 
For $\s_{\bar{1}}$, we restrict to the individual copies of $S$ through diagonal
matrices.  To preserve the $[\J, \Q]$ bracket, we must 
pre-compose with the same $u$ as above.  Therefore, we write
$\tilde{\sQ}(\vt) = \sQ(\uu \vt\bar{\uu} \mathbb{1})$.  We can now investigate how these automorphisms
affect our brackets
\begin{equation}
\begin{gathered}
\begin{split}
[\sJ(\omega), \sQ(\vt)] &= \tfrac12 \sQ(\omega \vt) \\
[\sB(\beta), \sQ(\vt)] &= \sQ(\beta \vt \mb) \\
[\sP(\pi), \sQ(\vt)] &= \sQ(\pi \vt \mp)
\end{split} \qquad
\begin{split}
[\sH, \sQ(\vt)] &= \sQ(\vt\mh) \\
[\sZ, \sQ(\vt)] &= \sQ(\vt\mz)
\end{split} \\ \\
[\sQ(\vt), \sQ(\vt)] = \Re(\vt N_0 \vt^\dagger) \sH + \Re(\vt N_1 \vt^\dagger) \sZ - \sJ(\vt N_2 \vt^\dagger) 
- \sB(\vt N_3 \vt^\dagger) - \sP(\vt N_4 \vt^\dagger).
\end{gathered}
\end{equation}
Transforming the basis, we have
\begin{equation} \label{eq:transformed_brackets}
\begin{gathered}
\begin{split}
[\tilde{\sJ}(\omega), \tilde{\sQ}(\vt)] &= \tfrac12 \tilde{\sQ}(\omega \vt) \\
[\tilde{\sB}(\beta), \tilde{\sQ}(\vt)] &= \tilde{\sQ}(\beta \vt \tilde{\mb}) \\
[\tilde{\sP}(\pi), \tilde{\sQ}(\vt)] &= \tilde{\sQ}(\pi \vt \tilde{\mp})
\end{split} \qquad
\begin{split}
[\tilde{\sH}, \tilde{\sQ}(\vt)] &= \tilde{\sQ}(\vt\tilde{\mh}) \\
[\tilde{\sZ}, \tilde{\sQ}(\vt)] &= \tilde{\sQ}(\vt\tilde{\mz})
\end{split} \\ \\
[\tilde{\sQ}(\vt), \tilde{\sQ}(\vt)] = \Re(s\vt\widetilde{N_0} \vt^\dagger) \tilde{\sH} + \Re(\vt \widetilde{N_1} \vt^\dagger) \tilde{\sZ}
 - \tilde{\sJ}(\vt \widetilde{N_2} \vt^\dagger) - \tilde{\sB}(\vt \widetilde{N_3} \vt^\dagger) - \tilde{\sP}(\vt \widetilde{N_4} \vt^\dagger),
 \end{gathered}
\end{equation}
with $\tilde{\sH} = \sH$, $\tilde{\sZ} = \sZ$, $\tilde{\sJ} = \sJ\circ \Ad_{\uu}$, 
$\tilde{\sB} = \sB\circ \Ad_{\uu}$, $\tilde{\sP} = \sP\circ \Ad_{\uu}$, and $\tilde{\sQ} = \sQ\circ \Ad_{\uu}$,
where it is understood that $\Ad_{\uu}$ acts diagonally on the $\s_{\bar{1}}$ basis, $\sQ$.  The transformed
matrices are 
\begin{equation}
\begin{split}
\tilde{\mh} &= D \mh D^{-1}\\\
\tilde{\mz} &= D \mz D^{-1}
\end{split} \quad
\begin{split}
\tilde{\mb} &= D \mb D^{-1} \\
\tilde{\mp} &= D \mp D^{-1}
\end{split} \quad
\widetilde{N_i} = D N_i D^\dagger,
\end{equation}
where $D = \uu \mathbb{1}$ for $\uu \in \Sp(1)$ and $i \in \{0, 1, ..., 4\}$.  Therefore, $D^{-1} = D^\dagger = \bar{\uu}\mathbb{1}$.
These automorphisms simultaneously rotate all quaternions, all the components of the matrices $\mh, \mz, \mb, \mp$ and $N_i$, by the same $\Sp(1)$ element.
\\ \\
Next, we want to consider the $\so(3)$-equivariant linear maps 
which leave the rotational subalgebra invariant: $(\sJ, \sB, \sP, \sH, \sZ, \sQ) \rightarrow
(\sJ, \tilde{\sB}, \tilde{\sP}, \tilde{\sH}, \tilde{\sZ}, \tilde{\sQ})$.  These take the general form
\begin{equation} \label{eq:autos_2}
\begin{split}
\tilde{\sH} &= a \sH + b \sZ \\
\tilde{\sZ} &= c \sH + d \sZ \\
\tilde{\sB}(\beta) &= e \sB(\beta) + f \sP(\beta) + g \sJ(\beta) \\
\tilde{\sP}(\pi) &= h \sB(\pi) + i \sP(\pi) + j \sJ(\pi) \\
\tilde{\sQ}(\vt) &= \sQ(\vt M),
\end{split}
\end{equation}
where $a, ..., j \in \mathbb{R}$ and $M \in \GL(\mathbb{H}^2)$.  Crucially, 
\begin{equation}
	A = \begin{pmatrix}
		a & b \\ c & d
	\end{pmatrix} \in \GL(2, \mathbb{R}) \quad \text{and} \quad
	C = \begin{pmatrix}
		e & f & g \\ h & i & j \\ 0 & 0 & 1
	\end{pmatrix} \in \GL(3, \mathbb{R}),
\end{equation}
act on $(\sH, \sZ)^T$ and $(\sB, \sP, \sJ)^T$, respectively.  Each of the generalised
Bargmann algebra allows different transformations of this type; however, there are some
important general results.  Therefore, we will begin by working through the analysis of
these maps with the universal generalised Bargmann algebra before focussing on each algebra separately.
\\ \\
As in the $\N = 1$ case, the checking of brackets that include $\J$ is really verifying that
the above maps are $\so(3)$-equivariant, so this does not give us any information not already presented.  The first bracket 
we will consider is $[\B, \P] = \bZ$.  Substituting in the maps of \eqref{eq:autos_2}, we find
the following important results:
\begin{equation}
	d = ei - fh, \quad c = 0, \quad \text{and} \quad g = j = 0.
\end{equation}
The vanishing of $c$ tells us that $d \neq 0$ if we are to have $A \in \GL(2, \mathbb{R})$.  Also,
the vanishing of $g$ and $j$ shows that we can reduce $C$ to an element of $\GL(2, \mathbb{R})$,
\begin{equation}
	C = \begin{pmatrix}
		e & f \\ h & i 
	\end{pmatrix},
\end{equation}
acting on  $(\sB, \sP)^T$.  The remaining $[\s_{\bar{0}}, \s_{\bar{0}}]$ brackets are $[\bH, \B]$ and $[\bH, \P]$, which produce
\begin{equation} \label{eq:AB_constraints}
	\begin{split}
		0 &= \lambda e (a-1) + \eta af - \mu h\\
		0 &= \lambda f - \varepsilon af + \mu (i - ea) 
	\end{split} \quad \text{and} \quad
	\begin{split}
		0 &= \eta (e-ai) + \varepsilon h - \lambda ah \\
		0 &= \eta f + \varepsilon i (1-a) - \mu ah,
	\end{split}
\end{equation}
respectively.  Clearly, these conditions are dependent on the exact choice of generalised 
Bargmann algebra, so we will leave these results in this form for now.
\\ \\
Now, since the $[\s_{\bar{0}}, \s_{\bar{1}}]$ and $[\s_{\bar{1}}, \s_{\bar{1}}]$ brackets
are so far independent of the chosen algebra, the following results will hold for all the generalised Bargmann algebras.  
Reusing \eqref{eq:transformed_brackets}, in this instance we find
\begin{equation}
\tilde{\mh} = M (a \mh + b\mz) M^{-1} \quad 
\tilde{\mz} = d M \mz M^{-1} \quad
\tilde{\mb} = M (e\mb + f \mp) M^{-1} \quad 
\tilde{\mp} = M (h \mb + i\mp) M^{-1}  \nonumber
\end{equation} 
\begin{equation} \label{eq:N2_S_basis_transformations}
\begin{split}
\widetilde{N_0} &= \frac{1}{a} M N_0 M^\dagger \\
\widetilde{N_1} &= \frac{1}{ad} M (a N_1 - b N_0) M^\dagger
\end{split} \quad
\begin{split}
\widetilde{N_2} &= M N_2 M^\dagger \\
\widetilde{N_3} &= \frac{1}{ie-fh} M (i N_3 - h N_4) M^\dagger \\
\widetilde{N_4} &= \frac{1}{ie-fh} M (e N_4 - f N_3) M^\dagger.
\end{split}
\end{equation}
Putting the two types of transformation in $\G$ together, we have 
\begin{equation}
\begin{split}
\sJ &\mapsto \sJ\circ \Ad_{\uu} \\
\sB &\mapsto e \sB\circ \Ad_{\uu} + f \sP\circ \Ad_{\uu} \\
\sP &\mapsto h \sB\circ \Ad_{\uu} + i \sP\circ \Ad_{\uu} \\
\sH &\mapsto a \sH + b \sZ \\
\sZ &\mapsto d \sZ \\
\sQ &\mapsto \sQ\circ \Ad_{\uu} \circ R_M.
\end{split}
\end{equation}
These transformations may be summarised by $(A = \big( \begin{smallmatrix} a & b \\ 0 & d\end{smallmatrix} \big),
C = \big(\begin{smallmatrix} e & f \\ h & i \end{smallmatrix}\big), M, \uu) \in \GL(\mathbb{R}^2) \times \GL(\mathbb{R}^2)
\times \GL(\mathbb{H}^2) \times \mathbb{H}^\times$.  Now that we have the most general element of
the subgroup $G \subset \GL(\s_{\bar{0}}) \times \GL(\s_{\bar{1}})$ for
$\s_{\bar{0}} = \k$ the universal generalised Bargmann algebra, we can restrict ourselves to 
the automorphisms of $\s_{\bar{0}}$ and set the parameters
$\lambda, \mu, \eta, \varepsilon \in \mathbb{R}$ to determine the automorphism group
for each of the generalised Bargmann algebras.  The results of this investigation are presented
in Table \ref{tab:GBA-auts}.
\\
\paragraph{$\hat{\a}$} ~\\ \\
In this instance, all the conditions vanish as $\lambda = \mu = \eta
= \varepsilon = 0$; therefore, the matrices $A$ and $C$ are left as stated above.
\\ 
\paragraph{$\hat{\n}_-$}  ~\\ \\
Having $\lambda = - \varepsilon = 1$ and $\mu = \eta = 0$, the
conditions in \eqref{eq:AB_constraints} become
\begin{equation}
	\begin{split}
		0 &= e (a -1) \\
		0 &= f (1 +a) 
	\end{split}  \quad \text{and} \quad
	\begin{split}
		0 &= h (1+a) \\
		0 &= i (1-a).
	\end{split}
\end{equation}
Notice, if $a \notin \{ \pm 1\}$ then $C$ must vanish, which cannot happen if we are to retain
the basis elements $\B$ and $\P$.  Therefore, we are left with two cases: $a = 1$ and $a = -1$.
In the former instance, we have automorphisms with
\begin{equation}
	A = \begin{pmatrix}
		1 & b \\ 0 & ei
	\end{pmatrix} \quad \text{and} \quad 
	C = \begin{pmatrix}
		e & 0 \\ 0 & i
	\end{pmatrix}.
\end{equation}
In the latter instance, we have
\begin{equation}
	A = \begin{pmatrix}
		-1 & b \\ 0 & -hf
	\end{pmatrix} \quad \text{and} \quad 
	C = \begin{pmatrix}
		0 & h \\ f & 0
	\end{pmatrix}.
\end{equation}
\paragraph{$\hat{\n}_+$} ~\\ \\ In this case, $\lambda = \varepsilon = 0$ and $\mu = - \eta = 1$.
Therefore, our constraints become
\begin{equation}
	\begin{split}
		0 &= h + af \\
		0 &= i - ae
	\end{split} \quad \text{and} \quad
	\begin{split} 
		0 &= e - ai \\
		0 &= f + ah.
	\end{split}
\end{equation}
Taking the expressions for $h$ and $i$ from the conditions on the left and substituting them
into the conditions on the right, we find
\begin{equation}
	0 = (1-a^2) f \qquad 0 = (1 - a^2) e.
\end{equation}
If $a^2 \neq 1$, we would need both $f$ and $e$ to vanish, which contradicts our assumption
that $C \in \GL(2, \mathbb{R})$.  Therefore, we need $a^2 = 1$, which presents two cases:
$a = 1$ and $a=-1$.  In the former instance, we find automorphisms of the form
\begin{equation}
	A = \begin{pmatrix}
		1 & b \\ 0 & e^2 + h^2
	\end{pmatrix} \quad \text{and} \quad 
	C = \begin{pmatrix}
		e & h \\ -h & e
	\end{pmatrix}.
\end{equation}
In the latter instance, we get
\begin{equation}
	A = \begin{pmatrix}
		-1 & b \\ 0 & -e^2 -h^2
	\end{pmatrix} \quad \text{and} \quad
	C = \begin{pmatrix}
		e & h \\ h & -e
	\end{pmatrix}.
\end{equation}
\paragraph{$\hat{\g}$} ~\\ \\Finally, we have $\lambda = \eta = \varepsilon = 0$ and $\mu = -1$,
which, when substituted into \eqref{eq:AB_constraints}, produces
\begin{equation}
		0 = h \quad \text{and} \quad i = ae.
\end{equation}
Therefore, automorphisms for the Bargmann algebra take the form
\begin{equation}
	A = \begin{pmatrix}
		a & b \\ 0 & ie
	\end{pmatrix} \quad \text{and} \quad
	C = \begin{pmatrix}
		e & f \\ 0 & ae
	\end{pmatrix}.
\end{equation}
\begin{table}[h!]
  \centering
  \caption{Automorphisms of the generalised Bargmann algebras}
  \label{tab:GBA-auts}
  \setlength{\extrarowheight}{2pt}
  \begin{tabular}{l|>{$}l<{$}}\toprule
    \multicolumn{1}{c|}{$\k$} & \multicolumn{1}{c}{General $(A, C) \in \GL(\mathbb{R}^2) \times \GL(\mathbb{R}^2)$} \\
    \toprule
    $\hat{\a}$ & \left(\begin{pmatrix} a & b \\ 0 & d\end{pmatrix}, \quad  \begin{pmatrix} e & f \\ h & i \end{pmatrix}\right) \\
    $\hat{\n}_-$ & \left(\begin{pmatrix}
		1 & b \\ 0 & ei
	\end{pmatrix}, \quad 
	\begin{pmatrix}
		e & 0 \\ 0 & i
	\end{pmatrix}\right) \cup \left(\begin{pmatrix}
		-1 & b \\ 0 & -hf
	\end{pmatrix}, \quad 
	 \begin{pmatrix}
		0 & h \\ f & 0
	\end{pmatrix}\right) \\
    $\hat{\n}_+$ & \left(\begin{pmatrix}
		1 & b \\ 0 & e^2 + h^2
	\end{pmatrix}, \quad 
	 \begin{pmatrix}
		e & h \\ -h & e
	\end{pmatrix}\right) \cup \left(\begin{pmatrix}
		-1 & b \\ 0 & -e^2 -h^2
	\end{pmatrix}, \quad 
	 \begin{pmatrix}
		e & h \\ h & -e
	\end{pmatrix}\right) \\
    $\hat{\g}$ & \left(\begin{pmatrix}
		a & b \\ 0 & ie
	\end{pmatrix},\quad 
	\begin{pmatrix}
		e & f \\ 0 & ae
	\end{pmatrix}\right) \\
    \bottomrule
  \end{tabular}
\end{table}
\subsection{Establishing Branches} \label{subsec:N2_branches}
Before proceeding to the discussion in which the non-empty sub-branches are identified, we first establish the possible $[\s_{\bar{0}}, \s_{\bar{1}}]$
brackets.  More specifically, we establish the possible forms for $\mz, \mh, \mb, \mp \in \Mat_2(\mathbb{H})$.
In this section, we focus solely on the results of Lemma \ref{lem:N2_001}
concerning the $(\s_{\bar{0}}, \s_{\bar{0}}, \s_{\bar{1}})$ super-Jacobi identities.
Using the universal generalised Bargmann algebra, we find that $\mb, \mp \in \Mat_2(\mathbb{H})$,
which encode the brackets $[\B, \Q]$ and $[\P, \Q]$, respectively, form a
double complex.  Analysing this structure, we identify four possible
cases:
\begin{enumerate}
	\item $\mb = 0$ and $\mp = 0$
	\item $\mb = 0$ and $\mp \neq 0$
	\item $\mb \neq 0$ and $\mp = 0$
	\item $\mb \neq 0$ and $\mp \neq 0$. 
\end{enumerate}
Taking each of these cases in turn, we find forms for $\mz$ and $\mh$ to establish four branches.  These branches
will form the basis for our investigations into the possible super-extensions for each
of the generalised Bargmann algebras in section \ref{subsec:N2_class}.
\\ \\
Using the results of Lemma \ref{lem:N2_001}, we notice that $\mb^2 = \mp^2 = 0$ and 
$\mb\mp + \mp\mb = 0$; therefore, $\mb$ and $\mp$ are the differentials of a double complex in which
the modules are $\s_{\bar{1}}$.  What does 
this mean for the form of $\mb$ and $\mp$?  Notice that we could simply set $\mb$ and $\mp$ to zero.
However, assuming at least one component of these matrices is non-vanishing, we find the following
cases.  Take $\mp$ as our example and let
\begin{equation}
	\mp = \begin{pmatrix} \pp_1 & \pp_2 \\ \pp_3 & \pp_4 \end{pmatrix}.
\end{equation}
The fact that this squares to zero tells us
\begin{equation} \label{eq:sqrd_mats}
\pp_1^2 + \pp_2 \pp_3 = 0 \qquad \pp_1 \pp_2 + \pp_2 \pp_4 = 0 \qquad \pp_3 \pp_1 + \pp_4 \pp_3 = 0
\qquad \pp_3 \pp_2 + \pp_4^2 = 0.
\end{equation}
There are two cases, $\pp_3 = 0$ and $\pp_3 \neq 0$, which we shall now consider in turn.
\\ \\
In the $\pp_3 =0$ case, the constraints in \eqref{eq:sqrd_mats} become
\begin{equation}
	\pp_1^2 = 0 \qquad \pp_1 \pp_2 + \pp_2 \pp_4 = 0 \qquad  \pp_4^2 = 0.
\end{equation}
Therefore, $\pp_1 = \pp_4 = 0$ and $\pp_2$ is unconstrained, leaving the matrix
\begin{equation}
	\mp = \begin{pmatrix}0 & \pp_2 \\ 0 & 0 \end{pmatrix}.
\end{equation}
In the $\pp_3 \neq 0$ case, we can use the first and third constraints of 
\eqref{eq:sqrd_mats} to get $\pp_2 = - \pp_1^2 \pp_3^{-1}$ and 
$\pp_4 = - \pp_3 \pp_1 \pp_3^{-1}$, respectively.  These choices trivially satisfy the second and 
fourth constraints such that we arrive at
\begin{equation}
	\mp = \begin{pmatrix} \pp_1 & -\pp_1^2 \pp_3^{-1} \\ \pp_3 & -\pp_3\pp_1\pp_3^{-1} \end{pmatrix}.
\end{equation}
In a completely analogous manner, we find
\begin{equation}
	\mb = \begin{pmatrix}0 & \bb_2 \\ 0 & 0 \end{pmatrix} \quad \text{and} \quad \mb =
	\begin{pmatrix} \bb_1 & -\bb_1^2 \bb_3^{-1} \\ \bb_3 & -\bb_3\bb_1\bb_3^{-1} \end{pmatrix}.
\end{equation}
Now, what does the anti-commuting condition tell us about the non-vanishing matrices?  
Notice, we have four options:
\begin{enumerate}
	\item $\pp_3 \neq 0$, $\bb_3 \neq 0$, 
	\item $\pp_3 \neq 0$, $\bb_3 = 0$, 
	\item $\pp_3 = 0$, $\bb_3 \neq 0$, and
	\item $\pp_3 = 0$, $\bb_3 = 0$.
\end{enumerate}
\paragraph{\textbf{Option 1}}Here we will find three distinct sub-options.  Interestingly, 
these three sub-options are equivalent to options 2, 3, and 4 above.  Substituting
the matrices associated with  $\pp_3 \neq 0$ and $\bb_3 \neq 0$ into 
$\mb\mp + \mp\mb = 0$ gives us
\begin{equation} \label{eq:N2_anti_commuting_conditions}
	\begin{split}
		 0 &= \bb_1\pp_1 - \bb_1^2 \bb_3^{-1}\pp_3 + \pp_1 \bb_1 - \pp_1^2 \pp_3^{-1} \bb_3 \\
		 0 &= -\bb_1 \pp_1^2 \pp_3^{-1} + \bb_1^2 \bb_3^{-1} \pp_3 \pp_1 \pp_3^{-1} - \pp_1 \bb_1^2 
		 \bb_3^{-1} + \pp_1^2 \pp_3^{-1} \bb_3 \bb_1 \bb_3^{-1} \\
		 0 &= \bb_3 \pp_1 - \bb_3 \bb_1 \bb_3^{-1} \pp_3 + \pp_3 \bb_1 - \pp_3 \pp_1 \pp_3^{-1} \bb_3 \\
		 0 &= - \bb_3 \pp_1^2 \pp_3^{-1} + \bb_3 \bb_1 \bb_3^{-1} \pp_3 \pp_1 \pp_3^{-1} - \pp_3 \bb_1^2 \bb_3^{-1}
		 + \pp_3 \pp_1 \pp_3^{-1} \bb_3 \bb_1 \bb_3^{-1} .
	\end{split}
\end{equation}
Multiplying the first of these conditions on the right by $\bb_1\bb_3^{-1}$ and adding it to the second condition,
we obtain
\begin{equation}
0 = \bb_1 (\pp_1 - \bb_1\bb_3^{-1}\pp_3)(\bb_1\bb_3^{-1} - \pp_1 \pp_3^{-1}).
\end{equation}
Since the quaternions have no zero-divisors, one of these terms must vanish.  The
vanishing of the second is equivalent to the vanishing of the third, so we have two sub-options:
\begin{enumerate}[label=1.\arabic*]
	\item $\bb_1 = 0$, and
	\item $\bb_1\bb_3^{-1} = \pp_1 \pp_3^{-1}$.
\end{enumerate}
In the latter case, the third and fourth conditions of \eqref{eq:N2_anti_commuting_conditions} 
are trivially satisfied, but in the former case,
a little more work is required.  Setting $\bb_1 = 0$, we obtain
\begin{equation}
0 = \bb_3 \pp_1^2\pp_3^{-1} \quad \text{and} \quad 0 = \bb_3 \pp_1 - \pp_3 \pp_1 \pp_3^{-1} \bb_3.
\end{equation}
Again, using the fact the quaternions have no zero-divisors, these conditions mean this sub-option 
further divides into two sub-options:
\begin{enumerate}[label=1.1.\arabic*]
	\item $\bb_3 = 0$, and
	\item $\pp_1 = 0$,
\end{enumerate}
with $\pp_3$ left free. Recall that to arrive at these options we first made a choice to multiply 
the first condition of \eqref{eq:N2_anti_commuting_conditions} by $\bb_1\bb_3^{-1}$.  We could equally have multiplied
by $\pp_1\pp_3^{-1}$ such that case 1.1 above read $\pp_1 = 0$.  (Notice, the second
case is symmetric, so would remain the same in this instance.)  Analogous subsequent
calculations would lead to sub-options $\pp_3 = 0$ and $\bb_1 = 0$.  Putting all of this together, 
we have four sub-options to consider:
\begin{equation}
	\begin{alignedat}{2}
		\text{Sub-option 1:} \quad \mb &= 0 \quad \hspace*{26.25mm} \mp = \begin{pmatrix} \pp_1 & - \pp_1^2 \pp_3^{-1} \\
		\pp_3 & -\pp_3\pp_1\pp_3^{-1} \end{pmatrix}   \\
		\text{Sub-option 2:} \quad \mb &= \begin{pmatrix} 0 & 0 \\ \bb_3 & 0 \end{pmatrix} \quad
		\hspace*{3.25mm} \quad \qquad \mp = \begin{pmatrix} 0 & 0 \\ \pp_3 & 0 \end{pmatrix}  \\
		\text{Sub-option 3:} \quad \mb &= \begin{pmatrix} \bb_1 & -\bb_1^2 \bb_3^{-1} \\ \bb_3 &
		 -\bb_3\bb_1\bb_3^{-1} \end{pmatrix} \quad \mp = 0  \\
		 \text{Sub-option 4:} \quad \mb &= \begin{pmatrix} \bb_1 & -\bb_1^2 \bb_3^{-1} \\ \bb_3 &
		 -\bb_3\bb_1\bb_3^{-1} \end{pmatrix} \quad  \mp = \begin{pmatrix} \pp_1 & - \pp_1^2 \pp_3^{-1} \\
		\pp_3 & -\pp_3\pp_1\pp_3^{-1} \end{pmatrix} \quad \text{where} \quad \bb_1\bb_3^{-1} 
		= \pp_1 \pp_3^{-1} .
	\end{alignedat}
\end{equation}
In fact, this list can be simplified further.  For all generalised Bargmann algebras,
we can choose a transformation $(\mathbb{1}, \mathbb{1},
M, 1)$, where, $M$ takes the form
\begin{equation} \label{eq:N2_diagonalising_auto}
	M = \begin{pmatrix}
		1 & -\bb_1 \bb_3^{-1} \\ 0 & 1
	\end{pmatrix},
\end{equation}
such that sub-option 4 becomes sub-option 2.  In summary, the $\pp_3 \neq 0$ and 
$\bb_3 \neq 0$ assumption lead to three separate sub-options.
\begin{equation}
	\begin{alignedat}{2}
		\text{Sub-option 1:} \quad \mb &= 0 \quad \hspace*{26.25mm} \mp = \begin{pmatrix} \pp_1 & - \pp_1^2 \pp_3^{-1} \\
		\pp_3 & -\pp_3\pp_1\pp_3^{-1} \end{pmatrix}   \\
		\text{Sub-option 2:} \quad \mb &= \begin{pmatrix} 0 & 0 \\ \bb_3 & 0 \end{pmatrix} \quad
		\hspace*{3.25mm} \quad \qquad \mp = \begin{pmatrix} 0 & 0 \\ \pp_3 & 0 \end{pmatrix}  \\
		\text{Sub-option 3:} \quad \mb &= \begin{pmatrix} \bb_1 & -\bb_1^2 \bb_3^{-1} \\ \bb_3 &
		 -\bb_3\bb_1\bb_3^{-1} \end{pmatrix} \quad \mp = 0.
	\end{alignedat}
\end{equation}
\paragraph{\textbf{Option 2}}Letting $\pp_3 \neq 0$ and $\bb_3 = 0$, the anti-commuting condition
tells us 
\begin{equation}
	0=\bb_2\pp_3 \qquad \text{and} \qquad  \pp_1\bb_2 = \bb_2 \pp_3 \pp_1 \pp_3^{-1}.
\end{equation}
Using the first condition, $\bb_2 = 0$, and, with $\bb_2 = 0$, we are left with sub-option 1 
above.  \\
\paragraph{\textbf{Option 3}}
Now, consider $\pp_3 = 0$ and $\bb_3 \neq 0$.  Substituting the relevant forms of $\mp$ and $\mb$ 
into the anti-commuting condition, $\mb\mp + \mp\mb=0$, we find
\begin{equation}
	0=\pp_2\bb_3 \qquad \text{and} \qquad \bb_1\pp_2 = \pp_2 \bb_3 \bb_1 \bb_3^{-1}.
\end{equation}
This is identical to option 2 only $\bb$ and $\pp$ have been swapped.  Therefore,
we have a similar result: $\pp = 0$ such that we have sub-option 3 above. \\
\paragraph{\textbf{Option 4}}
The final case to consider is $\pp_3 = 0$ and $\bb_3 = 0$, where 
\begin{equation}
	\mp = \begin{pmatrix}0 & \pp_2 \\ 0 & 0\end{pmatrix} \quad \text{and}
		\quad \mb = \begin{pmatrix}0 & \bb_2 \\ 0 & 0\end{pmatrix}.
\end{equation}
These strictly upper-triangular matrices are equivalent to the strictly
lower-triangular matrices of sub-option 2 above.  Thus, again, we find no new cases to
carry forward.
\\ \\
To simplify the rest of the calculations, we will choose to use the transformation in 
\eqref{eq:N2_diagonalising_auto} for all generalised Bargmann algebras and all options.
Combining the case in which both $\mb$ and $\mp$ vanish with the non-vanishing
options, we find
\begin{equation}
		\begin{alignedat}{2}
		\text{Case 1:} \quad \mb &= 0 \quad \hspace*{26.25mm}  \mp = 0 \\
		\text{Case 2:} \quad \mb &= 0 \quad \hspace*{26.25mm} \mp = \begin{pmatrix} 0 & 0 \\
		\pp_3 & 0 \end{pmatrix}   \\
		\text{Case 3:} \quad \mb &= \begin{pmatrix} 0 & 0 \\ \bb_3 & 0 \end{pmatrix} \quad
		\hspace*{3.25mm} \quad \qquad \mp = \begin{pmatrix} 0 & 0 \\ \pp_3 & 0 \end{pmatrix}  \\
		\text{Case  4:} \quad \mb &= \begin{pmatrix} 0 & 0 \\ \bb_3 & 0 \end{pmatrix} 
		\hspace*{3.25mm} \quad \quad \qquad \mp = 0.
	\end{alignedat}
\end{equation}
In all cases, it is a straight-forward computation to show that 
$[\mb, \mp] = \mz$ tells us that $\mz = 0$.  Therefore, we are left with only $\mh$ to determine.  
From the results in Lemma \ref{lem:N2_001}, the conditions we have including $\mb$, $\mp$ and  $\mh$ are 
\begin{equation} \label{eq:general_BP}
[\mb,\mh] = \lambda \mb + \mu \mp \quad \text{and} \quad [\mp,\mh] = \eta \mb + \varepsilon \mp.
\end{equation}
\paragraph{\textbf{Case 1}} The vanishing of $\mb$ and $\mp$ in this instance, when substituted into \eqref{eq:general_BP},
means we do not obtain any conditions on $\mh$.  Thus, we find a branch with matrices
\begin{equation}
		\mb = \mp = \mz = 0 \quad \text{and} \quad \mh \quad \text{unconstrained} .
\end{equation}
\paragraph{\textbf{Case 2}}
Notice that the vanishing of $\mb$ means that the second condition in \eqref{eq:general_BP} becomes
\begin{equation}
	\varepsilon \begin{pmatrix} 0 & 0 \\ \pp_3 & 0 \end{pmatrix} = \begin{pmatrix} 0 & 0 \\ \pp_3 & 0 \end{pmatrix} 
	\begin{pmatrix} \hh_1 & \hh_2 \\ \hh_3 & \hh_4 \end{pmatrix} - \begin{pmatrix} \hh_1 & \hh_2 \\ \hh_3 & \hh_4 \end{pmatrix}
	\begin{pmatrix} 0 & 0 \\ \pp_3 & 0 \end{pmatrix}. 
\end{equation}
This gives us two constraints
\begin{equation}
	0 = \hh_2 \pp_3 \qquad \text{and} \qquad \varepsilon \pp_3 = \pp_3 \hh_1 - \hh_4 \pp_3.
\end{equation}
The first constraint here tells us that either $\hh_2$ or $\pp_3$ must vanish.  In the latter instance, we recover the matrices from case 1.  In the former instance, we can use the second constraint to write $\hh_4$ in terms of $\hh_1$ and 
obtain the matrices
\begin{equation}
	\mb = \mz = 0 \quad \mp = \begin{pmatrix} 0 & 0 \\ \pp_3 & 0 \end{pmatrix} \quad 
	\mh = \begin{pmatrix} \hh_1 & 0 \\ \hh_3 & \pp_3 \hh_1 \pp_3^{-1} - \varepsilon \end{pmatrix} .
\end{equation}
The first condition in \eqref{eq:general_BP} does not add any new branches to those already given
as, with $\mb = 0$, it reduces to $0 = \mu \mp$.  Therefore, for those generalised Bargmann algebras
with $\mu \neq 0$, it gives the branch identified in case 1, and, for those with $\mu = 0$, it
leaves $\pp_3$ free to fix $\hh_4$ as prescribed for the branch presented in this case.
\\
\paragraph{\textbf{Case 3}}
Substituting the $\mb$ and $\mp$ associated with this case into 
\eqref{eq:general_BP}, we get the following constraints
\begin{equation} \label{eq:branch4_h_constraints}
	0 = \hh_2 \bb_3 \qquad 0 = \hh_2 \pp_3 \qquad \lambda \bb_3 + \mu \pp_3 = \bb_3 \hh_1 - \hh_4 \bb_3 \qquad
	\eta \bb_3 + \varepsilon \pp_3 = \pp_3 \hh_1 - \hh_4 \pp_3.
\end{equation}
The first two constraints above tell us that if $\hh_2 \neq 0$, then we again arrive at the branch with
$\mb= \mp = \mz = 0$ and $\mh$ unconstrained. Letting $\hh_2 = 0$, we focus on the second two constraints.
Notice, for this branch to be distinct from the other two, we require $\bb_3 \neq 0$ and $\pp_3 \neq 0$.  
These assumptions allow us to take inverses of both $\bb_3$ and $\pp_3$ in the following calculations.
Multiplying the third constraint on the right by $\bb_3^{-1}$, we can rearrange for $\hh_4$ and substitute 
this into the fourth constraint to get
\begin{equation}
	\eta \bb_3 + \varepsilon \pp_3 = \pp_3 \hh_1 - \bb_3 \hh_1 \bb_3^{-1} \pp_3 + \mu \pp_3 \bb_3^{-1} \pp_3 
	+ \lambda \pp_3.
\end{equation}
Multiplying this expression by $\bb_3^{-1}$ on the left and rearranging, we find
\begin{equation}
	[\uu, \hh_1] = -\mu \uu^2 + (\varepsilon - \lambda) \uu + \eta,
\end{equation}
where $\uu = \bb_3^{-1} \pp_3$.  Alternatively, we could have chosen to multiply the fourth condition on
the right by $\pp_3^{-1}$ to get our expression for $\hh_4$ and substituted this into the third constraint.
Multiplying this on the left by $\pp_3^{-1}$ produces the similar condition
\begin{equation}
	[\vv, \hh_1] = \eta \vv^2 + (\lambda - \varepsilon) \vv + \mu,
\end{equation}
where $\vv = \pp_3^{-1} \bb_3$.  Depending on the generalised Bargmann algebra in question, one
of these will prove more useful than the other. We will leave these constraints in this form to be analysed 
separately for each generalised Bargmann algebra.
\\
\paragraph{\textbf{Case 4}}The calculations for this case are nearly identical to those for case 2.
The vanishing of $\mp$ means that the first constraint in \eqref{eq:general_BP} produces
\begin{equation}
	0 = \hh_2 \bb_3 \qquad \text{and} \qquad \lambda \bb_3 = \bb_3 \hh_1 - \hh_4 \bb_3.
\end{equation}
From the first expression above, if $\hh_2 \neq 0$, we recover the branch presented in case 1.
However, setting $\hh_2 = 0$, $\bb_3$
is general, and we can use the second constraint to write $\hh_4$ in terms of $\hh_1$ and $\bb_3$:
\begin{equation}
	\mp = \mz = 0 \quad \mb = \begin{pmatrix} 0 & 0 \\ \bb_3 & 0 \end{pmatrix} \quad 
	\mh = \begin{pmatrix} \hh_1 & 0 \\ \hh_3 & \bb_3 \hh_1 \bb_3^{-1} - \lambda \end{pmatrix}.
\end{equation}
The second constraint in \eqref{eq:general_BP} does not produce any new branches for
$\mb, \mp,$ and $\mh$.  Substituting in $\mp = 0$,  it becomes $0 = \eta \mb$.  Therefore, 
if $\eta \neq 0$, $\mb$ must vanish leaving the branch from case 1; and, if $\eta =0$,
$\bb_3$ is left free so we can write $\hh_4$ as prescribed for the branch presented
here.
\\ \\
In summary, we have the following three branches for all generalised Bargmann algebras
\begin{enumerate}
	\item $\mb=\mp=\mz=0 \quad \text{and} \quad \mh \quad \text{unconstrained}$
	\item $\mb = \mz = 0 \quad \mp = \begin{pmatrix} 0 & 0 \\ \pp_3 & 0 \end{pmatrix} \quad 
	\mh = \begin{pmatrix} \hh_1 & 0 \\ \hh_3 & \pp_3 \hh_1 \pp_3^{-1} - \varepsilon\end{pmatrix}$
	\item  $\mp = \mz = 0 \quad \mb = \begin{pmatrix} 0 & 0 \\ \bb_3 & 0 \end{pmatrix} \quad 
	\mh = \begin{pmatrix} \hh_1 & 0 \\ \hh_3 & \bb_3 \hh_1 \bb_3^{-1} - \lambda \end{pmatrix}$.
\end{enumerate}
There is also a possible fourth branch depending on the generalised Bargmann algebra:
\begin{equation}
	\mz = 0 \quad \mh = \begin{pmatrix} \hh_1 & 0 \\ \hh_3 & \hh_4 \end{pmatrix} \quad 
	\mb = \begin{pmatrix} 0 & 0 \\ \bb_3 & 0 \end{pmatrix} \quad 
	\mp = \begin{pmatrix} 0 & 0 \\ \pp_3 & 0 \end{pmatrix},
\end{equation}
subject to
\begin{equation}
	[\uu, \hh_1] = -\mu \uu^2 + (\varepsilon - \lambda) \uu + \eta \quad 
	\text{and} \quad [\vv, \hh_1] = \eta \vv^2 + (\lambda - \varepsilon) \vv + \mu
\end{equation}
where $\uu = \bb_3^{-1} \pp_3$ and $\vv = \pp_3^{-1} \bb_3$.
\subsection{Classification} \label{subsec:N2_class}
In this section, we complete the story started in section \ref{subsec:N2_branches}.  Each branch we identified in section \ref{subsec:N2_branches} encodes the possible $[\s_{\bar{0}}, \s_{\bar{1}}]$ brackets for a generalised Bargmann superalgebra $\s$.  Here, we take each branch in turn and find corresponding $[\Q, \Q]$ brackets.  Since
our interests are in supersymmetry, we will always impose the condition that $[\Q, \Q] \neq 0$;
therefore, we are only interested in branches for which at least one of the $N_i$ matrices
does not vanish.  Note, the imposition of this condition means that the various branches 
 identified here belong to the sub-variety $\cS$ of the real
algebraic variety cut out by the super-Jacobi identities $\cJ \subset \cV$.
\\ \\
We will begin our investigation into each branch by stating the associated matrices, $\mb, \mp, \mh, \mz \in \Mat_2(\mathbb{H})$.  These matrices are then substituted into the conditions from Lemmas \ref{lem:N2_011} and \ref{lem:N2_111}, which use the Lie brackets of the universal generalised Bargmann algebra.  This process produces a system of equations containing $\mb, \mp, \mh, \mz$ encoding the $[\s_{\bar{0}}, \s_{\bar{1}}]$ components of the bracket, the matrices $N_i$ for $ i \in \{0, 1, ..., 4\}$ encoding the $[\s_{\bar{1}}, \s_{\bar{1}}]$ components of the bracket, and the four parameters of the universal generalised Bargmann algebra, $\lambda, \mu, \eta, \varepsilon \in \mathbb{R}$.  Any conditions which do not contain one of the parameters $\lambda, \mu, \eta, \varepsilon$ are analysed and possible dependencies among the $N_i$ matrices are found.  Once these dependencies have been established, we start setting parameters to consider the various generalised Bargmann algebras.  In branches 1 and 2, we will see that multiple generalised Bargmann algebras produce the same set of conditions.  In these instances, we will highlight the relevant algebras but only analyse the system once to avoid repetition.  
\\ \\
In branches $\mathsf{2}$, $\mathsf{3}$ and $\mathsf{4}$, we find that the vanishing of certain matrices $N_i$ imposes the vanishing of other $N_i$. Thus, we end up with a chain of dependencies, which lead to different sub-branches.  These sub-branches will be labelled such that sub-branches with a larger branch number will have more non-vanishing matrices $N_i$.  For example, sub-branch $\mathsf{2.2}$ may have non-vanishing $N_0$ and $N_1$, but sub-branch $\mathsf{2.3}$ may additionally have non-vanishing $N_3$. Within each sub-branch, we regularly find two options: one in which $N_0$ vanishes, leaving $\mh$ free, and one in which $\mh=0$ such that $N_0$ is unconstrained.  Using sub-branch $\mathsf{2.2}$ as our example, the former instance, with $N_0 = 0$, will be labelled $\mathsf{2.2.i}$, and the latter instance will be labelled $\mathsf{2.2.ii}$.  In branch 4, we will find some instances in which both $N_0$ and $\mh$ can be non-vanishing.  Using sub-branch 4.3 as an example, we will label these cases as $\mathsf{4.3.iii}$.
\\ \\
Each sub-branch is designed to have a unique set of non-vanishing matrices.  However, the components within the matrices are not completely fixed by the super-Jacobi identities.  Therefore, each sub-branch is given as a tuple $(\mathcal{M}_{\k , \, \mathsf{X} }, \mathcal{C}_{\k , \, \mathsf{X} } )$, where $\k$ labels the underlying generalised Bargmann algebra, and $\mathsf{X}$ will be the branch number.  This tuple consists of $\mathcal{M}$, the subset of matrices in $\{\mb, \mp, \mh, \mz, N_0, N_1, N_2, N_3, N_4\}$ describing the branch, and $\mathcal{C}$, the set of constraints on the components of the matrices.  After stating $(\mathcal{M}, \mathcal{C})$ for a given sub-branch, we proceed to a discussion on possible parameterisations of the super-extensions in the sub-branch.  In particular, the aim of these discussions is to highlight the existence of super-extensions in the sub-branch.  First we set as many of the parameters to zero as possible.  In general, this will involve setting $\mh$ to zero along with a small number of components in the matrices $N_i$.  Then, using any residual transformations in the group $\G \subset \GL(\s_{\bar{0}}) \times \GL(\s_{\bar{1}})$, we fix the remaining parameters.  Once the existence of super-extensions has been established, we introduce some other parameters to produce further examples of generalised Bargmann superalgebras contained within the sub-branch.
\\ \\
Recall, we build the $[\Q, \Q]$ bracket from the $N_i$ matrices as follows
\begin{equation}
	[\sQ(\vt), \sQ(\vt)] = \Re(\vt N_0\vt^\dagger) \sH + \Re(\vt N_1\vt^\dagger) \sZ - \sJ(\vt N_2\vt^\dagger)  - \sB(\vt N_3\vt^\dagger)  - \sP(\vt N_4\vt^\dagger),
\end{equation}
where $N_0$ and $N_1$ are quaternion Hermitian, $N_i^\dagger = N_i$, and $N_2$, $N_3$ and $N_4$ are quaternion skew-Hermitian,
$N_j^\dagger = - N_j$.  Throughout this section, we will use the following forms for the quaternion Hermitian matrices:
\begin{equation} \label{eq:N2_C_H_Herm_Mats}
	N_0 = \begin{pmatrix} a & \qq \\ \bar{\qq} & b \end{pmatrix} \quad \text{and} \quad 
	N_1 = \begin{pmatrix} c & \rr \\ \bar{\rr} & d \end{pmatrix},
\end{equation}
where $a, b, c, d \in \mathbb{R}$, and $\qq, \rr \in \mathbb{H}$.  The
quaternion skew-Hermitian matrices will be defined
\begin{equation} \label{eq:N2_C_H_SkewHerm_Mats}
	N_3 = \begin{pmatrix} \ee & \ff \\ -\bar{\ff} & \gg \end{pmatrix} \quad \text{and} \quad 
	N_4 = \begin{pmatrix} \nn & \mm \\ -\bar{\mm} & \ll \end{pmatrix},
\end{equation}
where $\ee, \gg, \nn, \ll \in \Im(\mathbb{H})$ and $\ff, \mm \in \mathbb{H}$.
We will only briefly need to consider parts of the $N_2$ matrix explicitly; therefore, we will define its components as necessary.
\subsubsection{Branch 1} \label{subsubsec:N2_C_branch1}
\begin{equation}
	\mb=\mp=\mz=0 \quad \text{and} \quad \mh \quad \text{unconstrained}.
\end{equation}
Using the remaining conditions from the $(\s_{\bar{0}}, \s_{\bar{1}}, \s_{\bar{1}})$ and $(\s_{\bar{1}},
\s_{\bar{1}}, \s_{\bar{1}})$ super-Jacobi identities, we can look to find some expressions for the matrices $N_i$.  The
conditions derived from the $[\B, \Q, \Q]$ identity in Lemma \ref{lem:N2_011} immediately give 
$N_4 = 0$ due to the vanishing of $\mb$. Similarly, the $[\P, \Q, \Q]$ conditions give us $N_3 = 0$ due to the vanishing of $\mp$.  We are thus left with
\begin{equation}
	\begin{split}
		0 &= \mh N_i + N_i \mh^\dagger \quad i \in \{0, 1, 2\} \\
		0 &= \Re(\vt N_0\vt^\dagger) \vt \mh - \tfrac12 \vt N_2 \vt^\dagger \vt
	\end{split} \qquad 
	\begin{split}
		0 &= \mu \Re(\vt N_0 \vt^\dagger) \\
		0 &= \eta \Re(\vt N_0 \vt^\dagger) \\
		0 &= \lambda \Re(\vt N_0\vt^\dagger) \beta + \tfrac12 [\beta, \vt N_2 \vt^\dagger ]  \\
		0 &= \varepsilon \Re(\vt N_0\vt^\dagger) \pi + \tfrac12 [\pi, \vt N_2 \vt^\dagger] \\
	\end{split} \qquad \forall \beta, \pi \in \Im(\mathbb{H}), \; \forall \vt \in \mathbb{H}^2.
\end{equation}
Since $[\cc, \dd]$ is perpendicular to both $\cc$ and $\dd$ for $\cc, \dd \in \mathbb{H}$, the final two
conditions can be reduced to 
\begin{equation}
	0 = \lambda \Re(\vt N_0\vt^\dagger) \qquad 0 = [\beta, \vt N_2 \vt^\dagger ]  \qquad
	0 = \varepsilon \Re(\vt N_0\vt^\dagger) \qquad 0 = [\pi, \vt N_2 \vt^\dagger] .
\end{equation}
Substituting $\vt = (1, 0)$, $\vt = (0, 1)$, and $\vt = (1, 1)$ into the $N_2$ conditions above, we find that 
\begin{equation}
N_2 = \begin{pmatrix} 0 & e \\ -e & 0\end{pmatrix},
\end{equation}
where $e \in \mathbb{R}$.  Now substituting $\vt = (1, \ii)$ into the $N_2$ conditions, we find
\begin{equation}
0 = -2 e [\beta, \ii].
\end{equation}
We can choose any $\beta \in \Im(\mathbb{H})$; therefore, we may choose
$\beta = \jj$.  Thus we find that $e$ must vanish, making $N_2 = 0$.
This result reduces the conditions further:
\begin{equation} \label{eq:N2_branch1_general_conditions}
	\begin{split}
		0 &= \mh N_i + N_i \mh^\dagger \quad i \in \{0, 1\} \\
		0 &= \Re(\vt N_0\vt^\dagger) \vt \mh 
	\end{split} \qquad 
	\begin{split}
		0 &= \mu \Re(\vt N_0 \vt^\dagger) \\
		0 &= \eta \Re(\vt N_0 \vt^\dagger) \\
		0 &= \lambda \Re(\vt N_0\vt^\dagger) \beta \\
		0 &= \varepsilon \Re(\vt N_0\vt^\dagger) \pi \\
	\end{split} \qquad \forall \beta, \pi \in \Im(\mathbb{H}), \; \forall \vt \in \mathbb{H}^2.
\end{equation}
Focussing on the conditions common to all generalised Bargmann algebras,
i.e.\ those conditions which do not contain $\lambda$, $\mu$, $\eta$, or $\varepsilon$,
we have only
\begin{equation}
	\begin{split}
		0 &= \mh N_i + N_i \mh^\dagger \quad i \in \{0, 1\} \\
		0 &= \Re(\vt N_0\vt^\dagger) \vt \mh.
	\end{split}
\end{equation}
Since the second condition must hold for all $\vt \in \mathbb{H}^2$, we find
that either
\begin{enumerate}[label=(\roman*)]
	\item $N_0 = 0$ and $\mh \neq 0$, or
	\item $N_0 \neq 0$ and $\mh = 0$.
\end{enumerate}
We can now split this analysis in two depending on the generalised
Bargmann algebra of interest. First, we will discuss the algebras in which at least 
one of the parameters $\lambda, \mu, \eta, \varepsilon$
are non-vanishing.  Subsequently, we will consider the algebras in which 
all of these parameters vanish.  The former instance encapsulates 
$\hat{\n}_{\pm}$ and $\hat{\g}$, and the latter encapsulates $\hat{\a}$.  
\\
\paragraph{\textit{$\hat{\n}_{\pm}$ and $\hat{\g}$}}~\\ \\
All of these algebras have non-vanishing values for at least one of the parameters,
$\lambda, \mu, \eta, \varepsilon$.  Therefore, all have the conditions for branch $\mathsf{1}$ reduce to
\begin{equation}
	\begin{split}
		0 &= \mh N_i + N_i \mh^\dagger \quad i \in \{0, 1\} \\
		0 &= \Re(\vt N_0 \vt^\dagger) \\
		0 &= \Re(\vt N_0\vt^\dagger) \vt \mh.
	\end{split}
\end{equation}
Substituting $\vt = (1, 0)$, $\vt = (0, 1)$, and $\vt = (1, 1)$ 
into the second condition above, we find that 
\begin{equation}
N_0 = \begin{pmatrix} 0 & \Im(\qq) \\ -\Im(\qq) & 0 \end{pmatrix}.
\end{equation}
Now substitute $\vt = (1, \ii)$ into this condition
using the convention that $\qq = q_1 \ii + q_2 \jj + q_3 \kk$
to find
\begin{equation}
 	0 = \Re(\ii \bar{\qq}) = q_1.
\end{equation}
Using $\vt = (1, \jj)$ and $\vt= (1, \kk)$, we get analogous expressions for $q_2$ and $q_3$, so $\qq = 0$.
Therefore, $N_0 = 0$, and we cannot produce a super-extension in sub-branch $\mathsf{1.ii}$ for
these generalised Bargmann algebras.
\\ \\
The only remaining matrices are $\mh$ and $N_1$, such that
\begin{equation}
0 = \mh N_i + N_i \mh^\dagger,
\end{equation}
with no constraints on $\mh$ and $N_1 = N_1^\dagger$.  So far, we have not used any basis transformations
for this branch; therefore, we can choose $N_1$ to be the canonical quaternion Hermitian form, $\mathbb{1}$.
The above condition then states that $\mh^\dagger = - \mh$.  Thus, this branch produces one non-empty sub-branch for $\hat{\n}_\pm$ and $\hat{\g}$, with the set of non-vanishing \hypertarget{N2_ng_1}{matrices} given by
\begin{equation}
 \mathcal{M}_{\hat{\n}_\pm\,  \text{and} \, \hat{\g}, \, \mathsf{1.i}} = \Big\{ \mh = \begin{pmatrix} \hh_1 & \hh_2 \\ -\bar{\hh_2} & \hh_3 \end{pmatrix}, \quad  N_1 = \begin{pmatrix}
	1 & 0 \\ 0 & 1
\end{pmatrix} \Big\} .
\end{equation}
Although already explicit in the forms of $\mh$ and $N_1$, we note that the set of constraints for this sub-branch is
\begin{equation}
	\mathcal{C}_{\hat{\n}_\pm\,  \text{and} \, \hat{\g}, \, \mathsf{1.i}} = \{ \mh^\dagger = - \mh, \quad  N_1 = N_1^\dagger \} .
\end{equation}
Our only comment on $\mh$ going into the analysis of this branch was that it was unconstrained; therefore, we may choose to have $\mh = 0$.  Thus there is certainly a super-extension in this sub-branch, one with only $N_1 = \mathbb{1}$ non-vanishing.  However, wanting to introduce some more parameters, we may let $\hh_1$, $\hh_2$ and $\hh_3$ be non-vanishing.  These quaternions can be fixed using the group of basis transformations $\G \subset \GL(\s_{\bar{0}})\times\GL(\s_{\bar{1}})$ by noticing that $\mh^\dagger = -\mh$ tells us that $\mh \in \sp(2)$.  Therefore, the residual $\Sp(2) \subset \GL(\s_{\bar{1}})$ which fixes $N_1 = \mathbb{1}$ acts on $\mh$ via the adjoint action of $\Sp(2)$ on its Lie algebra.  Thus, we can make $\mh$ diagonal and choose the two imaginary quaternions parameterising it, arriving at 
\begin{equation}
	 \mh = \begin{pmatrix} \ii & 0 \\ 0 & \jj \end{pmatrix} \quad \text{and} \quad N_1 = \begin{pmatrix}
	1 & 0 \\ 0 & 1
\end{pmatrix} .
\end{equation}
\paragraph{$\hat{\a}$} ~\\ \\
Since $\hat{\a}$ has $\lambda = \mu = \eta = \varepsilon = 0$, the conditions in \eqref{eq:N2_branch1_general_conditions} become
\begin{equation}
	\begin{split}
		0 &= \mh N_i + N_i \mh^\dagger \quad i \in \{0, 1\} \\
		0 &= \Re(\vt N_0\vt^\dagger) \vt \mh.
	\end{split}
\end{equation}
Unlike the $\hat{\n}_{\pm}$ and $\hat{\g}$ case, these conditions do not instantly
set $N_0=0$; therefore, we may have super-extensions with either $(\mathsf{i})$ $N_0 = 0$ and $\mh \neq 0$, or $(\mathsf{ii})$ $N_0 \neq 0$ and $\mh=0$.
First, setting $\mh \neq 0$, we know this imposes $N_0 = 0$, and,
as in the $\hat{\n}_\pm$ and $\hat{\g}$ case, we may use the basis transformations to set 
$N_1 = \mathbb{1}$, such that $\mh^\dagger = -\mh$.  
Therefore, one of the possible super-extensions for $\hat{\a}$ has non-vanishing \hypertarget{N2_a_1i}{matrices}
\begin{equation} 
 	\mathcal{M}_{\hat{\a}, \, \mathsf{1.i}} = \Big\{ \mh = \begin{pmatrix} \hh_1 & \hh_2 \\ -\bar{\hh_2} & \hh_3 \end{pmatrix}, \quad  N_1 = \begin{pmatrix}
	1 & 0 \\ 0 & 1
\end{pmatrix} \Big\} .
\end{equation}
As before, the set of conditions for this super-extension is
\begin{equation}
		\mathcal{C}_{\hat{\a}, \, \mathsf{1.i}} = \{ \mh^\dagger = -\mh, \quad N_1^\dagger = N_1 \} ,
\end{equation}
and we can use $\G$ to fix the quaternions in $\hh$.  Alternatively, setting $N_0 \neq 0$, we need $\mh= 0$.  
Thus the second possible super-extension in this
\hypertarget{N2_a_1ii}{branch} has
\begin{equation}
	\mathcal{M}_{\hat{\a}, \, \mathsf{1.ii}} = \Big\{ N_0 = \begin{pmatrix}
		a & \qq \\ \bar{\qq} & b
	\end{pmatrix},\quad N_1 = \begin{pmatrix}
		c & \rr \\ \bar{\rr} & d
	\end{pmatrix} \Big\} \quad \text{and} \quad 
	\mathcal{C}_{\hat{\a}, \, \mathsf{1.ii}} = \{ N_0^\dagger = N_0, \quad N_1^\dagger = N_1 \} .		
\end{equation}
Since the primary constraint on these matrices is that both be non-vanishing, we can choose to have $b$, $\qq$, $c$ and $\rr$ vanish.  Using the scaling symmetry of the $\s_{\bar{0}}$ basis elements present in $\G \subset \GL(\s_{\bar{0}}) \times \GL(\s_{\bar{1}})$, we can write down the super-extension
\begin{equation}
	N_0 = \begin{pmatrix}
		1 & 0 \\ 0 & 0
	\end{pmatrix}\quad N_1 = \begin{pmatrix}
		0 & 0 \\ 0 & 1
	\end{pmatrix}.
\end{equation}
Therefore, this sub-branch is not empty.  Additionally, we may choose to keep all the parameters in the matrices of $\mathcal{M}_{\hat{\a}, \, \mathsf{1.ii}}$ and use the basis transformations to fix them.  In particular, we can let $N_0 = \mathbb{1}$.  This choice leaves us with a residual $\Sp(2)$ action with which to fix the parameters of $N_1$, which may give us $N_1 = \mathbb{1}$.  
\subsubsection{Branch 2} \label{subsubsec:N2_C_branch2}
\begin{equation}
	\mb = \mz = 0 \quad \mp = \begin{pmatrix} 0 & 0 \\ \pp_3 & 0 \end{pmatrix} \quad 
	\mh = \begin{pmatrix} \hh_1 & 0 \\ \hh_3 & \pp_3 \hh_1 \pp_3^{-1} - \varepsilon\end{pmatrix}.
\end{equation}
As above, it is useful to exploit the vanishing matrices of the branch to simplify the conditions from Lemmas \ref{lem:N2_011} and \ref{lem:N2_111}.  In particular, the $[\B, \Q, \Q]$ super-Jacobi identity tells us $N_4 = 0$ due to the vanishing of $\mb$.  The rest of the $[\B, \Q, \Q]$ conditions tells us that
\begin{equation}
	\begin{split}
		0 &= \lambda \Re(\vt N_0 \vt^\dagger) \beta + \tfrac12 [\beta, \vt N_2 \vt^\dagger] \\
		0 &= \mu \Re(\vt N_0\vt^\dagger) \beta .			
	\end{split}
\end{equation}
The $[\P, \Q, \Q]$ conditions become
\begin{equation}
	\begin{split}
		0 &= \mp N_0 - N_0 \mp^\dagger \\	
		-N_3 &= \mp N_1 - N_1 \mp^\dagger \\
		0 &= \pi \vt \mp N_2 \vt^\dagger + \vt N_2 (\pi \vt \mp)^\dagger \\
		\eta \Re(\vt N_0\vt^\dagger) \pi &= \pi \vt \mp N_3 \vt^\dagger + \vt N_3 (\pi \vt \mp)^\dagger \\
		0 &= \varepsilon \Re(\vt N_0 \vt^\dagger) \pi + \tfrac12 [\pi, \vt N_2 \vt^\dagger]. \\
	\end{split}
\end{equation}
Since the conditions from the $[\bZ,\Q, \Q]$ identity are all satisfied due to $\mz = 0$, the final conditions are
\begin{equation}
	\begin{split}
		0 &= \mh N_i + N_i \mh^\dagger \quad \text{where} \quad i \in \{0, 1, 2\} \\
		\lambda N_3 &= \mh N_3 + N_3 \mh^\dagger \\
		0 &= \mu N_3, \\
	\end{split}
\end{equation}
from $[\bH, \Q, \Q]$.  The result from Lemma \ref{lem:N2_111} then gives us 
\begin{equation}
	\Re(\vt N_0 \vt^\dagger) \vt \mh = \tfrac12 \vt N_2 \vt^\dagger \vt.
\end{equation}
As in branch 1, the conditions
\begin{equation}
			0 = \lambda \Re(\vt N_0 \vt^\dagger) \beta + \tfrac12 [\beta, \vt N_2 \vt^\dagger] \quad \text{and} \quad 
			0 = \varepsilon \Re(\vt N_0 \vt^\dagger) \pi + \tfrac12 [\pi, \vt N_2 \vt^\dagger],
\end{equation}
tell us that $N_2 = 0$.  Therefore, the conditions reduce further to
\begin{equation}  \label{eq:N2_C_branch2_constraints}
	\begin{split}
		0 &= \mu N_3 \\
		0 &= \mh N_i + N_i \mh^\dagger \quad \text{where} \quad i \in \{0, 1\} \\
		0 &= \mp N_0 - N_0 \mp^\dagger \\
		0 &= \lambda \Re(\vt N_0 \vt^\dagger) \beta\\
		0 &= \mu \Re(\vt N_0\vt^\dagger) \beta \\
		0 &= \varepsilon \Re(\vt N_0 \vt^\dagger) \pi \\
		0 &= \Re(\vt N_0\vt^\dagger) \vt \mh 
	\end{split}
	\begin{split}
		\lambda N_3 &= \mh N_3 + N_3 \mh^\dagger \\
		-N_3 &= \mp N_1 - N_1 \mp^\dagger \\
		\eta \Re(\vt N_0\vt^\dagger) \pi &= \pi \vt \mp N_3 \vt^\dagger + \vt N_3 (\pi \vt \mp)^\dagger \\
	\end{split} \quad \forall \beta, \pi \in \Im(\mathbb{H}), \forall \vt \in \mathbb{H}^2.
\end{equation}
We can now use the following two conditions common to all generalised Bargmann algebras to highlight
the possible sub-branches:
\begin{equation}
		-N_3 = \mp N_1 - N_1 \mp^\dagger \quad \text{and} \quad
		0 = \Re(\vt N_0\vt^\dagger) \vt \mh.
\end{equation}
Substituting the $N_1$ from \eqref{eq:N2_C_H_Herm_Mats} and the
$N_3$ from \eqref{eq:N2_C_H_SkewHerm_Mats} into the first condition here, we can write
\begin{equation}
	N_3 = \begin{pmatrix} 0 & c\bar{\pp_3} \\ - c\pp_3 & \bar{\rr}\bar{\pp_3} - \pp_3 \rr \end{pmatrix}.
\end{equation}
This result tells us that $N_3$ is dependent on $N_1$: if $N_1 = 0$ then
$N_3=0$.  Therefore, we may organise our investigation into the possible super-extensions by
considering each of the following sub-branches in turn
\begin{enumerate}
	\item $N_1 = 0$ and $N_3 = 0$,
	\item $N_1 \neq 0$ and $N_3 = 0$,
	\item $N_1 \neq 0$ and $N_3 \neq 0$.
\end{enumerate}
Next, consider the condition from the $[\Q, \Q, \Q]$ identity:
\begin{equation}
	0 = \Re(\vt N_0 \vt^\dagger) \vt \mh.
\end{equation}
Notice, this is identical to the condition from the $[\Q, \Q, \Q]$ identity we found in \hyperlink{N2_a_1i}{branch 1}.  Therefore, as before, we have two cases to consider in each sub-branch: 
\begin{enumerate}[label=(\roman*)]
	\item $N_0 = 0$ and $\mh \neq 0$, and
	\item  $N_0 \neq 0$ and $\mh = 0$.
\end{enumerate}
We will now consider each generalised Bargmann algebra in turn to determine
whether they have super-extensions associated to these sub-branches.
\\
\paragraph{$\hat{\a}$} ~\\ \\
In addition to the conditions already discussed in producing the possible sub-branches,
\begin{equation} \label{eq:N2_C_branch2_a_conditions1}
	-N_3 = \mp N_1 - N_1 \mp^\dagger \quad \text{and} \quad
		0 = \Re(\vt N_0\vt^\dagger) \vt \mh,
\end{equation}
substituting $\lambda = \mu = \eta = \varepsilon = 0$ into \eqref{eq:N2_C_branch2_constraints} leaves us with
\begin{equation} \label{eq:N2_C_branch2_a_conditions2}
	\begin{split}
		0 &= \mh N_i + N_i \mh^\dagger \quad \text{where} \quad i \in \{0, 1, 3\} \\
		0 &= \mp N_0 - N_0 \mp^\dagger \\
		0 &= \pi \vt \pp N_3 \vt^\dagger + \vt N_3 (\pi \vt \mp)^\dagger. \\
	\end{split}
\end{equation}
None of these conditions force the vanishing of any more $N_i$; therefore, \textit{a priori}
we may find super-extensions in each of the sub-branches.  The only restriction
to the matrices so far has been the re-writing of $N_3$:
\begin{equation}
	N_3 = \begin{pmatrix} 0 & c\bar{\pp_3} \\ - c\pp_3 & \bar{\rr}\bar{\pp_3} - \pp_3 \rr \end{pmatrix}.
\end{equation}
\paragraph{\textbf{Sub-branch 2.1}} Setting $N_1 = N_3 = 0$, we are left with only $N_0$, subject to 
\begin{equation}
		0 = \mh N_0 + N_0 \mh^\dagger \quad \text{and} \quad
		0 = \mp N_0 - N_0 \mp^\dagger.
\end{equation}
We know that we may have two possible cases for this sub-branch: either $(\mathsf{i})$ $N_0 = 0$ and $\mh \neq 0$,
or $(\mathsf{ii})$ $N_0 \neq 0$ and $\mh = 0$.  Since we need $N_0 \neq 0$ for a supersymmetric extension,
we must have the latter case.  This leaves only the second condition above with which to 
restrict the form of $N_0$.  Since $\pp_3 \neq 0$, this tells us
\begin{equation} \label{eq:case2_N0_constraints}
	0 = a \quad \text{and} \quad 0 = \pp_3\qq - \bar{\qq} \bar{\pp_3}.
\end{equation}
Thus the \hypertarget{N2_a_21}{sub-branch} is given by
\begin{equation}
	 \mathcal{M}_{\hat{\a}, \, \mathsf{2.1.ii}} = \Big\{ \mp = \begin{pmatrix} 0 & 0 \\ \pp_3 & 0 \end{pmatrix}, \quad 
	N_0 = \begin{pmatrix} 0 & \qq \\ \bar{\qq} & b \end{pmatrix} \Big\} \quad \text{and} \quad 
	\mathcal{C}_{\hat{\a}, \, \mathsf{2.1.ii}} = \{ 0 =  \pp_3\qq - \bar{\qq} \bar{\pp_3} \}.
\end{equation}
This sub-branch is parameterised by two collinear quaternions $\pp_3$ and $\qq$, and a single real scalar $b$, such that it defines an 8-dimensional space in the sub-variety $\cS$.  Notice that we can choose either $\qq = 0$ or $b = 0$ and this sub-branch remains supersymmetric.  Choosing the former case, we can use the endomorphisms of $\s_{\bar{1}}$ to set $\pp_3 = \ii$ and the scaling symmetry of $\sH$ to produce
\begin{equation}
	 \mp = \begin{pmatrix} 0 & 0 \\ \ii & 0 \end{pmatrix} \quad \text{and} \quad
	N_0 = \begin{pmatrix} 0 & 0 \\ 0 & 1 \end{pmatrix}.
\end{equation}
In the latter case, we can still choose $\pp_3 = \ii$, and the condition in $\mathcal{C}_{\hat{\a}, \, \mathsf{2.1.ii}}$ will impose that $\qq$ must also lie along $\ii$.  Again using the scaling symmetry of $\sH$ in $\G \subset \GL(\s_{\bar{0}}) \times \GL(\s_{\bar{1}})$, we arrive at
\begin{equation}
	 \mp = \begin{pmatrix} 0 & 0 \\ \ii & 0 \end{pmatrix} \quad \text{and} \quad 
	N_0 = \begin{pmatrix} 0 & \ii \\ -\ii & 0 \end{pmatrix}.
\end{equation}
These two examples turn out to be the only super-extensions in this sub-branch.  Keeping both $b$ and $\qq$ at the outset, we can use the endomorphisms of $\s_{\bar{1}}$ to set $b = 0$ while imposing that $\pp_3$ and $\qq$ lie along $\ii$.  Thus, in this case, we could always retrieve the second example above.
\\
\paragraph{\textbf{Sub-branch 2.2}} Setting $N_1 \neq 0$ but keeping $N_3 = 0$, the conditions in  \eqref{eq:N2_C_branch2_a_conditions1} and \eqref{eq:N2_C_branch2_a_conditions2} become
\begin{equation}
	\begin{split}
		0 &= \mp N_i - N_i \mp^\dagger  \\
		0 &= \mh N_i + N_i \mh^\dagger
	\end{split}  \quad \text{where} \quad i \in \{0, 1\}.
\end{equation}
Importantly, we can now have super-extensions in either of the two cases: $(\mathsf{i})$ $N_0 = 0$ and $\mh \neq 0$ , or $(\mathsf{ii})$ $N_0 \neq 0$ and $\mh = 0$.  In the former case, in which $N_0 = 0$, \eqref{eq:N2_C_branch2_a_conditions1} and \eqref{eq:N2_C_branch2_a_conditions2} become
\begin{equation}
		0 = \mp N_1 - N_1 \mp^\dagger \quad \text{and} \quad 
		0 = \mh N_1 + N_1 \mh^\dagger.
\end{equation}
The first of these conditions tells us that 
\begin{equation}
	N_1 = \begin{pmatrix}
		0 & \rr \\ \bar{\rr} & d
	\end{pmatrix},
\end{equation}
such that $0 = \pp_3\rr - \bar{\rr}\bar{\pp_3}$.  Substituting this $N_1$ into the latter condition, we find
\begin{equation}
	\begin{split}
		0 &= \hh_1 \rr + \rr \overbar{\pp_3\hh_1\pp_3^{-1}} \\
		0 &= \Re(\hh_3\rr) + d \Re(\hh_1).
	\end{split}
\end{equation}
Assuming $\rr \neq 0$ and $\hh_1 \neq 0$, take the real part of the first constraint to get $\Re(\hh_1) = 0$.  Alternatively, with $\rr = 0$, $d \neq 0$ for $N_1 \neq 0$; therefore, the second constraint would also impose $\Re(\hh_1) = 0$.  This result allows us to simply the constraints to 
\begin{equation}
		0 = \Re(\hh_1) \qquad 0 = [\hh_1, \rr\pp_3] \qquad 0 = \Re(\hh_3\rr).
\end{equation}
In fact, the second constraint above is satisfied by 
\begin{equation}
	0 = \pp_3\rr - \bar{\rr}\bar{\pp_3},
\end{equation}
so the set of constraints on this sub-branch becomes
\begin{equation}
	\mathcal{C}_{\hat{\a},\, \mathsf{2.2.i}} = \{ 0 = \Re(\hh_1), \quad 0 = \Re(\hh_3\rr), \quad 0 = \pp_3\rr - \bar{\rr}\bar{\pp_3} \}.
\end{equation}
Subject to these constraints, we have the following non-vanishing \hypertarget{N2_a_22ii}{matrices}
\begin{equation}
	\mathcal{M}_{\hat{\a},\, \mathsf{2.2.i}} = \Big\{ \mp = \begin{pmatrix} 0 & 0 \\ \pp_3 & 0 \end{pmatrix}, \quad 
	\mh = \begin{pmatrix} \hh_1 & 0 \\ \hh_3 & \pp_3 \hh_1 \pp_3^{-1} \end{pmatrix}, \quad
	N_1 = \begin{pmatrix} 0 & \rr \\ \bar{\rr} & d \end{pmatrix} \Big\}.
\end{equation}
This sub-branch consists of two collinear quaternions $\pp_3$ and $\rr$, one quaternion $\hh_3$ that is perpendicular to these two in $\Im(\mathbb{H})$, and one imaginary quaternion $\hh_1$.  In addition, there is a single real scalar, $d$.  Notice that if $\mh$ vanishes, we produce a system that is equivalent to the one found in sub-branch $\mathsf{2.1.ii}$; therefore, this sub-branch is certainly non-empty.  However to investigate the role of $\mh$, we will require at least one of its components to be non-vanishing.  To simplify $\mh$ as far as possible, let $\hh_1 = 0$.  Now we can choose either $\rr$ or $d$ to vanish while maintaining supersymmetry.  Letting $\rr = 0$, we can use the endomorphisms of $\s_{\bar{1}}$ on $\pp_3$ and $\hh_3$, and employ the scaling of $\sZ$ on $N_1$ to arrive at
\begin{equation}
	\mp = \begin{pmatrix} 0 & 0 \\ \ii& 0 \end{pmatrix}, \quad 
	\mh = \begin{pmatrix} 0 & 0 \\ \ii & 0 \end{pmatrix}, \quad
	N_1 = \begin{pmatrix} 0 & 0\\ 0 & 1 \end{pmatrix}.
\end{equation}
Thus, there exist super-extensions in this sub-branch for which $\mh \neq 0$.  Wanting to be more be a little more general, we can choose for only $\hh_3$ to vanish.  Then, using the endomorphisms in $\GL(\s_{\bar{1}})$, we can set $\rr = d \ii$ such that $\pp_3$ also lies along $\ii$.  Utilising the scaling symmetry of $\sP$ and $\sZ$ in $\GL(\s_{\bar{0}})$, we can remove the constants from the matrices $\mp$ and $N_1$ to get
\begin{equation}
	\mp = \begin{pmatrix}
		0 & 0 \\ \ii & 0
	\end{pmatrix} \quad \text{and} \quad
	N_1 = \begin{pmatrix}
		0 & \ii \\ -\ii & 1
	\end{pmatrix}.
\end{equation}
Employing the residual endomorphisms of $\s_{\bar{1}}$, we can now choose $\hh_1$ to lie along $\ii$. This change allows us to use the scaling symmetry of $\sH$ in $\GL(\s_{\bar{0}})$ such that $\mh$ becomes
\begin{equation}
	\mh = \begin{pmatrix}
		\ii & 0 \\ 0 & \ii
	\end{pmatrix}.
\end{equation}
Now, returning to the latter case, in which $\mh$ vanishes, we have only
\begin{equation}
		0 = \mp N_i - N_i \mp^\dagger \quad \text{where} \quad i \in \{0, 1\},
\end{equation}
which tells us that 
\begin{equation}
	N_0 = \begin{pmatrix}
		0 & \qq \\ \bar{\qq} & b
	\end{pmatrix} \quad \text{and} \quad N_1 = \begin{pmatrix}
		0 & \rr \\ \bar{\rr} & d
	\end{pmatrix},
\end{equation} 
where
\begin{equation}
	0 = \pp_3\qq - \bar{\qq}\bar{\pp_3} \quad \text{and} \quad 0 = \pp_3\rr - \bar{\rr}\bar{\pp_3}.
\end{equation} 
Therefore, the set of non-vanishing \hypertarget{N2_a_22ii}{matrices} is given by
\begin{equation}
	\mathcal{M}_{\hat{\a},\,\mathsf{2.2.ii}} = \Big\{ \mp = \begin{pmatrix} 0 & 0 \\ \pp_3 & 0 \end{pmatrix}, \quad 
	N_0 = \begin{pmatrix} 0 & \qq \\ \bar{\qq} & b \end{pmatrix}, \quad 
	N_1 = \begin{pmatrix} 0 & \rr \\ \bar{\rr} & d \end{pmatrix} \Big\},
\end{equation}
subject to 
\begin{equation}
	\mathcal{C}_{\hat{\a},\,\mathsf{2.2.ii}} = \{ 0 = \pp_3\qq - \bar{\qq} \bar{\pp_3} \quad \text{and} \quad 0 = \pp_3\rr - \bar{\rr}\bar{\pp_3} \}.
\end{equation}
Notice that the matrices $N_0$ and $N_1$ and the constraints on their components take the same form as the matrix $N_0$ and its constraints in sub-branch \hyperlink{N2_a_21ii}{$\mathsf{2.1.ii}$}.  However, this sub-branch is distinct.  Notice that, using the endomorphisms of $\s_{\bar{1}}$ and the conditions in $\mathcal{C}_{\hat{\a},\,\mathsf{2.2.ii}}$, we can make all the quaternions parameterising this sub-branch of $\cS$ lie along $\ii$.  The scaling symmetry of $\sP$ may then be employed to set $\pp_3 = \ii$, leaving only $b$ and $d$ unfixed.  The last of the endomorphisms of $\s_{\bar{1}}$ may set one of these parameters to zero, but not both; therefore, we cannot have $N_0 = N_1$, which would be a necessary condition for this sub-branch to be equivalent to $(\mathcal{M}_{\hat{\a},\,\mathsf{2.1.ii}}, \mathcal{C}_{\hat{\a},\,\mathsf{2.1.ii}})$.  However, we can fix all the parameters of this sub-branch.  Had we chosen $\qq = b \ii$ with the initial $\s_{\bar{1}}$ endomorphism and set $d = 0$, we could scale $\sH$ and $\sZ$ to find
\begin{equation}
	\mp = \begin{pmatrix} 0 & 0 \\ \ii& 0 \end{pmatrix}, \quad 
	N_0 = \begin{pmatrix} 0 & \ii \\ -\ii & 1 \end{pmatrix}, \quad 
	N_1 = \begin{pmatrix} 0 & \ii \\ -\ii & 0 \end{pmatrix}.
\end{equation}
Thus, this sub-branch is non-empty and we can fix all parameters in each super-extension it contains. 
\\
\paragraph{\textbf{Sub-branch 2.3}} 
Finally, with $N_1 \neq 0$ and $N_3 \neq 0$, we can substitute $\vt = (0, 1)$ into
\begin{equation}
	0 = \pi \vt \mp N_3 \vt^\dagger + \vt N_3 (\pi \vt \mp)^\dagger
\end{equation} 
to find $c = 0$.  Therefore, $N_1$ and $N_3$ are reduced to 
\begin{equation}
	N_1 = \begin{pmatrix}
	0 & \rr \\ \bar{\rr} & d
	\end{pmatrix} \quad \text{and} \quad 
	N_3 = \begin{pmatrix}
	0 & 0 \\ 0 & \bar{\rr}\bar{\pp_3} - \pp_3 \rr 
	\end{pmatrix}.
\end{equation}
Recall, the condition
\begin{equation}
	0 = \Re(\vt N_0 \vt^\dagger) \vt \mh
\end{equation}
tells us that either $(\mathsf{i})$ $N_0 = 0$ and $\mh \neq 0$, or $(\mathsf{ii})$ $N_0 \neq 0$ and $\mh = 0$.
Letting $N_0 = 0$, the final conditions for this sub-branch are
\begin{equation}
		0 = \mh N_i + N_i \mh^\dagger \quad \text{where} \quad i \in \{1, \, 3\}.
\end{equation}
From the discussion in sub-branch \hyperlink{N2_a_22i}{$\mathsf{2.2.i}$}, we know that
the $N_1$ case produces the constraints
\begin{equation} 
	0 = \Re(\hh_1), \quad 0 = [\hh_1, \rr\pp_3], \quad \text{and} \quad 0 = \Re(\hh_3 \rr) .
\end{equation}
Interestingly, the $N_3$ condition adds no new constraints to this set; therefore, we have 
\begin{equation} 
	\mathcal{C}_{\hat{\a}, \, \mathsf{2.3.i}} = \{ 0 = \Re(\hh_1), \quad 0 = [\hh_1, \rr\pp_3],
	\quad 0 = \Re(\hh_3 \rr) \}.
\end{equation}
The corresponding \hypertarget{N2_a_23i}{matrices} for this sub-branch are given by
\begin{equation}
	\mathcal{M}_{\hat{\a}, \, \mathsf{2.3.i}} = \Big\{ \mp = \begin{pmatrix} 0 & 0 \\ \pp_3 & 0 \end{pmatrix}, \quad 
	\mh = \begin{pmatrix} \hh_1 & 0 \\ \hh_3 & \pp_3 \hh_1 \pp_3^{-1} \end{pmatrix}, \quad
	N_1 = \begin{pmatrix} 0 & \rr \\ \bar{\rr} & d \end{pmatrix}, \quad
	N_3 = \begin{pmatrix} 0 & 0 \\ 0 & \bar{\rr}\bar{\pp_3} - \pp_3 \rr \end{pmatrix} \Big\}.
\end{equation}
To establish the existence of super-extensions in this sub-branch, begin by setting $\mh = 0$ and $d = 0$.  The endomorphisms of $\s_{\bar{1}}$ may be used to set $\pp_3$ to lie along $\ii$ and scale $\rr$ such that $\rr \in \Sp(1)$.  We can then utilise  the automorphisms of $\mathbb{H}$ and the scaling symmetry of $\sP$ and $\sB$ in $\GL(\s_{\bar{0}})$ to set $\rr$ and fix the parameters in $\mp$ and $N_3$.  This leaves us with a super-extension whose matrices are written
\begin{equation}
	 \mp = \begin{pmatrix} 0 & 0 \\ \ii & 0 \end{pmatrix}, \quad 
	N_1 = \begin{pmatrix} 0 & 1 + \jj \\ 1-\jj & 0 \end{pmatrix}, \quad
	N_3 = \begin{pmatrix} 0 & 0 \\ 0 & \kk \end{pmatrix}.
\end{equation}
Using this parameterisation, we can also introduce $\hh_1$.  Substituting $\pp_3 = \ii$ and $\rr = 1 + \jj$ into the constraints of $\mathcal{C}_{\hat{\a}, \, \mathsf{2.3.i}}$, we find
\begin{equation}
	\mp = \begin{pmatrix} 0 & 0 \\ \ii & 0 \end{pmatrix}, \quad 
	\mh = \begin{pmatrix} \ii - \kk & 0 \\ 0 & \ii + \kk \end{pmatrix}, \quad
	N_1 = \begin{pmatrix} 0 & 1 + \jj \\ 1-\jj & 0 \end{pmatrix}, \quad
	N_3 = \begin{pmatrix} 0 & 0 \\ 0 & \kk \end{pmatrix}.
\end{equation}
Looking to include $\hh_3$ or $d$ leads to the introduction of parameters that cannot be fixed using the basis transformations $\G \subset \GL(\s_{\bar{0}}) \times \GL(\s_{\bar{1}})$ and the constraints.
\\ \\
In the latter case, for which $N_0 \neq 0$, the only remaining condition is
\begin{equation}
		0 = \mp N_0 - N_0 \mp^\dagger,
\end{equation}
which we know from the previous sub-branches, tells us that $\pp_3$ and $\qq$
are collinear, and that $a = 0$.  Therefore, the non-vanishing
\hypertarget{N2_a_23ii}{matrices} for this sub-branch are
\begin{equation}
	\mathcal{M}_{\hat{\a},\, \mathsf{2.3.ii}} = \Big\{\mp = \begin{pmatrix} 0 & 0 \\ \pp_3 & 0 \end{pmatrix}, \quad 
	N_0 = \begin{pmatrix} 0 & \qq \\ \bar{\qq} & b \end{pmatrix}, \quad
	N_1 = \begin{pmatrix} 0 & \rr \\ \bar{\rr} & d \end{pmatrix}, \quad
	N_3 = \begin{pmatrix} 0 & 0 \\ 0 & \bar{\rr}\bar{\pp_3} - \pp_3 \rr \end{pmatrix} \Big\},
\end{equation}
and the constraints are given by
\begin{equation}
	\mathcal{C}_{\hat{\a}, \, \mathsf{2.3.ii}} = \{ 0 = \pp_3\qq - \bar{\qq}\bar{\pp_3} \}.
\end{equation}
This sub-branch of $\cS$ has 13 real parameters, being parameterised by two collinear quaternions $\pp_3$ and $\qq$, an additional quaternion $\rr$ and two real scalars, $b$ and $d$.  Letting $d = 0$ and $\qq = 0$, we can use the same transformations as in sub-branch \hyperlink{N2_a_23i}{$\mathsf{2.3.i}$} to fix 
\begin{equation}
	 \mp = \begin{pmatrix} 0 & 0 \\ \ii & 0 \end{pmatrix}, \quad 
	N_1 = \begin{pmatrix} 0 & 1 + \jj \\ 1-\jj & 0 \end{pmatrix}, \quad
	N_3 = \begin{pmatrix} 0 & 0 \\ 0 & \kk \end{pmatrix}.
\end{equation}
Subsequently employing the scaling symmetry of $\sH$ in $\GL(\s_{\bar{0}})$, we can fix $b$ such that
\begin{equation}
	N_0 = \begin{pmatrix} 0 & 0 \\ 0 & 1 \end{pmatrix}.
\end{equation}
Therefore, there are certainly super-extensions in this sub-branch.  We can introduce either $\qq$ or $d$ while continuing to fix all the parameters of the super-extension; however, attempting to include both leads to the inclusion of a parameter that we cannot fix with the constraints of $\mathcal{C}_{\hat{\a}, \, \mathsf{2.3.ii}}$ and basis transformations in $\G$.
\\
\paragraph{$\hat{\n}_-$}~\\ \\
Setting $\mu = \eta = 0$, $\lambda = 1$ and $\varepsilon = -1$, the conditions reduce to 
\begin{equation}
	\begin{split}
		0 &= \mh N_i + N_i \mh^\dagger \quad \text{where} \quad i \in \{0, 1\} \\
		0 &= \mp N_0 - N_0 \mp^\dagger \\
		0 &= \Re(\vt N_0 \vt^\dagger) \beta\\
		0 &= - \Re(\vt N_0 \vt^\dagger) \pi \\
		0 &= \Re(\vt N_0\vt^\dagger) \vt \mh
	\end{split} \qquad 
	\begin{split}
		0 &= \pi \vt \mp N_3 \vt^\dagger + \vt N_3 (\pi \vt \mp)^\dagger \\
		N_3 &= \mh N_3 + N_3 \mh^\dagger \\
		-N_3 &= \mp N_1 - N_1 \mp^\dagger. \\
	\end{split}
\end{equation}
From the third and fourth condition, we instantly get $N_0 = 0$.  Therefore, we cannot have any solutions along sub-branches satisfying case $(\mathsf{ii})$ for $\hat{\n}_-$.  We are left with 
\begin{equation} \label{eq:N2_n-_22_conditions}
	\begin{split}
		0 &= \mh N_1 + N_1 \mh^\dagger \\
		0 &= \pi \vt \mp N_3 \vt^\dagger + \vt N_3 (\pi \vt \mp)^\dagger
	\end{split} \qquad 
	\begin{split}
		N_3 &= \mh N_3 + N_3 \mh^\dagger \\
		-N_3 &= \mp N_1 - N_1 \mp^\dagger. \\
	\end{split}
\end{equation}
\paragraph{\textbf{Sub-branch 2.1}} Since $N_1$ and $N_3$ are the only possible non-vanishing
matrices encoding the $[\Q, \Q]$ bracket, we cannot have a super-extension in this branch.
\\
\paragraph{\textbf{Sub-branch 2.2}} With $N_3 = 0$, we are left with
\begin{equation} 
		0 = \mh N_1 + N_1 \mh^\dagger \quad \text{and} \quad 
		0 = \mp N_1 - N_1 \mp^\dagger.
\end{equation}
The latter condition tells us that $\pp_3$ and $\rr$ are collinear and $0 = c\pp_3$.  Since we must have $\pp_3 \neq 0$ in this branch, we have $c = 0$.  Using this result, the first condition above tells us
\begin{equation}
	\begin{split}
		0 &= \hh_1 \rr + \rr \overbar{\pp_3\hh_1\pp_3^{-1}+1} \\
		0 &= \Re(\hh_3\rr) + d ( \Re(\hh_1) +1).
	\end{split}
\end{equation}
In fact, utilising the collinearity of $\pp_3$ and $\rr$, the first of these constraints becomes
\begin{equation}
	0 = (2 \Re(\hh_1) + 1 ) \Re(\pp_3 \rr).
\end{equation}
Thus, we have
\begin{equation}
	\mathcal{C}_{\hat{\n}_-,\, \mathsf{2.2.i}} = \{ 0 = \bar{\rr}\bar{\pp_3} - \pp_3 \rr, \quad 0 = (2 \Re(\hh_1) + 1 ) \Re(\pp_3 \rr),
	\quad 0 = \Re(\hh_3\rr) + d ( \Re(\hh_1) +1) \}.
\end{equation}
The non-vanishing \hypertarget{N2_n-_22i}{matrices} in this instance are
\begin{equation}
	\mathcal{M}_{\hat{\n}_-, \, \mathsf{2.2.i}} = \Big\{\mp = \begin{pmatrix} 0 & 0 \\ \pp_3 & 0 \end{pmatrix}, \quad 
	\mh = \begin{pmatrix} \hh_1 & 0 \\ \hh_3 & \pp_3 \hh_1 \pp_3^{-1} + 1 \end{pmatrix}, \quad
	N_1 = \begin{pmatrix} 0 & \rr \\ \bar{\rr} & d \end{pmatrix} \Big\}.	
\end{equation}
Therefore, the sub-branch in $\cS$ for these super-extensions of $\hat{\n}_-$ is parameterised by two collinear quaternions $\pp_3$ and $\rr$, two quaternions encoding the action of $\sH$ on $\s_{\bar{1}}$, $\hh_1$ and $\hh_3$, and one real scalar $d$. Notice, this is the first instance in which setting some parameters to zero imposes particular values for other parameters in the extension.  
In particular, the vanishing of $\rr$ imposes $\Re(\hh_1) = -1$ by the third constraint in $\mathcal{C}_{\hat{n}_-,\, \mathsf{2.2.i}}$, since $d \neq 0$ in this instance.  However, if $\rr \neq 0$, the second constraint implies $2\Re(\hh_1) = -1$.  In the former case, we can set $\hh_3$ and the imaginary part of $\hh_1$ to zero.  Using the endomorphisms of $\s_{\bar{1}}$ to set $\pp_3 = \ii$, we can subsequently employ the scaling symmetry of $\sH$ and $\sZ$ to obtain a super-extension with matrices
\begin{equation}
	\mp = \begin{pmatrix} 0 & 0 \\ \ii & 0 \end{pmatrix}, \quad 
	\mh = \begin{pmatrix} 1 & 0 \\ 0 & 0 \end{pmatrix} \quad  \text{and} \quad 
	N_1 = \begin{pmatrix} 0 & 0 \\ 0 & 1 \end{pmatrix}.
\end{equation}
Therefore, there exist super-extensions in this sub-branch for which $\rr = 0$.  Letting $\rr \neq 0$, we may again use the endomorphisms of $\s_{\bar{1}}$ to impose that $\pp_3$ lies along $\ii$; however, due to the first constraint in $\mathcal{C}_{\hat{\n}_-, \, \mathsf{2.2.i}}$, this also means that $\rr$ lies along $\ii$. Utilising the scaling symmetry of the $\s_{\bar{0}}$ basis elements, we may write down the matrices
\begin{equation}
	\mp = \begin{pmatrix} 0 & 0 \\ \ii & 0 \end{pmatrix}, \quad 
	\mh = \begin{pmatrix} 1 & 0 \\ 0 & -1 \end{pmatrix} \quad  \text{and} \quad 
	N_1 = \begin{pmatrix} 0 & \ii \\ -\ii & 0 \end{pmatrix}.
\end{equation}
Thus, super-extensions for which $\rr \neq 0$ exist in this sub-branch.  In both cases, residual $\s_{\bar{1}}$ endomorphisms may be used to set $\hh_3$ and the imaginary part of $\hh_1$ should we choose to include them. 
\\ 
\paragraph{\textbf{Sub-branch 2.3}} Setting $N_3 \neq 0$, we must now consider 
\begin{equation}
	0 = \pi \vt \mp N_3 \vt^\dagger + \vt N_3 (\pi \vt \mp)^\dagger,
\end{equation}
which, on substituting in $\vt = (0, 1)$, tells us that $c = 0$.  Therefore, as in sub-branch \hyperlink{N2_n-_22}{2.2}, the first condition
of \eqref{eq:N2_n-_22_conditions} tells us 
\begin{equation} \label{eq:N2_n-_23_N1_conditons}
	\begin{split}
		0 &= \hh_1 \rr + \rr \overbar{\pp_3\hh_1\pp_3^{-1}+1} \\
		0 &= \Re(\hh_3\rr) + d ( \Re(\hh_1) +1).
	\end{split}
\end{equation}
However, unlike sub-branch \hyperlink{N2_n-_22i}{2.2}, $\rr$ and $\pp_3$ are not collinear
since the imaginary part of $\pp_3\rr$ makes up the only non-vanishing component of $N_3$:
\begin{equation}
	N_3 = \begin{pmatrix}
		0 & 0 \\ 0 & \bar{\rr}\bar{\pp_3} - \pp_3 \rr 
	\end{pmatrix}.
\end{equation}
Substituting this $N_3$ into its condition from the $[\bH, \Q, \Q]$ identity, we find
\begin{equation}
	(1 - 2 \Re(\hh_4)) \Im(\ll) = [ \Im(\hh_4), \Im(\ll)] ,
\end{equation}
where $\hh_4 = \pp_3 \hh_1 \pp_3^{-1} + 1$ and $\ll = \bar{\rr}\bar{\pp_3} - \pp_3 \rr $.  Since $\Im(\ll)$ is
perpendicular to $[ \Im(\hh_4), \Im(\ll)]$, both sides of this expression must vanish separately.  Substituting 
$\hh_4$ and $\ll$ into the above expressions, we find
\begin{equation}
	0 = (1 + 2 \Re(\hh_1) ) \Im(\pp_3 \rr) \quad \text{and} \quad 0 = [\hh_1, \rr\pp_3].
\end{equation}
As stated above, $\rr$ and $\pp_3$ are not collinear; therefore, the first constraint here tells us that
\begin{equation}
	2 \Re(\hh_1) = - 1. 
\end{equation}
Substituting this result into the second constraint in \eqref{eq:N2_n-_23_N1_conditons}, we find
\begin{equation}
	2 \Re(\hh_3\rr) = -d. 
\end{equation}
Putting all these results together, the constraints are
\begin{equation}
	\mathcal{C}_{\hat{\n}_-,\,\mathsf{2.3.i}} = \{ 2 \Re(\hh_1) = - 1, \quad 2 \Re(\hh_3 \rr) = -d, \quad 0 = [\hh_1, \rr\pp_3] \},
\end{equation}
for the non-vanishing \hypertarget{N2_n-_23i}{matrices}
\begin{equation}
	\mathcal{M}_{\hat{\n}_-,\,\mathsf{2.3.i}} = \Big\{ \mp = \begin{pmatrix} 0 & 0 \\ \pp_3 & 0 \end{pmatrix}, \quad 
	\mh = \begin{pmatrix} \hh_1 & 0 \\ \hh_3 & \pp_3 \hh_1 \pp_3^{-1} + 1 \end{pmatrix}, \quad
	N_1 = \begin{pmatrix} 0 & \rr \\ \bar{\rr} & d \end{pmatrix}, \quad
	N_3 = \begin{pmatrix} 0 & 0 \\ 0 &  \bar{\rr}\bar{\pp_3} - \pp_3 \rr \end{pmatrix}	 \Big\}.		
\end{equation}
Notice, the sub-branch in $\cS$ describing these super-extensions of $\hat{\n}_-$ is parameterised by four quaternions $\pp_3$, $\hh_1$, $\hh_3$ and $\rr$, and one real scalar $d$.  Wanting to establish the existence of super-extensions in this sub-branch, we can choose to set $\hh_3$, $d$, and the imaginary part of $\hh_1$ to zero.  Then, utilising the endomorphisms of $\s_{\bar{1}}$, we can impose that $\pp_3$ must lie along $\ii$ and that $\rr$ must have unit norm.  Subsequently employing $\Aut(\mathbb{H})$ to fix $\rr$, we can finally scale $\sH$, $\sZ$, $\sP$, and $\sB$ to get the super-extension
\begin{equation}
	\mp = \begin{pmatrix} 0 & 0 \\ \ii & 0 \end{pmatrix}, \quad 
	\mh = \begin{pmatrix} 1 & 0 \\ 0 & -1 \end{pmatrix}, \quad
	N_1 = \begin{pmatrix} 0 & 1 + \jj \\ 1 - \jj & 0 \end{pmatrix}, \quad
	N_3 = \begin{pmatrix} 0 & 0 \\ 0 &  \kk \end{pmatrix}.
\end{equation}
Having established that this sub-branch is not empty, we may look to introduce the components we have set to zero for this example.  Notably, we may introduce the imaginary part of $\hh_1$ while still fixing all parameters using the basis transformations $\G \subset \GL(\s_{\bar{0}}) \times \GL(\s_{\bar{1}})$.  However, the inclusion of either $\hh_3$ or $d$ will introduce parameters that cannot be fixed.
\\
\paragraph{$\hat{\n}_+$ and $\hat{\g}$}~\\ \\
Substituting $\lambda = \varepsilon = 0$, $\mu = \pm 1$ into the conditions of \eqref{eq:N2_C_branch2_constraints},\footnote{Whether 
we are in the $\hat{\n}_+$ or $\hat{\g}$ case makes no difference: the distinction between
the two is the value of $\eta$, which, if non-vanishing, would add the condition 
\begin{equation}
	0 = \Re(\vt N_0\vt^\dagger) \pi.
\end{equation} 
This condition sets $N_0 = 0$, but we already have this result from another condition.
Therefore, the super-extensions are the same for both of these generalised Bargmann algebras.}
we instantly
have $N_3 = 0$ and
\begin{equation}
	\begin{split}
		0 &= \mh N_i + N_i \mh^\dagger \quad \text{where} \quad i \in \{0, 1\} \\
		0 &= \mp N_i - N_i \mp^\dagger \quad \text{where} \quad i \in \{0, 1\} \\
		0 &= \Re(\vt N_0\vt^\dagger) \vt \mh \\
		0 &= \pm \Re(\vt N_0\vt^\dagger) \beta. \\
	\end{split}
\end{equation}
The final condition here states that $N_0 = 0$; therefore, $N_1$ is the only possible 
non-vanishing matrix of those encoding $[\Q, \Q]$.  This result tells us there will be no
sub-branch 2.1 or 2.3 for these algebras and no sub-branch satisfying case $(\mathsf{ii})$, in which $N_0 \neq 0$.  Therefore,
the conditions reduce to 
\begin{equation} \label{eq:N2_C_branch2_ng_conditions}
		0 = \mh N_1 + N_1\mh^\dagger \quad \text{and} \quad 
		0 = \mp N_1 - N_1 \mp^\dagger.
\end{equation}
Under the assumption that $\pp_3 \neq 0$ for this branch of super-extensions,
the latter condition tells us that $c = 0$ and that $\pp_3$ and $\rr$ are collinear:
\begin{equation}
	0 = \bar{\rr} \bar{\pp_3} - \pp_3 \rr.
\end{equation}
Substituting these results into the first condition, we find
\begin{equation}
	0 = \Re(\hh_1), \quad 0 = [\hh_1, \rr\pp_3] \quad \text{and} \quad 0 = \Re(\hh_3\rr).
\end{equation}
Notice that since $\pp_3$ and $\rr$ are collinear, the second constraint is instantly
satisfied.  Thus, our constraints reduce to 
\begin{equation}
	\mathcal{C}_{\hat{\n}_+ \,\text{and}\,\hat{\g}, \, \mathsf{2.2.i}} = \{ 0 = \Re(\hh_1), \quad 0 = \Re(\hh_3\rr), \quad 0 = \bar{\rr} \bar{\pp_3} - \pp_3 \rr \}.
\end{equation}
The non-vanishing \hypertarget{N2_ng_22i}{matrices} in this instance are
\begin{equation}
	\mathcal{M}_{\hat{\n}_+\,\text{and}\,\hat{\g},\,\mathsf{2.2.i}} = \Big\{ \mp = \begin{pmatrix} 0 & 0 \\ \pp_3 & 0 \end{pmatrix}, \quad 
	\mh = \begin{pmatrix} \hh_1 & 0 \\ \hh_3 & \pp_3 \hh_1 \pp_3^{-1} \end{pmatrix}, \quad
	N_1 = \begin{pmatrix} 0 & \rr \\ \bar{\rr} & d \end{pmatrix} \Big\}.	
\end{equation}
This sub-branch has identical $(\mathcal{M}, \mathcal{C})$ to sub-branch \hyperlink{N2_a_22i}{$\mathsf{2.2.i}$} for $\hat{\a}$.  Therefore, for a discussion on the existence of such super-extensions, we refer the reader to the discussion found there.
\subsubsection{Branch 3} \label{subsubsec:N2_C_branch3}
\begin{equation}
	\mp = \mz = 0 \quad \mb = \begin{pmatrix} 0 & 0 \\ \bb_3 & 0 \end{pmatrix} \quad 
	\mh = \begin{pmatrix} \hh_1 & 0 \\ \hh_3 & \bb_3 \hh_1 \bb_3^{-1} - \lambda \end{pmatrix}.
\end{equation}
Exploiting the vanishing of $\mz$ and $\mp$, we can reduce the conditions from Lemmas \ref{lem:N2_011}
and \ref{lem:N2_111}.  In particular, the vanishing of $\mp$, when substituted into the conditions from the $[\P, \Q, \Q]$ super-Jacobi identity, tells us that $N_3 = 0$ and
\begin{equation}
	\begin{split}
		0 &= \eta \Re(\vt N_0\vt^\dagger) \pi \\
		0 &= \varepsilon \Re(\vt N_0 \vt^\dagger) \pi + \tfrac12 [\pi, \vt N_2 \vt^\dagger]. \\		
	\end{split}
\end{equation}
The $[\B, \Q, \Q]$ identity then produce
\begin{equation}
	\begin{split}
		0 &= \mb N_0 - N_0 \mb^\dagger \\	
		N_4 &= \mb N_1 - N_1 \mb^\dagger \\
		0 &= \beta \vt \mb N_2 \vt^\dagger + \vt N_2 (\beta \vt \mb)^\dagger \\
		0 &= \lambda \Re(\vt N_0 \vt^\dagger) \beta + \tfrac12 [\beta, \vt N_2 \vt^\dagger] \\
		\mu \Re(\vt N_0\vt^\dagger) \beta &= \beta \vt \mb N_4 \vt^\dagger + \vt N_4 (\beta \vt \mb)^\dagger. \\
	\end{split}
\end{equation}
The conditions from the $[\bZ, \Q, \Q]$ identity are satisfied since $\mz=0$, and, lastly, the $[\bH, \Q, \Q]$ super-Jacobi identity
produces
\begin{equation}
	\begin{split}
		0 &= \mh N_i + N_i \mh^\dagger \quad \text{where} \quad i \in \{0, 1, 2\} \\
		0 &= \eta N_4\\
		\varepsilon N_4 &= \mh N_4 + N_4 \mh^\dagger.
	\end{split}
\end{equation}	
From Lemma \ref{lem:N2_111}, we get
\begin{equation}
		\Re(\vt N_0\vt^\dagger) \vt \mh  = \tfrac12 \vt N_2 \vt^\dagger \vt.
\end{equation}
As in both previous branches, the conditions 
\begin{equation}
	\begin{split}
		0 &= \lambda \Re(\vt N_0 \vt^\dagger) \beta + \tfrac12 [\beta, \vt N_2 \vt^\dagger] \\
		0 &= \varepsilon \Re(\vt N_0 \vt^\dagger) \pi + \tfrac12 [\pi, \vt N_2 \vt^\dagger], \\
	\end{split}
\end{equation}
tell us $N_2 = 0$, such that, putting everything together, we have
\begin{equation} \label{eq:N2_C_branch3_constraints}
	\begin{split}
		0 &= \eta N_4\\
		0 &= \mh N_i + N_i \mh^\dagger \quad \text{where} \quad i \in \{0, 1\} \\
		0 &= \mb N_0 - N_0 \mb^\dagger \\
		0 &= \eta \Re(\vt N_0\vt^\dagger) \pi \\
		0 &= \lambda \Re(\vt N_0 \vt^\dagger) \beta \\
		0 &= \varepsilon \Re(\vt N_0 \vt^\dagger) \pi \\
		0 &= \Re(\vt N_0\vt^\dagger) \vt \mh
	\end{split}
	\begin{split}
		\varepsilon N_4 &= \mh N_4 + N_4 \mh^\dagger. \\
		N_4 &= \mb N_1 - N_1 \mb^\dagger \\
		\mu \Re(\vt N_0\vt^\dagger) \beta &= \beta \vt \mb N_4 \vt^\dagger + \vt N_4 (\beta \vt \mb)^\dagger \\
	\end{split} \quad \forall \beta, \pi \in \Im(\mathbb{H}), \forall \vt \in \mathbb{H}.
\end{equation}
We can now use some of the conditions common to all generalised Bargmann algebras to identify possible sub-branches with which we can organise our investigations.  Substituting the $N_1$ from \eqref{eq:N2_C_H_Herm_Mats} and the $N_4$ from \eqref{eq:N2_C_H_SkewHerm_Mats} into the condition
\begin{equation}
	N_4 = \mb N_1 - N_1 \mb^\dagger ,
\end{equation} 
we can write $N_4$ in terms of the parameters in $N_1$ and $\mb$:
\begin{equation} \label{eq:case_3_N4}
	N_4 = \begin{pmatrix} 0 & -c\bar{\bb_3} \\ c\bb_3 & \bb_3 \rr  - \bar{\rr}\bar{\bb_3}\end{pmatrix}.
\end{equation}
Notice that this means $N_4$ is completely dependent on $N_1$: if $N_1 = 0$ then $N_4 = 0$.
Therefore, in general, we have the following sub-branches:
\begin{enumerate}
	\item $N_1 = 0$ and $N_4 = 0$,
	\item $N_1 \neq 0$ and $N_4 = 0$,
	\item $N_1 \neq 0$ and $N_4 \neq 0$.
\end{enumerate}
Also, as in branches 1 and 2, the condition derived from the $[\Q, \Q, \Q]$ identity tells us that either $N_0$ or $\mh$ vanishes.  We will consider both of these cases within each sub-branch, identifying them as
\begin{enumerate}[label=(\roman*)]
	\item $N_0 = 0$ and $\mh \neq 0$, and
	\item  $N_0 \neq 0$ and $\mh = 0$.
\end{enumerate}
\paragraph{$\hat{\a}$}~\\ \\
Setting $\lambda = \mu = \eta =\varepsilon = 0$, the conditions in \eqref{eq:N2_C_branch3_constraints}
reduce to 
\begin{equation} \label{eq:N2_C_branch3_a_constraints}
	\begin{split}
		0 &= \mh N_i + N_i \mh^\dagger \quad \text{where} \quad i \in \{0, 1, 4 \} \\
		0 &= \mb N_0 - N_0 \mb^\dagger \\
		0 &= \beta \vt \mb N_4 \vt^\dagger + \vt N_4 (\beta \vt \mb)^\dagger \\
		0 &= \Re(\vt N_0\vt^\dagger) \vt \mh \\
		N_4 &= \mb N_1 - N_1 \mb^\dagger. \\
	\end{split}
\end{equation}
As in branch 2, none of these conditions force the vanishing of any more $N_i$; therefore,
super-extensions may be found in each of the sub-branches.  In fact, because of the symmetry
of the generators $\sB$ and $\sP$ in this generalised Bargmann algebra, we may use
automorphisms to transform the above conditions into those
in \eqref{eq:N2_C_branch2_a_conditions1} and \eqref{eq:N2_C_branch2_a_conditions2}, which describe
the super-extensions of $\hat{\a}$ in branch 2.  More explicitly, substitute the transformation with matrices
\begin{equation}
	A = \begin{pmatrix}
		1 & 0 \\ 0 & 1
	\end{pmatrix}, \quad 
	C = \begin{pmatrix}
		0 & -1 \\ 1 & 0
	\end{pmatrix}, \quad \text{and} \quad
	M = \begin{pmatrix}
		1 & 0 \\ 0 & 1
	\end{pmatrix},
\end{equation}
and the quaternion $\uu = 1$, into \eqref{eq:N2_S_basis_transformations}.  Putting the transformed matrices
into the conditions of \eqref{eq:N2_C_branch3_a_constraints}, we recover the conditions of 
\eqref{eq:N2_C_branch2_a_conditions1} and \eqref{eq:N2_C_branch2_a_conditions2}.
Therefore, all the super-extensions of $\hat{\a}$ in this branch are equivalent to the super-extensions
of branch 2.  Thus, for this particular generalised Bargmann algebra, this branch produces no
new super-extensions.
\\
\paragraph{$\hat{\n}_-$}~\\ \\
Setting $\mu = \eta = 0$, $\lambda = 1$ and $\varepsilon = -1$, the conditions of \eqref{eq:N2_C_branch3_constraints} become
\begin{equation} \label{eq:case_3_constraints}
	\begin{split}
		0 &= \mh N_i + N_i \mh^\dagger \quad \text{where} \quad i \in \{0, 1\} \\
		0 &= \mb N_0 - N_0 \mb^\dagger \\
		0 &= \Re(\vt N_0 \vt^\dagger) \beta \\
		0 &= - \Re(\vt N_0 \vt^\dagger) \pi \\
		0 &= \Re(\vt N_0\vt^\dagger) \vt \mh
	\end{split}
	\begin{split}
		- N_4 &= \mh N_4 + N_4 \mh^\dagger. \\
		N_4 &= \mb N_1 - N_1 \mb^\dagger \\
		0 &= \beta \vt \mb N_4 \vt^\dagger + \vt N_4 (\beta \vt \mb)^\dagger. \\
	\end{split}
\end{equation}
The conditions 
\begin{equation}
		0 = \Re(\vt N_0 \vt^\dagger) \beta \quad \text{and} \quad 0 = - \Re(\vt N_0 \vt^\dagger) \pi
\end{equation}
tell us that $N_0$ must vanish, leaving only
\begin{equation} \label{eq:N2_C_n-_constraints}
	\begin{split}
		0 &= \mh N_1 + N_1 \mh^\dagger \\
		0 &= \beta \vt \mb N_4 \vt^\dagger + \vt N_4 (\beta \vt \mb)^\dagger
	\end{split}\quad 
	\begin{split}
		- N_4 &= \mh N_4 + N_4 \mh^\dagger \\
		N_4 &= \mb N_1 - N_1 \mb^\dagger. \\
	\end{split}
\end{equation}
Notice, this result tells us that we cannot have any sub-branches satisfying case $(\mathsf{ii})$; therefore, all sub-branches $(\mathcal{M}, \mathcal{C})$ discussed below will have a subscript ending in $\mathsf{i}$.  Like the $\hat{\a}$ case,
this generalised Bargmann algebra allows for an automorphism which transforms the conditions for this branch into
the conditions for branch 2.  However, in this instance, this branch will produce some distinct super-extensions.  This result
is a consequence of the parameters $\varepsilon$ and $\lambda$ and their appearance in $\mh$.  In branch 2, the 
matrix $\mh$ is written as
\begin{equation}
	\mh = \begin{pmatrix}
		\hh_1 & 0 \\ \hh_3 & \pp_3\hh_1\pp_3^{-1} - \varepsilon
	\end{pmatrix},
\end{equation}
and in this branch, it is written
\begin{equation}
	\mh = \begin{pmatrix}
		\hh_1 & 0 \\ \hh_3 & \pp_3\hh_1\pp_3^{-1} - \lambda
	\end{pmatrix}.
\end{equation}
Since $\hat{\n}_-$ has $\varepsilon = -1$ and $\lambda = 1$, this matrix differs in these branches, if only
be a sign.  Thus, although the investigations into the super-extensions of $\hat{\n}_-$ in this branch will be very similar to 
those in the previous branch, we will give a partial presentation of them here to demonstrate any consequences of this
change in sign.  In particular, we will omit the discussions on the existence of super-extensions and parameter fixing
as these require only trivial adjustments from those found in branch 2.
\\
\paragraph{\textbf{Sub-branch 3.1}} As $N_0 = 0$, we cannot have both $N_1$ and $N_4$ vanish; therefore,
there is no super-extension in this sub-branch.
\\
\paragraph{\textbf{Sub-branch 3.2}} Letting $N_1 \neq 0$ and $N_4 = 0$, we are left with only the conditions
\begin{equation}
		0 = \mh N_1 + N_1 \mh^\dagger \quad \text{and} \quad 
		0 = \mb N_1 - N_1 \mb^\dagger.
\end{equation}
The second condition above tells us that 
\begin{equation}
	0 = c\bb_3 \quad \text{and} \quad 0 = \bb_3 \rr  - \bar{\rr}\bar{\bb_3}.
\end{equation}
As $\bb_3 \neq 0$ by assumption, $c = 0$.  Substituting this result into the first
condition above, we find 
\begin{equation}
	\begin{split}
		0 &= \hh_1 \rr  + \rr \overbar{\bb_3\hh_1\bb_3^{-1} - 1} \\
		0 &= \Re(\hh_3\rr) + d ( \Re(\hh_1) - 1 ).
	\end{split}
\end{equation}
Using the collinearity of $\bb_3$ and $\rr$, the first of these constraints tells us that
\begin{equation}
	0 = (2 \Re(\hh_1) - 1) \Re(\bb_3\rr).
\end{equation}
Therefore, the constraints in this instance are given by
\begin{equation}
	\mathcal{C}_{\hat{\n}_-,\,\mathsf{3.2.i}} = \{ 0 = \bb_3 \rr  - \bar{\rr}\bar{\bb_3}, \quad
	0 = (2 \Re(\hh_1) - 1) \Re(\bb_3\rr), \quad 0 = \Re(\hh_3\rr) + d ( \Re(\hh_1) - 1 ) \}.
\end{equation}
The non-vanishing \hypertarget{N2_n-_32i}{matrices} in this instance are
\begin{equation}
	\mathcal{M}_{\hat{\n}_-,\,\mathsf{3.2.i}} = \Big\{ \mb = \begin{pmatrix} 0 & 0 \\ \bb_3 & 0 \end{pmatrix}, \quad 
	\mh = \begin{pmatrix} \hh_1 & 0 \\ \hh_3 & \bb_3\hh_1\bb_3^{-1} -1 \end{pmatrix}, \quad
	N_1 = \begin{pmatrix} 0 & \rr \\ \bar{\rr} & d \end{pmatrix} \Big\}.
\end{equation}
This sub-branch of $\cS$ is parameterised by two collinear quaternions $\bb_3$ and $\rr$, two quaternions encoding the action of $\sH$ on $\s_{\bar{1}}$, $\hh_1$ and $\hh_3$, and one real scalar $d$.  Notice that the real component of $\hh_1$ varies depending on whether $\rr$ vanishes.  Together with the super-extensions in sub-branch \hyperlink{N2_n-_22i}{$\mathsf{2.2.i}$} for $\hat{\n}_-$, these are the only super-extensions that demonstrate this type of dependency.   If $\rr = 0$, the first two constraints of $\mathcal{C}_{\hat{\n}_-,\,\mathsf{3.2.i}}$ are trivial, and the third condition tell us that $\Re(\hh_1) = 1$, since $d \neq 0$ for $N_1 \neq 0$.  However, if $\rr \neq 0$, the second constraint requires $2 \Re(\hh_1) = 1$.  In this instance, the third constraint then becomes $ 2 \Re(\hh_3 \rr) = d$.  As the matrices and conditions for this sub-branch are so similar to those in \hyperlink{N2_n-_22i}{$\mathsf{2.2.i}$}, we refer the reader the discussion on existence of super-extensions and parameter fixing presented there.  
\\
\paragraph{\textbf{Sub-branch 3.3}} Finally, let $N_1 \neq 0$ and $N_4 \neq 0$.  
The condition 
\begin{equation}
	0 = \beta \vt \mb N_4 \vt^\dagger + \vt N_4 (\beta \vt \mb)^\dagger
\end{equation}
 imposes $c = 0$, such that
\begin{equation} \label{eq:case3_n-_Ns}
	N_1 = \begin{pmatrix} 0 & \rr \\ \bar{\rr} & d \end{pmatrix} \quad \text{and} \quad
	N_4 = \begin{pmatrix} 0 & 0 \\ 0 & \bb_3\rr - \bar{\rr} \bar{\bb_3}  \end{pmatrix}.
\end{equation}
This result reduces the conditions in \eqref{eq:N2_C_n-_constraints} to
\begin{equation} \label{eq:N2_C_n-_subbranch33_constraints}
	\begin{split}
		0 &= \mh N_1 + N_1 \mh^\dagger \\
		- N_4 &= \mh N_4 + N_4 \mh^\dagger. \\
	\end{split}
\end{equation}
Substituting the $N_4$ from \eqref{eq:case3_n-_Ns} into the second condition above, we have
\begin{equation}
	- \ll = \hh_4 \ll +\ll \bar{\hh_4},
\end{equation}
where $\ll = \bb_3 \rr-\bar{\rr}\bar{\bb_3}$ and $\hh_4 = \bb_3\hh_1\bb_3^{-1} - 1$.  We can rewrite
this condition as
\begin{equation}
	(1+2 \Re(\hh_4) ) \ll = [\ll, \hh_4].
\end{equation}
Notice that the R.H.S. of this expression must lie in $\Im(\mathbb{H})$ and be orthogonal
to $\ll$, which is imaginary by construction.  Therefore, both sides of this expression must vanish independently:
\begin{equation}
	0 = (1 + 2 \Re(\hh_4)) \ll \qquad 0 = [\ll, \hh_4].
\end{equation}
Substituting $\ll$ and $\hh_4$ into these constraints, we find 
\begin{equation} \label{eq:case3_n-_N4_condition}
	0 = (2 \Re(\hh_1) - 1) (\bb_3 \rr - \bar{\rr}\bar{\bb_3}) \quad \text{and} \quad 0 = [\hh_1, \rr\bb_3],
\end{equation}
respectively.  For $N_4$ to not vanish,
we must have $\Im(\bb_3\rr) \neq 0$, so, by the first constraint above, we need $2 \Re(\hh_1) = 1$.
The first condition in \eqref{eq:N2_C_n-_subbranch33_constraints} produces the same constraints as
in sub-branch \hyperlink{N2_n-_32i}{$\mathsf{3.2}$}; namely,
\begin{equation}
	\begin{split}
		0 &= \hh_1 \rr  + \rr \overbar{\bb_3\hh_1\bb_3^{-1} - 1} \\
		0 &= \Re(\hh_3\rr) + d ( \Re(\hh_1) - 1 ).
	\end{split}
\end{equation}
Notice that the requirement of setting $2 \Re(\hh_1) = 1$ makes the second constraint here $2 \Re(\hh_3\rr) = d$, and
that the first constraint is equivalent to $0 = [\hh_1, \rr\bb_3]$.  Therefore, the constraints on this sub-branch are given by
\begin{equation}
		\mathcal{C}_{\hat{\n}_-,\,\mathsf{3.3.i}} = \{ d = 2 \Re(\hh_3\rr), \quad 1 = 2 \Re(\hh_1) \quad \text{and} \quad 0 = [\hh_1, \rr\bb_3] \},
\end{equation}
and the non-vanishing \hypertarget{N2_n-_33i}{matrices} are
\begin{equation}
	\mathcal{M}_{\hat{\n}_-,\,\mathsf{3.3.i}} = \Big\{ \mb = \begin{pmatrix} 0 & 0 \\ \bb_3 & 0 \end{pmatrix}, \quad 
	\mh = \begin{pmatrix} \hh_1 & 0 \\ \hh_3 & \bb_3\hh_1\bb_3^{-1} -1 \end{pmatrix}, \quad
	N_1 = \begin{pmatrix} 0 & \rr \\ \bar{\rr} & d \end{pmatrix}, \quad
	N_4 = \begin{pmatrix} 0 & 0 \\ 0 & \bb_3\rr - \bar{\rr} \bar{\bb_3} \end{pmatrix}	\Big\}.
\end{equation}
For the discussion on existence of super-extensions and how to fix the parameters of the matrices describing this sub-branch of $\cS$, we refer the reader to sub-branch \hyperlink{N2_n-_23i}{$\mathsf{2.3.i}$}.  The application of the discussion to the present case requires only minor adjustments.
\\
\paragraph{$\hat{\n}_+$}~\\ \\
Substituting $\lambda = \varepsilon = 0$, $\mu = 1$, and $\eta = -1$ into the results for
Lemmas \ref{lem:N2_011} and \ref{lem:N2_111}, we find
\begin{equation}
	\begin{split}
		0 &= - N_4\\
		0 &= \mh N_i + N_i \mh^\dagger \quad \text{where} \quad i \in \{0, 1, 4\} \\
		0 &= \mb N_0 - N_0 \mb^\dagger \\
		0 &= - \Re(\vt N_0\vt^\dagger) \pi \\
		0 &= \Re(\vt N_0\vt^\dagger) \vt \mh
	\end{split}
	\begin{split}
		N_4 &= \mb N_1 - N_1 \mb^\dagger \\
		\Re(\vt N_0\vt^\dagger) \beta &= \beta \vt \mb N_4 \vt^\dagger + \vt N_4 (\beta \vt \mb)^\dagger. \\
	\end{split}
\end{equation}
Therefore, $N_4$ vanishes, and $N_0$ vanishes by $0 = - \Re(\vt N_0\vt^\dagger) \pi$.
This leaves us with
\begin{equation}
		0 = \mh N_1 + N_1 \mh^\dagger \quad \text{and} \quad 
		0 = \mb N_1 - N_1 \mb^\dagger.
\end{equation}
Notice that these conditions are similar to those of \eqref{eq:N2_C_branch2_ng_conditions}, which describe
the super-extensions of $\hat{\n}_+$ in branch 2.  In fact, we can utilise the automorphisms of $\hat{\n}_+$
to transform the above conditions into those in \eqref{eq:N2_C_branch2_ng_conditions}.  Unlike the $\hat{\n}_-$
case, since $\hat{\n}_+$ has vanishing $\varepsilon$ and $\lambda$, there is no discrepancy between
the transformed matrices and those of branch 2; therefore, the super-extensions of $\hat{\n}_+$ in branches 2 and 3
are equivalent.  Thus, we have no new super-extensions here.
\\
\paragraph{$\hat{\g}$}~\\ \\
Substituting $\lambda = \eta = \varepsilon = 0$ and $\mu = -1$ into \eqref{eq:N2_C_branch3_constraints},
we have
\begin{equation} \label{eq:N2_C_branch3_g_constraints}
	\begin{split}
		0 &= \mh N_i + N_i \mh^\dagger \quad \text{where} \quad i \in \{0, 1, 4\} \\
		0 &= \mb N_0 - N_0 \mb^\dagger \\
		0 &= \Re(\vt N_0\vt^\dagger) \vt \mh
	\end{split}
	\begin{split}
		N_4 &= \mb N_1 - N_1 \mb^\dagger \\
		- \Re(\vt N_0\vt^\dagger) \beta &= \beta \vt \mb N_4 \vt^\dagger + \vt N_4 (\beta \vt \mb)^\dagger. \\
	\end{split}
\end{equation}
With these conditions, we can now investigate the three sub-branches.
\\
\paragraph{\textbf{Sub-branch 3.1}} We cannot have $N_1 = N_4 = 0$, since the vanishing
of $N_4$ means $N_0 =0$ through
\begin{equation}
	- \Re(\vt N_0\vt^\dagger) \beta = \beta \vt \mb N_4 \vt^\dagger + \vt N_4 (\beta \vt \mb)^\dagger.
\end{equation}
This would cause all $N_i$ to vanish such that $[\Q, \Q] = 0$. Therefore, 
there is no super-extension in this sub-branch.
\\
\paragraph{\textbf{Sub-branch 3.2}} With only $N_1 \neq 0$, the conditions reduce to
\begin{equation}
		0 = \mh N_1 + N_1 \mh^\dagger \quad \text{and} \quad
		0 = \mb N_1 - N_1 \mb^\dagger.
\end{equation}
Notice that this is the same set of conditions as the $\hat{\n}_+$ case above.  Therefore,
we may expect the analysis for this generalised Bargmann algebra to be analogous. However,
there is a very important distinction.  In the $\hat{\n}_+$ case, we were able to use the automorphisms
to transform the conditions into those of sub-branch 2.2.  This automorphism is not permitted by
the generalised Bargmann algebra $\hat{\g}$.  Therefore, although the analysis will be the same
\textit{mutatis mutandis} as that of sub-branch 2.2, the resulting super-extensions will be distinct.
\\ \\
Now, since $N_1$ is the only possible non-vanishing matrix in the $[\Q, \Q]$ bracket, it must have non-zero
components.  The latter condition above tells us that $c=0$ and $\bb_3$ and $\rr$
are collinear quaternions, while the former condition imposes
\begin{equation}
	\begin{split}
		0 &= \hh_1 \rr + \rr \overbar{\bb_3\hh_1\bb_3^{-1}} \\
		0 &= \Re(\hh_3\rr) + d \Re(\hh_1).
	\end{split}
\end{equation}
Notice that if $\rr = 0$, we need $d \neq 0$ for the existence of a super-extension; therefore, the
final constraint above would impose $\Re(\hh_1) = 0$.  Similarly, if $\rr \neq 0$, the first constraint
would also enforce $\Re(\hh_1) = 0$.  Thus, in all super-extensions, we require $\Re(\hh_1) = 0$.
Using this result, these two constraints simplify to 
\begin{equation}
		0 = [\hh_1, \rr\bb_3] \quad \text{and} \quad 0 = \Re(\hh_3\rr).
\end{equation}
However, since $\bb_3$ and $\rr$ are collinear and it is only the imaginary part
of $\rr\bb_3$ that will contribute to $[\hh_1, \rr\bb_3]$, the first of these constraints
is already satisfied.  Therefore, the final set of constraints on this sub-branch is 
\begin{equation}
	\mathcal{C}_{\hat{\g},\,\mathsf{3.2.i}} = \{ 0 = \bb_3\rr - \bar{\rr} \bar{\bb_3}, \quad 0 = \Re(\hh_1), 
		\quad 0 = \Re(\hh_3\rr) \}.
\end{equation}
Subject to these constraints, we have non-vanishing \hypertarget{N2_g_32i}{matrices} are 
\begin{equation}
	\mathcal{M}_{\hat{\g},\,\mathsf{3.2.i}} = \Big\{ \mb = \begin{pmatrix} 0 & 0 \\ \bb_3 & 0 \end{pmatrix}, \quad 
	\mh = \begin{pmatrix} \hh_1 & 0 \\ \hh_3 & \bb_3 \hh_1 \bb_3^{-1} \end{pmatrix}, \quad
	N_1 = \begin{pmatrix}  0 & \rr \\ \bar{\rr} & d \end{pmatrix} \Big\}.
\end{equation}
Since this $(\mathcal{M}, \mathcal{C})$ is analogous to the one found in branch \hyperlink{N2_ng_22i}{$\mathsf{2}$} 
for $\hat{n}_+$ and $\hat{\g}$, we will omit the discussion on existence of super-extensions and
parameter fixing.  
\\
\paragraph{\textbf{Sub-branch 3.3}} Finally, with $N_4 \neq 0$, we can think of setting $N_0 \neq 0$
and $\mh = 0$.  But first, try setting $N_0 = 0$ to allow $\mh \neq 0$.  The conditions become
\begin{equation}
	\begin{split}
		0 &= \mh N_i + N_i \mh^\dagger \quad \text{where} \quad i \in \{1, 4\} \\
		N_4 &= \mb N_1 - N_1 \mb^\dagger \\
		0 &= \beta \vt \mb N_4 \vt^\dagger + \vt N_4 (\beta \vt \mb)^\dagger. \\
	\end{split}
\end{equation}
Notice that the second condition above allows us to write $N_4$ in terms of $\mb$ and $N_1$:
\begin{equation}
	N_4 = \begin{pmatrix} 0 & -c\bar{\bb_3} \\ c\bb_3 & \bb_3 \rr  - \bar{\rr}\bar{\bb_3}\end{pmatrix}.
\end{equation}
The third condition then imposes $c = 0$, since $\bb_3 \neq 0$, leaving us with
\begin{equation}
	N_1 = \begin{pmatrix}
		0 & \rr \\ \bar{\rr} & d
	\end{pmatrix} \quad \text{and} \quad 
	N_4 = \begin{pmatrix}
		0 & 0 \\ 0 & \bb_3 \rr - \bar{\rr} \bar{\bb_3}
	\end{pmatrix},
\end{equation}
Using these matrices in the final conditions,
\begin{equation}
			0 = \mh N_i + N_i \mh^\dagger \quad \text{where} \quad i \in \{1, 4\},
\end{equation}
produces 
\begin{equation}
 	0 = [\hh_1, \rr\bb_3],
\end{equation}
when $i =4$, and, when $i = 1$, we obtain
\begin{equation}
	\begin{split}
		0 &= \hh_1 \rr + \rr \overbar{\bb_3\hh_1\bb_3^{-1}} \\
		0 &= \Re(\hh_3\rr) + d \Re(\hh_1).
	\end{split}
\end{equation}
Since $\rr \neq 0$ for $N_4 \neq 0$, the first condition here states that $\Re(\hh_1) = 0$.
Therefore, the constraints on the parameters of this super-extension are given by
\begin{equation}
	\mathcal{C}_{\hat{\g},\,\mathsf{3.3.i}} = \{ 0 = \Re(\hh_1), \quad  0 = [\hh_1, \rr\bb_3], \quad
	 	0 = \Re(\hh_3\rr) \}.
\end{equation}
The \hypertarget{N2_g_33i}{non-vanishing} matrices associated with this sub-branch are
\begin{equation}
	\mathcal{M}_{\hat{\g},\,\mathsf{3.3.i}} = \Big\{ \mb = \begin{pmatrix}
		0 & 0 \\ \bb_3 & 0
	\end{pmatrix}, \quad
	\mh = \begin{pmatrix}
		\hh_1 & 0 \\ \hh_3 & \bb_3 \hh_1\bb_3^{-1}
	\end{pmatrix}, \quad
	N_1 = \begin{pmatrix}
		0 & \rr \\ \bar{\rr} & d
	\end{pmatrix}, \quad
	N_4 = \begin{pmatrix}
		0 & 0 \\ 0 & \bb_3 \rr - \bar{\rr} \bar{\bb_3}
	\end{pmatrix} \Big\}.
\end{equation}
The $(\mathcal{M}, \mathcal{C})$ of this sub-branch is the same \textit{mutatis mutandis} as that of sub-branch \hyperlink{N2_a_23i}{$\mathsf{2.3.i}$} for $\hat{\a}$; therefore, we refer the reader to the discussion found there on existence of super-extensions and parameter fixing. 
\\ \\
Finally, let $N_0 \neq 0$ such that $\mh = 0$.  The conditions remaining from \eqref{eq:N2_C_branch3_g_constraints} are
\begin{equation}
	\begin{split}
		0 &= \mb N_0 - N_0 \mb^\dagger \\
		N_4 &= \mb N_1 - N_1 \mb^\dagger \\
		- \Re(\vt N_0\vt^\dagger) \beta &= \beta \vt \mb N_4 \vt^\dagger + \vt N_4 (\beta \vt \mb)^\dagger. \\
	\end{split}
\end{equation}
We know how the second condition acts from the discussion at the beginning of this branch.  The first of these conditions tells us
\begin{equation}
	0 = a \quad \text{and} \quad 0 = \bb_3\qq - \bar{\qq}\bar{\bb_3},
\end{equation}
and the third, substituting in $\vt = (0, 1)$, produces
\begin{equation}
	- b = 2 c |\bb_3|^2.
\end{equation}
Now substituting $\vt = (1, \ss)$ into the third condition, we find
\begin{equation}
	- 2 \Re(\ss\bar{\qq}) - b |\ss|^2 = 2c |\ss|^2 |\bb_3|^2.
\end{equation}
Therefore, using the previous result and letting $\ss = 1$, $\ss = \ii$, $\ss = \jj$
and $\ss = \kk$, we see that all components of $\qq$ must vanish.  We thus \hypertarget{N2_g_33ii}{have} non-vanishing matrices
\begin{equation}
	\mathcal{M}_{\hat{\g},\,\mathsf{3.3.ii}} = \Big\{ \mb = \begin{pmatrix}
		0 & 0 \\ \bb_3 & 0
	\end{pmatrix}, \quad
	N_0 = \begin{pmatrix}
		0 & 0 \\ 0 & -2c |\bb_3|^2
	\end{pmatrix}, \quad
	N_1 = \begin{pmatrix}
		c & \rr \\ \bar{\rr} & d
	\end{pmatrix}, \quad
	N_4 = \begin{pmatrix}
		0 & -c\bar{\bb_3} \\ c\bb_3 & \bb_3 \rr - \bar{\rr} \bar{\bb_3}
	\end{pmatrix} \Big\}.
\end{equation}
Interestingly, there are no additional constraints to the parameters of this sub-branch; therefore, $\mathcal{C}_{\hat{\g},\,\mathsf{3.3.ii}}$ is empty. Notice the sub-branch of $\cS$ for this type of super-extension is parameterised by two quaternions $\bb_3$ and $\rr$, and two real scalars $c$ and $d$.  To demonstrate that this sub-branch is not empty, we begin by setting both $\rr$ and $d$ to zero.  This choice allows us to utilise the endomorphisms of $\s_{\bar{1}}$ to set $\bb_3 = \ii$ and $c=1$.  Employing the scaling symmetry of the basis elements, we arrive at
\begin{equation} \label{eq:notation_example}
	 \mb = \begin{pmatrix}
		0 & 0 \\ \ii & 0
	\end{pmatrix}, \quad
	N_0 = \begin{pmatrix}
		0 & 0 \\ 0 & 1
	\end{pmatrix}, \quad
	N_1 = \begin{pmatrix}
		1 & 0 \\ 0 & 0
	\end{pmatrix}, \quad
	N_4 = \begin{pmatrix}
		0 & \ii \\ \ii & 0
	\end{pmatrix}.
\end{equation}
We may now look to introduce $\rr$ and $d$.  Again using the endomorphisms of $\s_{\bar{1}}$, we can impose that $\bb_3$ must lie along $\ii$, set $|\rr|^2 = 1$, and choose $\sqrt{2} c =1$.  This choice for $\rr$ imposes that $\rr \in \Sp(1)$, and we may utilise $\Aut(\mathbb{H})$ to fix $\sqrt{2} \rr = 1 + \ii$.  Having chosen $\rr \neq 0$, we can always employ the residual endomorphisms of $\s_{\bar{1}}$ to set $d = 0$.  Using the only remaining symmetry, the scaling of $\sH$, $\sZ$, $\sB$, and $\sP$, we find
\begin{equation}
	 \mb = \begin{pmatrix}
		0 & 0 \\ \sqrt{2} \ii & 0
	\end{pmatrix}, \quad
	N_0 = \begin{pmatrix}
		0 & 0 \\ 0 & 1
	\end{pmatrix}, \quad
	N_1 = \begin{pmatrix}
		1 & 1+\ii \\ 1-\ii & 0
	\end{pmatrix}, \quad
	N_4 = \begin{pmatrix}
		0 & \ii \\ \ii & 2\ii
	\end{pmatrix}.
\end{equation}
\subsubsection{Branch 4} \label{subsubsec:N2_C_branch4}
\begin{equation}
	\mz = 0 \quad \mh = \begin{pmatrix} \hh_1 & 0 \\ \hh_3 & \hh_4 \end{pmatrix} \quad 
	\mb = \begin{pmatrix} 0 & 0 \\ \bb_3 & 0 \end{pmatrix} \quad 
	\mp = \begin{pmatrix} 0 & 0 \\ \pp_3 & 0 \end{pmatrix},
\end{equation}
subject to
\begin{equation} \label{eq:N2_branch4_condition}
	[\uu, \hh_1] = - \mu \uu^2 + (\lambda - \varepsilon) \uu + \eta \quad \text{or} \quad [\vv, \hh_1] = \eta \vv^2 + (\lambda - \varepsilon) \vv + \mu,
\end{equation}
where $0 \neq \uu = \bb_3^{-1} \pp_3$ and $0 \neq \vv = \pp_3^{-1} \bb_3$.
Recall, we keep both of these constraints as, depending on the generalised Bargmann algebra under investigation,
one of them will prove more useful than the other.  
We still need to determine the generalised
Bargmann algebras for which this branch could provide a super-extension.  Therefore, we will consider
each algebra in turn, and analyse those for which the above constraints may hold.
\\
\paragraph{$\hat{\a}$}~\\ \\ Setting $\lambda = \mu = \eta = \varepsilon = 0$ in 
\eqref{eq:N2_branch4_condition}, we could still get a super-extension, as long
as we impose
\begin{equation}
	0 = [\uu, \hh_1].
\end{equation}
Throughout this section, we will choose to write parameters in terms of $\bb_3$; therefore, 
we write $\pp_3 = \bb_3\uu$ and $\hh_4 = \bb_3 \hh_1 \bb_3^{-1}$, where
$\uu \in \mathbb{H}$.  Notice that the significance of $\uu$ is only manifest when 
$\hh_1 \neq 0$: when $\hh_1$ vanishes, we are simply replacing $\pp_3$ with $\uu$.  However,
since $\uu$ will be important is several instances, we will always use this notation.
\\ \\
Since neither $\mb$ nor $\mp$ vanish, there are no immediate results as in the three previous
branches: all the conditions of Lemmas \ref{lem:N2_011} and \ref{lem:N2_111} must be
taken into consideration.  However, as with branches 2 and 3, we can organise our 
investigations based on dependencies.  In particular, the conditions
\begin{equation}
	\begin{split}
			N_4 &= \mb N_1 - N_1 \mb^\dagger \\
			-N_3 &= \mp N_1 - N_1 \mp^\dagger \\
			\tfrac12 [\beta, \vt N_2 \vt^\dagger] &= \beta \vt \mb N_3 \vt^\dagger
			+ \vt N_3 (\beta \vt \mb)^\dagger \\	
			\tfrac12 [\pi, \vt N_2 \vt^\dagger] &= \pi \vt \mp N_4 \vt^\dagger
			+ \vt N_4 (\pi \vt \mp)^\dagger	,
	\end{split}
\end{equation}
show us that if $N_1$ vanishes, so must $N_2, N_3$, and $N_4$.  Additionally, the vanishing of
either $N_3$ or $N_4$ means we must have $N_2 = 0$.  Therefore, we can divide
our investigations into the following sub-branches.
\begin{enumerate}
	\item $N_1 = N_2 = N_3 = N_4 = 0$
	\item $N_1 \neq 0$ and $N_2 = N_3 = N_4 = 0$
	\item $N_1 \neq 0$, $N_3 \neq 0$, and $N_2 = N_4 = 0$
	\item $N_1 \neq 0$, $N_4 \neq 0$, and $N_2 = N_3 = 0$
	\item $N_1 \neq 0$, $N_3 \neq 0$, $N_4 \neq 0$, and $N_2 = 0$
	\item $N_1 \neq 0$, $N_2 \neq 0$, $N_3 \neq 0$, and $N_4 \neq 0$.
\end{enumerate}
Unlike branches 1, 2 and 3, the $[\Q, \Q, \Q]$ super-Jacobi identity will not always result in the cases
$(\mathsf{i})$, in which $N_0 = 0$ and $\mh \neq 0$, or $(\mathsf{ii})$, in which $N_0 \neq 0$
and $\mh = 0$.  There are instances in which both $N_0$ and $\mh$ may not vanish.  These
cases, will be labelled $(\mathsf{iii})$.  
\\
\paragraph{\textbf{Sub-branch 4.1}} With only $N_0$ left available, it cannot vanish for a supersymmetric
extension to exist.
Therefore, the $[\Q, \Q, \Q]$ identity,
\begin{equation}
	\Re(\vt N_0 \vt^\dagger) \vt \mh = 0,
\end{equation}
tells us we must have $\mh = 0$.  The remaining conditions are then
\begin{equation}
		0 = \mb N_0 - N_0 \mb^\dagger \quad \text{and} \quad 
		0 = \mp N_0 - N_0 \mp^\dagger,
\end{equation}
which tell us
\begin{equation}
	0 = a, \quad 0 = \bb_3 \qq - \bar{\qq} \bar{\bb_3} \quad \text{and} \quad
	0 = \bb_3\uu \qq - \bar{\qq} \bar{\uu} \bar{\bb_3}.
\end{equation}
This sub-branch thus has non-vanishing \hypertarget{N2_a_41}{matrices}
\begin{equation}
	\mb = \begin{pmatrix}
		0 & 0 \\ \bb_3 & 0
	\end{pmatrix}, \quad
	\mp = \begin{pmatrix}
		0 & 0 \\ \bb_3\uu & 0
	\end{pmatrix}, \quad
	N_0 = \begin{pmatrix}
		0 & \qq \\ \bar{\qq} & b
	\end{pmatrix},
\end{equation}
subject to the constraints 
\begin{equation}
	 0 = \bb_3 \qq - \bar{\qq} \bar{\bb_3}, \quad 0 = \bb_3\uu \qq - \bar{\qq} \bar{\uu}\bar{\bb_3}.
\end{equation}
Notice that these matrices and constraints are very similar to $(\mathcal{M}_{\hat{\a},\,\mathsf{2.1.ii}}, \mathcal{C}_{\hat{\a},\,\mathsf{2.1.ii}})$.  In fact, employing the automorphisms of $\hat{\a}$, we can show that the above system is equivalent to sub-branch \hyperlink{N2_a_21ii}{$\mathsf{2.1.ii}$}.  Using the endomorphisms of $\s_{\bar{1}}$ and the constraints above, we can set $\bb_3$, $\bb_3\uu$ and $\qq$ to lie along $\ii$, and set $b = 0$.  In particular, this means that $\uu \in \mathbb{R}$.  Scaling $\sB$, $\sP$, and $\sH$, we find the matrices
\begin{equation}
 	\mb = \begin{pmatrix}
		0 & 0 \\ \ii & 0
	\end{pmatrix}, \quad
	\mp = \begin{pmatrix}
		0 & 0 \\ \ii & 0
	\end{pmatrix}, \quad
	N_0 = \begin{pmatrix}
		0 & \ii \\ -\ii & 0
	\end{pmatrix},
\end{equation}
which under the basis transformation with 
\begin{equation}
	C = \begin{pmatrix}
		1 & -1 \\ 0 & 1
	\end{pmatrix},
\end{equation}
recovers the maximal super-extension of sub-branch \hyperlink{N2_a_21ii}{$\mathsf{2.1.ii}$}.  Thus, this sub-branch does not contribute any new super-extensions to $\hat{\a}$.
\\
\paragraph{\textbf{Sub-branch 4.2}}
The $[\Q, \Q, \Q]$ identity still imposes that either $N_0$ or $\mh$ must vanish in this
sub-branch; however, we can now consider the case where $N_0 = 0$ as we have $N_1 \neq 0$.
First, consider case $(\mathsf{i})$, with $N_0 = 0$ such that $\mh \neq 0$.  The conditions remaining are
\begin{equation}
	\begin{split}
		0 &= \mh N_1 + N_1 \mh^\dagger \\
		0 &= \mb N_1 - N_1 \mb^\dagger \\
		0 &= \mp N_1 - N_1 \mp^\dagger, \\
	\end{split}
\end{equation}
The latter two conditions tell us that $c=0$ and $\bb_3$ is collinear with $\bb_3\uu$ and $\rr$.  Substituting these
results into the first condition, we find
\begin{equation}
	\begin{split}
		0 &= \hh_1 \rr + \rr \overbar{\bb_3 \hh_1 \bb_3^{-1}} \\
		0 &= \Re(\hh_3\rr) + d \Re(\hh_1). \\
	\end{split}
\end{equation}
We know from the analysis of branch 3 that demanding $N_1 \neq 0$ under these conditions imposes $\Re(\hh_1) = 0$; and, that having the condition
\begin{equation}
		0 = \bb_3 \rr - \bar{\rr} \bar{\bb_3}
\end{equation}
means we always satisfy the imaginary part of 
\begin{equation}
	0 = \hh_1 \rr + \rr \overbar{\bb_3 \hh_1 \bb_3^{-1}}.
\end{equation}
Putting all this together, we find the constraints on this sub-branch to be
\begin{equation}
	0 = \Re(\hh_1), \quad 0 = \Re(\hh_3 \rr) , \quad 0 = \bb_3 \rr - \bar{\rr}\bar{\bb_3}, \quad
		0 = \bb_3 \uu \rr - \bar{\rr}\bar{\uu} \bar{\bb_3}, \quad 0 = [\uu, \hh_1].
\end{equation}
The \hypertarget{N2_a_42i}{non-vanishing} matrices are then 
\begin{equation}
	\mb = \begin{pmatrix}
		0 & 0 \\ \bb_3 & 0
	\end{pmatrix}, \quad
	\mp = \begin{pmatrix}
		0 & 0 \\ \bb_3\uu & 0
	\end{pmatrix}, \quad 
	\mh = \begin{pmatrix}
		\hh_1 & 0 \\ \hh_3 & \bb_3 \hh_1 \bb_3^{-1}
	\end{pmatrix} \quad \text{and} \quad 
	N_1 = \begin{pmatrix}
		0 & \rr \\ \bar{\rr} & d
	\end{pmatrix}.
\end{equation}
Notice that $\bb_3$, $\rr$ and $\bb_3\uu$ all being collinear implies that $\uu \in \mathbb{R}$.  Thus the final constraint is satisfied, and, as in sub-branch 4.1, we can use the endomorphisms of $\s_{\bar{1}}$ and the automorphisms of $\hat{\a}$ to rotate $\sB$ and $\sP$ such that we only have the matrix $\mp$, in which $\pp_3 = \ii$.  The resulting matrices and constraints are then equivalent to those found in sub-branch \hyperlink{N2_a_22i}{$\mathsf{2.2.i}$}, and, therefore, this sub-branch does not produce any new super-extensions for $\hat{\a}$.
\\ \\
Now, considering case $(\mathsf{ii})$, let $\mh = 0$.  The remaining conditions are
\begin{equation}
	\begin{split}
		0 &= \mb N_i - N_i \mb^\dagger \quad \text{where} \quad i \in \{0, 1\} \\
		0 &= \mp N_i - N_i \mp^\dagger \quad \text{where} \quad i \in \{0, 1\}.
	\end{split}
\end{equation}
Therefore, $N_0$ and $N_1$ take the same form in this instance: both $a$ and $c$ vanish, with $\qq$ and $\rr$ being collinear to both $\bb_3$ and $\bb_3\uu$.  In summary, the constraints are
\begin{equation}
	0 = \bb_3 \qq - \bar{\qq} \bar{\bb_3}, \quad 0 = \bb_3\uu \qq - \bar{\qq} \bar{\uu} \bar{\bb_3}, \quad 0 = \bb_3 \rr - \bar{\rr} \bar{\bb_3}, \quad 0 = \bb_3 \uu \rr - \bar{\rr} \bar{\uu} \bar{\bb_3},
\end{equation}
and the non-vanishing \hypertarget{N2_a_42ii}{matrices} are
\begin{equation}
	\mb = \begin{pmatrix}
		0 & 0 \\ \bb_3 & 0
	\end{pmatrix}, \quad
	\mp = \begin{pmatrix}
		0 & 0 \\ \bb_3\uu & 0
	\end{pmatrix}, \quad
	N_0 = \begin{pmatrix}
		0 & \qq \\ \bar{\qq} & b
	\end{pmatrix}, \quad \text{and} \quad 
	N_1 = \begin{pmatrix}
		0 & \rr \\ \bar{\rr} & d
	\end{pmatrix}.
\end{equation}
Through the same use of the subgroup $\G \subset \GL(\s_{\bar{0}}) \times \GL(\s_{\bar{1}})$ as discussed for case $(\mathsf{i})$, we find that this sub-branch is equivalent to \hyperlink{N2_a_22ii}{$\mathsf{2.2.ii}$} for $\hat{\a}$.
\\
\paragraph{\textbf{Sub-branch 4.3}} Now with $N_3 \neq 0$, we can use 
\begin{equation}
	-N_3 = \mp N_1 - N_1 \mp^\dagger \quad \text{and} \quad  0 = \pi \vt \mp N_3 \vt^\dagger + \vt N_3 (\pi \vt \mp)^\dagger
\end{equation}
to first write $N_3$ in terms of $\mp$ and $N_1$ before setting $c=0$ by substituting $\vt = (0,1)$ into the latter condition.
This produces the matrix
\begin{equation}
	N_3 = \begin{pmatrix}
		0 & 0 \\ 0 &  \bar{\rr} \bar{\uu}\bar{\bb_3} - \bb_3\uu\rr
	\end{pmatrix}.
\end{equation}
Since $N_3$ and $\mb$ are non-vanishing, the condition from the $[\Q, \Q, \Q]$ identity no longer
states that we must set either $N_0$ or $\mh$ to zero.  We have
\begin{equation}
	\Re(\vt N_0 \vt^\dagger) \vt\mh = \vt N_3 \vt^\dagger \vt \mb.
\end{equation}
Substituting $\vt = (0, 1)$ into the above condition, we find
\begin{equation} \label{eq:N2_branch4_a_43_h3}
	b \hh_3 = (\bar{\rr} \bar{\uu}\bar{\bb_3} - \bb_3\uu\rr) \bb_3 \quad \text{and} \quad b \hh_1 = 0.
\end{equation}
By assumption $N_3 \neq 0$; therefore, both $b$ and $\hh_3$ cannot vanish.  Using this result, the second constraint tells us that $\hh_1 = 0$.  Thus $\mh$ is reduced to a strictly lower-diagonal matrix.  As in sub-branch \hyperlink{N2_a_42i}{$\mathsf{4.2}$}, we have 
\begin{equation}
	\begin{split}
		0 &= \mb N_i - N_i \mb^\dagger \quad \text{where} \quad i \in \{0, 1\} \\
		0 &= \mp N_0 - N_0 \mp^\dagger,
	\end{split}
\end{equation}
which tell us $a$ and $c$ vanish, and
\begin{equation}
	0 = \bb_3\qq - \bar{\qq}\bar{\bb_3}, \quad 0 = \bb_3\uu\qq - \bar{\qq}\bar{\uu}\bar{\bb_3}, \quad \text{and}
	\quad 0 = \bb_3 \rr - \bar{\rr}\bar{\bb_3}.
\end{equation}
Using these results and the rewriting of $\hh_3$ in \eqref{eq:N2_branch4_a_43_h3}, the conditions from the $[\bH, \Q, \Q]$ identity are instantly satisfied.  Therefore, the constraints on the parameters of this sub-branch are
\begin{equation}
	0 = \bb_3\qq - \bar{\qq}\bar{\bb_3}, \; 0 = \bb_3\uu\qq - \bar{\qq}\bar{\uu}\bar{\bb_3}, \; 0 = \bb_3 \rr - \bar{\rr}\bar{\bb_3}, \; \text{and} \; b \hh_3 = (\bar{\rr}\bar{\uu}\bar{\bb_3} - \bb_3\uu\rr) \bb_3.
\end{equation}
Notice that the first three constraints here tell us that $\bb_3$ is collinear with both $\qq$ and $\rr$ and that $\bb_3\uu$ is collinear with $\qq$.  In particular, were we to use the endomorphisms of $\s_{\bar{1}}$ to set $\qq$ to lie along $\ii$, $\bb_3$, $\bb_3\uu$ and $\rr$ would all lie along $\ii$ as well.  Thus, $\bb_3\uu\rr \in \mathbb{R}$, such that $N_3 = 0$.  Therefore, this sub-branch is empty.
\\
\paragraph{\textbf{Sub-branch 4.4}} This sub-branch will be very similar to the one above
due to the similarity in the conditions the super-Jacobi identities imposes on $N_3$ and $N_4$.
Using 
\begin{equation}
		N_4 = \mb N_1 - N_1 \mb^\dagger \quad \text{and} \quad 0 = \beta \vt \mb N_4 \vt^\dagger + \vt N_4 (\beta \vt \mb)^\dagger,
\end{equation}
we know $N_4$ may be written
\begin{equation}
	N_4 = \begin{pmatrix}
		0 & 0 \\ 0 & \bb_3\rr - \bar{\rr}\bar{\bb_3}
	\end{pmatrix}.
\end{equation}
Lemma \ref{lem:N2_111} then tells us that
\begin{equation}
	\Re(\vt N_0 \vt^\dagger) \vt\mh = \vt N_4 \vt^\dagger \vt \mp.
\end{equation}
Substituting $\vt = (0, 1)$ into this condition produces
\begin{equation} \label{eq:N2_C_a_branch44_h_condition}
	b \hh_3 = (\bb_3\rr - \bar{\rr}\bar{\bb_3}) \bb_3\uu \quad \text{and} \quad b \hh_1 = 0.
\end{equation}
Since $\bb_3\uu \neq 0$ and $\bb_3\rr - \bar{\rr}\bar{\bb_3} \neq 0$ by assumption, $b$ cannot vanish; therefore,
$\hh_1 = 0$.  The conditions
\begin{equation}
	\begin{split}
		0 &= \mp N_i - N_i \mp^\dagger \quad \text{where} \quad i \in \{0, 1\} \\
		0 &= \mb N_0 - N_0 \mb^\dagger,  \\
	\end{split}
\end{equation}
tell us that both $a$ and $c$ vanish, and 
\begin{equation}
	0 = \bb_3\qq- \bar{\qq}\bar{\bb_3}, \quad 0 =\bb_3\uu\qq - \bar{\qq}\bar{\uu}\bar{\bb_3}, \quad \text{and}
	\quad 0 = \bb_3\uu \rr - \bar{\rr}\bar{\uu}\bar{\bb_3}.
\end{equation}
Finally, we have the conditions from the $[\bH, \Q, \Q]$ identity, which impose
\begin{equation}
	0 = \Re(\hh_3 \qq) \quad \text{and} \quad 0 = \Re(\hh_3 \rr) .
\end{equation}
However, using the form of $\hh_3$ in \eqref{eq:N2_C_a_branch44_h_condition} and the collinearity of $\bb_3\uu$ with $\qq$ and $\rr$, both of these constraints are already satisfied.  Therefore, the final set of constraints on this sub-branch is
\begin{equation}
	0 = \bb_3\qq- \bar{\qq}\bar{\bb_3}, \quad 0 =\bb_3\uu\qq - \bar{\qq}\bar{\uu}\bar{\bb_3}, \quad 0 = \bb_3\uu \rr - \bar{\rr}\bar{\uu}\bar{\bb_3}, \quad b \hh_3 = (\bb_3\rr - \bar{\rr}\bar{\bb_3}) \bb_3\uu.
\end{equation}
Notice that the first three constraints tell us that $\bb_3$, $\bb_3\uu$, $\qq$, and $\rr$ are collinear.  This tells us that $\bb_3\rr \in \mathbb{R}$; therefore, significantly, $N_4 = 0$.  Thus this sub-branch is empty. 
\\
\paragraph{\textbf{Sub-branch 4.5}}  Now with non-vanishing $N_3$ and $N_4$, we can begin by using
\begin{equation}
	\begin{split}
		-N_3 &= \mp N_1 - N_1 \mp^\dagger \\
		0 &= \pi \vt \mp N_3 \vt^\dagger + \vt N_3 (\pi \vt \mp)^\dagger
	\end{split} \quad \text{and} \quad 
	\begin{split}
		N_4 &= \mb N_1 + N_1 \mb^\dagger \\
		0 &= \beta \vt \mb N_4 \vt^\dagger + \vt N_4 (\beta \vt \mb)^\dagger,
	\end{split}
\end{equation}
to write
\begin{equation}
	N_1 =  \begin{pmatrix}
		0 & \rr \\ \bar{\rr} & d
	\end{pmatrix}, \quad 
	N_3 = \begin{pmatrix}
		0 & 0 \\ 0 & \bar{\rr}\bar{\uu}\bar{\bb_3}-\bb_3\uu\rr
	\end{pmatrix} \quad \text{and} \quad 
	N_4 = \begin{pmatrix}
		0 & 0 \\ 0 & \bb_3\rr - \bar{\rr}\bar{\bb_3}
	\end{pmatrix}.
\end{equation}
Using these results, substitute $\vt = (0, 1)$ into the condition from the $[\Q, \Q, \Q]$
identity to find
\begin{equation} \label{eq:h_3_prescription}
	b \hh_3 = (\bar{\rr}\bar{\uu}\bar{\bb_3} - \bb_3\uu\rr) \bb_3 + (\bb_3\rr - \bar{\rr}\bar{\bb_3})\bb_3\uu \quad \text{and} \quad b \hh_1 = 0.
\end{equation}
As in all previous sub-branches, the $[\P, \Q, \Q]$ and $[\B, \Q, \Q]$ conditions on $N_0$ tell us
\begin{equation}
	0 = a, \quad 0 = \bb_3\qq - \bar{\qq}\bar{\bb_3} \quad \text{and} \quad 0 = \bb_3\uu\qq - \bar{\qq}\bar{\uu}\bar{\bb_3}.
\end{equation}
Finally, the $[\bH, \Q, \Q]$ identities tell us
\begin{equation}
	\begin{split}
		0 &= \Re(\hh_3 \qq) + b \Re(\hh_1)\\
		0 &= \hh_1 \qq + \qq \overbar{\bb_3\hh_1\bb_3^{-1}}
	\end{split} \quad \text{and} \quad
	\begin{split} 
		0 &= \Re(\hh_3 \rr) + d \Re(\hh_1) \\
		0 &= \hh_1 \rr + \rr \overbar{\bb_3\hh_1\bb_3^{-1}}.
	\end{split}
\end{equation}
Since, by assumption, $N_1 \neq 0$, these constraints mean we must have $\Re(\hh_1) = 0$.
If $\rr = 0$, we would need $d \neq 0$, which, when substituted into $0 = d \Re(\hh_1)$, mean $\Re(\hh_1)=0$.  Alternatively, if $\rr \neq 0$, we multiply
\begin{equation}
	0 = \hh_1 \rr + \rr \overbar{\bb_3\hh_1\bb_3^{-1}} 
\end{equation}
on the right by $\rr^{-1}$ and take the real part to obtain $\Re(\hh_1) = 0$.  Knowing this, we can use the fact $\overbar{\bb_3\hh_1\bb_3^{-1}} \in \Im(\mathbb{H})$ to rewrite the remaining imaginary part of this constraint as
\begin{equation}
	0 = [\hh_1, \rr\bb_3].
\end{equation}
Additionally, since $\Re(\hh_1)=0$, we can use $ 0 = \bb_3\qq - \bar{\qq}\bar{\bb_3}$ to instantly
satisfy the condition
\begin{equation}
 0 = \hh_1 \qq + \qq \overbar{\bb_3\hh_1\bb_3^{-1}}.
\end{equation}
These results leave us with
\begin{equation}
	\begin{split}
	\mathcal{C}_{\hat{\a},\,\mathsf{4.5.iii}} = \{
		0 &= \bb_3\qq - \bar{\qq}\bar{\bb_3}, \quad
		0 = \bb_3\uu\qq - \bar{\qq}\bar{\uu}\bar{\bb_3}, \\
		0 &= \Re(\hh_3\qq), \quad
		0 = \Re(\hh_3\rr), \quad 
		0 = \Re(\hh_1), \\
		0 &= [\hh_1, \rr\bb_3], \quad 
		0 = b \hh_1, \quad
		b \hh_3 = - 2 \Im(\bb_3\uu\rr)\bb_3 + 2 \Im(\bb_3\rr)\bb_3\uu \}.
	\end{split}
\end{equation}
Subject to these constraints, the non-vanishing \hypertarget{N2_a_45}{matrices} are
\begin{equation}
	\begin{split}
	\mathcal{M}_{\hat{\a},\,\mathsf{4.5.iii}} = \Big\{ \mb &= \begin{pmatrix}
		0 & 0 \\ \bb_3 & 0
	\end{pmatrix}, \quad
	\mp = \begin{pmatrix}
		0 & 0 \\ \bb_3\uu & 0
	\end{pmatrix}, \quad
	\mh = \begin{pmatrix}
		\hh_1 & 0 \\ \hh_3 & \bb_3\hh_1\bb_3^{-1}
	\end{pmatrix}, \\
	N_0 &= \begin{pmatrix}
		0 & \qq \\ \bar{\qq} & b
	\end{pmatrix}, \quad
	N_1 = \begin{pmatrix}
		0 & \rr \\ \bar{\rr} & d
	\end{pmatrix}, \quad
	N_3 = \begin{pmatrix}
		0 & 0 \\ 0 & \bar{\rr}\bar{\uu}\bar{\bb_3}-\bb_3\uu\rr
	\end{pmatrix}, \quad 
	N_4 = \begin{pmatrix}
		0 & 0 \\ 0 & \bb_3\rr - \bar{\rr}\bar{\bb_3}
	\end{pmatrix} \Big\}.
	\end{split}
\end{equation}
The wealth of parameters describing this sub-branch mean we will only highlight one parameterisation of these super-extensions here, though many more may exist.  In particular, we will choose to set $b$, $d$ and $\hh_3$ to zero.  We will also utilise the subgroup $\G \subset \GL(\s_{\bar{0}}) \times \GL(\s_{\bar{1}})$ to impose that $\qq$, $\bb_3$ and $\bb_3\uu$, lie along $\ii$. The residual endomorphisms of $\s_{\bar{1}}$ may then scale $\rr$ such that its norm becomes $1$. Employing $\Aut(\mathbb{H})$, we can set $\hh_1$ to lie along $\ii$ as well. Having made these choices, the constraint
\begin{equation}
	0 = [\hh_1, \rr\bb_3]
\end{equation}
tells us $\rr \in \mathbb{R} \langle 1, \ii \rangle$.  Notice that for $N_3$ and $N_4$ to be non-vanishing $\rr$ must have a real component; therefore, to simplify the form of the matrices in our example, we will choose $\rr = 1$.  The remaining constraints in $\mathcal{C}_{\hat{\a},\,\mathsf{4.5.iii}}$ are then satisfied, and we can use the scaling symmetry of the $\s_{\bar{0}}$ basis elements to produce 
\begin{equation}
	\begin{split}
	 \mb &= \begin{pmatrix}
		0 & 0 \\ \ii & 0
	\end{pmatrix}, \quad
	\mp = \begin{pmatrix}
		0 & 0 \\ \ii & 0
	\end{pmatrix}, \quad
	\mh = \begin{pmatrix}
		\ii & 0 \\ 0 & \ii
	\end{pmatrix}, \\
	N_0 &= \begin{pmatrix}
		0 & \ii \\ -\ii & 0
	\end{pmatrix}, \quad
	N_1 = \begin{pmatrix}
		0 & 1 \\ 1 & 0
	\end{pmatrix}, \quad
	N_3 = \begin{pmatrix}
		0 & 0 \\ 0 & \ii
	\end{pmatrix} \quad \text{and} \quad 
	N_4 = \begin{pmatrix}
		0 & 0 \\ 0 & \ii
	\end{pmatrix}.
	\end{split}
\end{equation}
\paragraph{\textbf{Sub-branch 4.6}} We find that this sub-branch is empty using the analysis from the previous
sub-branch. Again, we use
\begin{equation}
	\begin{split}
		-N_3 &= \mp N_1 - N_1 \mp^\dagger \\
		0 &= \pi \vt \mp N_3 \vt^\dagger + \vt N_3 (\pi \vt \mp)^\dagger
	\end{split} \quad \text{and} \quad 
	\begin{split}
		N_4 &= \mb N_1 + N_1 \mb^\dagger \\
		0 &= \beta \vt \mb N_4 \vt^\dagger + \vt N_4 (\beta \vt \mb)^\dagger,
	\end{split}
\end{equation}
to write
\begin{equation}
	N_3 = \begin{pmatrix}
		0 & 0 \\ 0 & \bar{\rr}\bar{\uu}\bar{\bb_3}- \bb_3\uu\rr
	\end{pmatrix} \quad
	N_4 = \begin{pmatrix}
		0 & 0 \\ 0 & \bb_3\rr - \bar{\rr}\bar{\bb_3}
	\end{pmatrix}.
\end{equation}
Substituting these matrices into 
\begin{equation}
	\begin{split}
		\tfrac12 [\beta, \vt N_2 \vt^\dagger] &= \beta \vt \mb N_3 \vt^\dagger
		+ \vt N_3 (\beta \vt \mb)^\dagger \\
		\tfrac12 [\pi, \vt N_2 \vt^\dagger] &= \pi \vt \mp N_4 \vt^\dagger
		+ \vt N_4 (\pi \vt \mp)^\dagger ,
	\end{split}
\end{equation}
the R.H.S. of both of these constraints vanishes, setting $N_2 = 0$.  Therefore, this
sub-branch is empty.
\\
\paragraph{$\hat{\n}_-$}~\\ \\
Setting $\mu = \eta = 0$, $\lambda = 1$ and 
$\varepsilon = -1$, the first condition
in \eqref{eq:N2_branch4_condition} becomes
\begin{equation}
	[\uu, \hh_1] = 2 \uu.
\end{equation}
Since $[\uu, \hh_1]$ is perpendicular to $\uu$ in $\Im(\mathbb{H})$ this branch 
cannot provide a super-extension for $\hat{\n}_-$.  
\\
\paragraph{$\hat{\n}_+$}~\\ \\In this case, for which $\lambda = \varepsilon = 0$, 
$\mu = 1$, and $\eta = -1$, the first constraint in \eqref{eq:N2_branch4_condition} gives us
\begin{equation} \label{eq:N2_branch4_n-}
	[\hh_1, \uu] = \uu^2 + 1.
\end{equation}
Taking the real part of \eqref{eq:N2_branch4_n-} produces
\begin{equation}
	\Re(\uu^2) = -1,
\end{equation}
therefore, $\uu \in \Im(\mathbb{H})$, such that $|\uu|^2 = 1$, i.e. it is
a unit-norm vector quaternion, or \textit{right versor}.  The imaginary part
of  \eqref{eq:N2_branch4_n-} imposes
\begin{equation}
	[\uu, \hh_1] = 0.
\end{equation} 
Thus, we could get a super-extension of $\hat{\n}_+$ in this branch.
Wishing to write our parameters in terms of $\bb_3$, we have
$\pp_3 = \bb_3 \uu$ and $\hh_4 = \bb_3 (\hh_1 - \uu) \bb_3^{-1}$, where
$\uu \in \Im(\mathbb{H})$, such that $\uu^2 = -1$.
\\ \\
As with the $\hat{\a}$ case above, all of the conditions
of Lemmas \ref{lem:N2_011} and \ref{lem:N2_111} must be taken into account.  
The conditions
\begin{equation}
		N_4 = \mb N_1 - N_1 \mb^\dagger \quad \text{and} \quad -N_3 = \mp N_1 - N_1 \mp^\dagger
\end{equation}
tell us that if $N_1 = 0$, $N_3 = 0$ and $N_4 = 0$.  Substituting these results into
\begin{equation}
	\begin{split}
		- \Re(\vt N_0\vt^\dagger) \pi &= \pi \vt \mp N_3 \vt^\dagger + \vt N_3 (\pi \vt \mp)^\dagger \\
		\tfrac12 [\beta, \vt N_2 \vt^\dagger] &= \beta \vt \mb N_3 \vt^\dagger
		+ \vt N_3 (\beta \vt \mb)^\dagger \\
	\end{split} \qquad
	\begin{split}
		\Re(\vt N_0\vt^\dagger) \beta &= \beta \vt \mb N_4 \vt^\dagger + \vt N_4 (\beta \vt \mb)^\dagger \\
		\tfrac12 [\pi, \vt N_2 \vt^\dagger] &= \pi \vt \mp N_4 \vt^\dagger
		+ s N_4 (\pi \vt \mp)^\dagger,	\\
	\end{split}
\end{equation}
we see that if $N_3$ or $N_4$ vanish, so must $N_0$ and $N_2$.  Equally, if $N_3$ vanishes
$N_4$ necessarily vanishes and vice versa due to
\begin{equation}
		- N_4 = \mh N_3 + N_3 \mh^\dagger \quad \text{and} \quad  N_3 = \mh N_4 + N_4 \mh^\dagger.
\end{equation}
Therefore, based on these dependencies, our investigation into this branch of possible super-extensions
of $\hat{\n}_+$ divides into the following sub-branches.
\begin{enumerate}
	\item $N_1 \neq 0$, and $N_0 = N_2 = N_3 = N_4 = 0$
	\item $N_1 \neq 0$, $N_3 \neq 0$, $N_4 \neq 0$, and $N_0 = N_2 = 0$
	\item $N_1 \neq 0$, $N_3 \neq 0$, $N_4 \neq 0$, $N_0 \neq 0$ and $N_2 = 0$
	\item $N_1 \neq 0$, $N_3 \neq 0$, $N_4 \neq 0$, $N_0 = 0$ and $N_2 \neq 0$
	\item $N_1 \neq 0$, $N_3 \neq 0$, $N_4 \neq 0$, $N_0 \neq 0$ and $N_2 \neq 0$
\end{enumerate}
Unlike the super-extensions of $\hat{\n}_+$ found in branches 1, 2 and 3, the $[\Q, \Q, \Q]$ identity will not impose 
that either $N_0$ or $\mh$ must vanish.  In the first sub-branch above, we instantly see that $N_0=0$; therefore, the super-extensions found here are extensions satisfying $(\mathsf{i})$.  However, all other sub-branches have either non-vanishing $N_3$ or non-vanishing $N_4$.  Since $\mb\neq 0$ and $\mp \neq 0$, the $[\Q, \Q, \Q]$ identity will now form relationships between $N_0$, $N_3$ and $N_4$, with, in general, $\mh \neq 0$.  Therefore, these super-extensions, for which $N_0 \neq 0$ and $\mh \neq 0$, will be labelled $(\mathsf{iii})$ to distinguish them from the cases $(\mathsf{i})$ and $(\mathsf{ii})$.
\\
\paragraph{\textbf{Sub-branch 4.1}} With only $N_1 \neq 0$, the conditions from Lemmas \ref{lem:N2_011}
and \ref{lem:N2_111} reduce to 
\begin{equation}
	\begin{split}
		0 &= \mh N_1 +N_1 \mh^\dagger \\
		0 &= \mp N_1 - N_1 \mp^\dagger \\
		0 &= \mb N_1 - N_1 \mb^\dagger.
	\end{split}
\end{equation}
The latter two conditions tell us 
\begin{equation}
	0 = c, \quad 0 = \bb_3 \rr - \bar{\rr}\bar{\bb_3} \quad \text{and} \quad 0 = \bb_3 \uu \rr - \bar{\rr}\bar{\uu}\bar{\bb_3},
\end{equation}
which, when substituted into the first conditions, produce
\begin{equation}
	0 = \Re(\hh_1) \quad \text{and} \quad 0 = \Re(\hh_3 \rr).
\end{equation}
Therefore, the non-vanishing \hypertarget{N2_n+_41i}{matrices} for this super-extension are
\begin{equation}
	\mb = \begin{pmatrix}
		0 & 0 \\ \bb_3 & 0
	\end{pmatrix}, \quad 
	\mp = \begin{pmatrix}
		0 & 0 \\ \bb_3\uu & 0
	\end{pmatrix}, \quad 
	\mh = \begin{pmatrix}
		\hh_1 & 0 \\ \hh_3 & \bb_3(\hh_1 - \uu)\bb_3^{-1} 
	\end{pmatrix}, \quad
	N_1 = \begin{pmatrix}
		0 & \rr \\ \bar{\rr} & d
	\end{pmatrix},
\end{equation}
subject to the constraints
\begin{equation}
	0 = [\uu, \hh_1], \quad 0 = \Re(\hh_1), \quad 0 = \Re(\hh_3 \rr), \quad 0 = \bb_3 \rr - \bar{\rr}\bar{\bb_3}, \quad 0 = \bb_3 \uu \rr - \bar{\rr}\bar{\uu}\bar{\bb_3}, \quad \uu^2 = -1.
\end{equation}
However, notice that the final three constraints listed above require one of $\bb_3$, $\uu$ or $\rr$ to vanish.  Since neither $\bb_3$ or $\uu$ can vanish in this sub-branch, it must be that $\rr = 0$.  Therefore, the final set of matrices is 
\begin{equation}
	\mathcal{M}_{\hat{\n}_+,\,\mathsf{4.1.i}} = \Big\{ \mb = \begin{pmatrix}
		0 & 0 \\ \bb_3 & 0
	\end{pmatrix}, \quad 
	\mp = \begin{pmatrix}
		0 & 0 \\ \bb_3\uu & 0
	\end{pmatrix}, \quad 
	\mh = \begin{pmatrix}
		\hh_1 & 0 \\ \hh_3 & \bb_3(\hh_1 - \uu)\bb_3^{-1} 
	\end{pmatrix}, \quad
	N_1 = \begin{pmatrix}
		0 & 0 \\ 0 & d
	\end{pmatrix} \Big\},
\end{equation}
and the final set of constraints is
\begin{equation}
	\mathcal{C}_{\hat{\n}_+,\,\mathsf{4.1.i}} = \{ 0 = [\uu, \hh_1], \quad  0 = \Re(\hh_1), \quad \uu^2 = -1\}.
\end{equation}
To demonstrate that this sub-branch of $\cS$ is not empty, choose to set $\hh_1$ and $\hh_3$ to zero.  Using the endomorphisms of $\s_{\bar{1}}$ and $\Aut(\mathbb{H})$ on $\bb_3$ and $\uu$, respectively, we may write $\bb_3 = \ii$ and $\uu = \jj$.  Employing the scaling symmetry of $\sZ$, we arrive at the super-extension
\begin{equation}
	 \mb = \begin{pmatrix}
		0 & 0 \\ \ii & 0
	\end{pmatrix}, \quad 
	\mp = \begin{pmatrix}
		0 & 0 \\ \kk & 0
	\end{pmatrix}, \quad 
	\mh = \begin{pmatrix}
		0 & 0 \\ 0 & \jj
	\end{pmatrix}, \quad
	N_1 = \begin{pmatrix}
		0 & 0 \\ 0 & 1
	\end{pmatrix}.
\end{equation}  
Thus, this sub-branch is not empty.  We may then introduce $\hh_1$ while continuing to fix all the parameters of the super-extension; however, this cannot be achieved on introducing $\hh_3$.
\\
\paragraph{\textbf{Sub-branch 4.2}} Now with $N_3 \neq 0$ and $N_4 \neq 0$ as well as $N_1\neq 0$,
we can use the conditions
\begin{equation}
	\begin{split}
		-N_3 &= \mp N_1 - N_1 \mp^\dagger \\
		N_4 &= \mb N_1 - N_1 \mb^\dagger
	\end{split} \qquad 
	\begin{split}
		0 &= \beta \vt \mb N_3 \vt^\dagger
		+ \vt N_3 (\beta \vt \mb)^\dagger \\
		0 &= \pi \vt \mp N_3 \vt^\dagger + \vt N_3 (\pi \vt \mp)^\dagger
	\end{split} \qquad
	\begin{split}
		0 &= \beta \vt \mb N_4 \vt^\dagger + \vt N_4 (\beta \vt \mb)^\dagger \\
		0 &= \pi \vt \mp N_4 \vt^\dagger
		+ \vt N_4 (\pi \vt \mp)^\dagger,
	\end{split}
\end{equation}
and the analysis of branches 2 and 3 to write
\begin{equation}
	N_1 = \begin{pmatrix}
		0 & \rr \\ \bar{\rr} & d
	\end{pmatrix} \quad 
	N_3 = \begin{pmatrix}
		0 & 0 \\ 0 & \bar{\rr}\bar{\uu}\bar{\bb_3}-\bb_3\uu\rr
	\end{pmatrix} \quad 
	N_4 = \begin{pmatrix}
		0 & 0 \\ 0 & \bb_3 \rr - \bar{\rr}\bar{\bb_3}
	\end{pmatrix}.
\end{equation}
This leaves only the $[\bH, \Q, \Q]$ conditions:
\begin{equation}
	\begin{split}
		0 &= \mh N_1 + N_1 \mh^\dagger \\
		- N_4 &= \mh N_3 + N_3 \mh^\dagger \\
		N_3 &= \mh N_4 + N_4 \mh^\dagger. \\	
	\end{split}
\end{equation}
We know from sub-branch \hyperlink{N2_n+_41i}{$\mathsf{4.1.i}$} that, since $c = 0$, the first of these produces
\begin{equation}
	0 = \Re(\hh_1) \quad \text{and} \quad 0 = \Re(\hh_3 \rr).
\end{equation}
Writing $\bar{\rr}\bar{\uu}\bar{\bb_3}-\bb_3\uu\rr=-2\Im(\bb_3\uu\rr)$ and $\bb_3 \rr - \bar{\rr}\bar{\bb_3}=2\Im(\bb_3\rr)$ to simplify our expressions, the second and third conditions give us
\begin{equation}
	\begin{split}
		\Im(\bb_3\rr) &= \hh_4 \Im(\bb_3\uu\rr) + \Im(\bb_3\uu\rr) \bar{\hh_4} \\
		- \Im(\bb_3\uu\rr) &= \hh_4 \Im(\bb_3\rr) + \Im(\bb_3\rr) \bar{\hh_4},
	\end{split}
\end{equation}
respectively, where $\hh_4 = \bb_3(\hh_1 - \uu)\bb_3^{-1}$.   Notice that since $\hh_1, \uu \in \Im(\mathbb{H})$, and $\hh_4$ is written
in terms of the adjoint action of $\bb_3 \in \mathbb{H}$, $\hh_4 \in \Im(\mathbb{H})$.
Therefore, using $\bar{\hh_4} = - \hh_4$, we find
\begin{equation}
	\Im(\bb_3\rr) = - [\hh_4, [\hh_4, \Im(\bb_3\rr)]] \quad \text{and} \quad 
	\Im(\bb_3\uu\rr) = - [\hh_4, [\hh_4, \Im(\bb_3\uu\rr)]].
\end{equation}
This imposes the constraint that $\Im(\bb_3\rr)$ and $\Im(\bb_3\uu\rr)$ must be perpendicular
to $\hh_4$ in $\Im(\mathbb{H})$.  The constraints for this sub-branch are summarised
as follows.
\begin{equation}
	\begin{split}
		\mathcal{C}_{\hat{\n}_+,\,\mathsf{4.2.iii}} = \{ 
			&0 = [\hh_1, \uu], \quad
			-1 = \uu^2 \quad
			0 = \Re(\hh_1), \quad
			0 = \Re(\hh_3 \rr), \\
		&\Im(\bb_3\rr) = - [\hh_4, [\hh_4, \Im(\bb_3\rr)]], \quad
		\Im(\bb_3\uu\rr) = - [\hh_4, [\hh_4, \Im(\bb_3\uu\rr)]] \}.
	\end{split}
\end{equation}
The non-vanishing \hypertarget{N2_n+_42i}{matrices} are then
\begin{equation}
	\begin{split} 
	\mathcal{M}_{\hat{\n}_+,\,\mathsf{4.2.iii}} = \Big\{ \mb &= \begin{pmatrix}
		0 & 0 \\ \bb_3 & 0
	\end{pmatrix}, \quad 
	\mp = \begin{pmatrix}
		0 & 0 \\ \bb_3\uu & 0
	\end{pmatrix}, \quad 
	\mh = \begin{pmatrix}
		\hh_1 & 0 \\ \hh_3 & \bb_3(\hh_1 - \uu)\bb_3^{-1} 
	\end{pmatrix}, \\
	N_1 &= \begin{pmatrix}
		0 & \rr \\ \bar{\rr} & d
	\end{pmatrix}, \quad 
	N_3 = \begin{pmatrix}
		0 & 0 \\ 0 & - 2 \Im(\bb_3\uu\rr)
	\end{pmatrix}, \quad
	N_4 = \begin{pmatrix}
		0 & 0 \\ 0 & 2 \Im(\bb_3\rr)
	\end{pmatrix} \Big\}. \end{split}		
\end{equation}
To demonstrate the existence of super-extensions in this sub-branch, we will begin by simplifying our parameter set as much as possible.  In particular, we begin by setting both $\hh_3$ and $d$ to zero.  We then utilise $\Aut(\mathbb{H})$ and the endomorphisms of $\s_{\bar{1}}$ to set $\uu = \jj$ and impose that $\bb_3$ lies along $\ii$.  Notice that with $\uu$ along $\jj$, the first constraint in $\mathcal{C}_{\hat{\n}_+,\,\mathsf{4.2.iii}}$ tells us that $\hh_1$ must also lie along $\jj$, as must $\hh_4 = \bb_3(\hh_1 - \uu)\bb_3^{-1}$.  With these choices, the two constraints involving $\hh_4$ impose $\rr \in \mathbb{R} \langle 1, \jj \rangle$, and that $|\hh_1| = \tfrac12$ or $|\hh_1| = \tfrac32$.  Residual endomorphisms then allow us to scale $\rr$ such that it has unit norm.  Finally, we can scale the $\s_{\bar{0}}$ basis elements to arrive at
\begin{equation}
	\begin{split} 
	\mb &= \begin{pmatrix}
		0 & 0 \\ \ii & 0
	\end{pmatrix}, \quad 
	\mp = \begin{pmatrix}
		0 & 0 \\ \kk & 0
	\end{pmatrix}, \quad 
	\mh = \begin{pmatrix}
		\jj & 0 \\ 0 & \jj 
	\end{pmatrix}, \\
	N_1 &= \begin{pmatrix}
		0 & 1 + \jj \\ 1 - \jj & 0
	\end{pmatrix}, \quad 
	N_3 = \begin{pmatrix}
		0 & 0 \\ 0 & \kk - \ii
	\end{pmatrix}, \quad
	N_4 = \begin{pmatrix}
		0 & 0 \\ 0 & \kk + \ii
	\end{pmatrix}.
	\end{split}
\end{equation} 
\paragraph{\textbf{Sub-branch 4.3}} The beginning of the investigation of this sub-branch is identical
to that of the previous sub-branch.  The $[\P, \Q, \Q]$ and $[\B, \Q, \Q]$ identities produce 
\begin{equation}
	\begin{split}
		-N_3 &= \mp N_1 - N_1 \mp^\dagger \\
		N_4 &= \mb N_1 - N_1 \mb^\dagger
	\end{split} \qquad
	\begin{split}
		0 &= \beta \vt \mb N_3 \vt^\dagger
		+ \vt N_3 (\beta \vt \mb)^\dagger \\
		0 &= \pi \vt \mp N_4 \vt^\dagger
		+ \vt N_4 (\pi \vt \mp)^\dagger,
	\end{split}
\end{equation}
where the first two conditions give $N_3$ and $N_4$ the form
\begin{equation}
	N_3 = \begin{pmatrix}
		0 & c \bb_3\uu \\ - c \overbar{\bb_3\uu} & \bar{\rr}\bar{\uu}\bar{\bb_3}- \bb_3\uu\rr
	\end{pmatrix} \quad \text{and} \quad
	N_4 = \begin{pmatrix}
		0 & - c\bar{\bb_3} \\ c\bb_3 & \bb_3\rr - \bar{\rr}\bar{\bb_3}
	\end{pmatrix}.
\end{equation}
Substituting this $N_3$ with $\vt = (0, 1)$ into
\begin{equation}
	0 = \beta \vt \mb N_3 \vt^\dagger
		+ \vt N_3 (\beta \vt \mb)^\dagger,
\end{equation}
we acquire
\begin{equation}
	0 = 2 c |\bb_3|^2 \Im(\beta \uu) \quad \forall \beta \in \Im(\mathbb{H}).
\end{equation}
As, by assumption, $\bb_3 \neq 0$ and $\uu \neq 0$, this imposes $c=0$, such that
\begin{equation}
	N_1 = \begin{pmatrix}
		0 & \rr \\ \bar{\rr} & d
	\end{pmatrix} \quad 
	N_3 = \begin{pmatrix}
		0 & 0 \\ 0 & \bar{\rr}\bar{\uu}\bar{\bb_3}-\bb_3\uu\rr
	\end{pmatrix} \quad 
	N_4 = \begin{pmatrix}
		0 & 0 \\ 0 & \bb_3 \rr - \bar{\rr}\bar{\bb_3}
	\end{pmatrix}.
\end{equation}
With this form of $N_3$ and $N_4$,
\begin{equation}
	\begin{split}
		- \Re(\vt N_0\vt^\dagger) \pi &= \pi \vt \mp N_3 \vt^\dagger + \vt N_3 (\pi \vt \mp)^\dagger \\
		\Re(\vt N_0\vt^\dagger) \beta &= \beta \vt \mb N_4 \vt^\dagger + \vt N_4 (\beta \vt \mb)^\dagger, \\
	\end{split}
\end{equation}
have a vanishing R.H.S., showing that $N_0 = 0$.  This result contradicts our assumption
that $N_0 \neq 0$; therefore, this sub-branch does not contain any super-extensions.
\\
\paragraph{\textbf{Sub-branch 4.4}} Letting $N_0 = 0$ and $N_2 \neq 0$, we can use 
\begin{equation}
	\begin{split}
		-N_3 &= \mp N_1 - N_1 \mp^\dagger \\
		N_4 &= \mb N_1 - N_1 \mb^\dagger
	\end{split} \qquad 
	\begin{split}
		0 &= \beta \vt \mb N_4 \vt^\dagger + \vt N_4 (\beta \vt \mb)^\dagger \\
		0 &= \pi \vt \mp N_3 \vt^\dagger + \vt N_3 (\pi \vt \mp)^\dagger,
	\end{split}
\end{equation}
to again write
\begin{equation}
	N_1 = \begin{pmatrix}
		0 & \rr \\ \bar{\rr} & d
	\end{pmatrix} \quad 
	N_3 = \begin{pmatrix}
		0 & 0 \\ 0 & \bar{\rr}\bar{\uu}\bar{\bb_3} -\bb_3\uu\rr
	\end{pmatrix} \quad 
	N_4 = \begin{pmatrix}
		0 & 0 \\ 0 & \bb_3 \rr - \bar{\rr}\bar{\bb_3}
	\end{pmatrix}.
\end{equation}
Substituting these $N_i$ into 
\begin{equation}
	\begin{split}
	\tfrac12 [\beta, \vt N_2 \vt^\dagger] &= \beta \vt \mb N_3 \vt^\dagger
		+ \vt N_3 (\beta \vt \mb)^\dagger \\
	\tfrac12 [\pi, \vt N_2 \vt^\dagger] &= \pi \vt \mp N_4 \vt^\dagger
		+ \vt N_4 (\pi \vt \mp)^\dagger, 
	\end{split}
\end{equation}
the R.H.S. vanishes for both, showing $N_2 = 0$, contradicting our initial assumption in this 
sub-branch.
\\
\paragraph{\textbf{Sub-branch 4.5}} With none of the $N_i$ vanishing, we start again by writing $N_3$ and
$N_4$ in terms of $N_1$ using conditions from the $[\P, \Q, \Q]$ and $[\B, \Q, \Q]$ identities:
\begin{equation}
	N_3 = \begin{pmatrix}
		0 & c \overbar{\bb_3\uu} \\ - c \bb_3\uu & \bar{\rr}\bar{\uu}\bar{\bb_3}- \bb_3\uu\rr
	\end{pmatrix} \quad \text{and} \quad
	N_4 = \begin{pmatrix}
		0 & - c\bar{\bb_3} \\ c\bb_3 & \bb_3\rr - \bar{\rr}\bar{\bb_3}	
	\end{pmatrix}.	
\end{equation}
Letting
\begin{equation}
	N_2 = \begin{pmatrix}
		\nn & \mm \\ -\bar{\mm} & \ll
	\end{pmatrix} \quad \text{where} \quad \nn, \ll \in \Im(\mathbb{H}), \quad \mm \in \mathbb{H},
\end{equation}
we can use
\begin{equation} \label{eq:N2_n+_branch4_N2_condition}
	\tfrac12 [\beta, \vt N_2 \vt^\dagger] = \beta \vt \mb N_3 \vt^\dagger
		+ \vt N_3 (\beta \vt \mb)^\dagger
\end{equation}
to write $N_2$ in terms of $\bb_3$ and $\uu$.  First let $\vt = (0,1)$ to find
\begin{equation}
	\tfrac12 [\beta, \ll] = -c [ \beta, \bb_3 \uu\bar{\bb_3}] \quad \forall 0 \neq \beta \in \Im(\mathbb{H}).
\end{equation}
Therefore, 
\begin{equation}
	\ll = -2c \bb_3\uu\bar{\bb_3}.
\end{equation}
Next, substitute in $\vt = (1, 1)$ to get
\begin{equation}
	[\beta, 2\Im(\mm)] + \tfrac12 [\beta, \ll] = -c [\beta, \bb_3\uu\bar{\bb_3}].
\end{equation}
Using the previous result, this tells us that $\Im(\mm) = 0$.  Analogous calculations
with $\vt = (0, \ii)$ and $\vt = (1, \ii)$ show that, in fact, all of $\mm$ must vanish.
Finally, substituting in $\vt = (1, 0)$ into \eqref{eq:N2_n+_branch4_N2_condition}, we find $\nn = 0$ since the R.H.S. vanishes.  
Therefore, we are left with
\begin{equation}
	N_2 = \begin{pmatrix}
		0 & 0 \\ 0 & -2c\bb_3\uu\bar{\bb_3}
	\end{pmatrix}.
\end{equation}
We would have arrived at the same expression had we used $N_4$ and
\begin{equation}
	\tfrac12 [\pi, \vt N_2 \vt^\dagger] = \pi \vt \mp N_4 \vt^\dagger
		+ \vt N_4 (\pi \vt \mp)^\dagger.
\end{equation}
This form of $N_2$ automatically satisfies all other conditions it is involved in
from both the $[\B, \Q, \Q]$ and $[\P, \Q, \Q]$ identities.  Finally, we can
put this $N_2$ into 
\begin{equation}
	0 = \mh N_2 + N_2 \mh^\dagger
\end{equation}
to get
\begin{equation}
	0 = \hh_4 \ll + \ll \bar{\hh_4},
\end{equation}
where $\hh_4 = \bb_3(\hh_1 - \uu)\bb_3^{-1}$ and $\ll= -2c\bb_3\uu\bar{\bb_3}$.  Working through
some algebra, noting $\Re(\uu^2) = -1$ and the fact $c \neq 0$ for $N_2 \neq 0$, we arrive
at $\hh_1 \uu = \overbar{\hh_1\uu}$.  Since $\uu \in \Im(\mathbb{H})$, this forces $\hh_1 \in \Im(\mathbb{H})$
such that $\uu$ and $\hh_1$ are collinear. 
\\ \\
Now turn to $N_0$ and consider
\begin{equation}
	\begin{split}
		- \Re(\vt N_0\vt^\dagger) \pi &= \pi \vt \mp N_3 \vt^\dagger + \vt N_3 (\pi \vt \mp)^\dagger \\
		\Re(\vt N_0\vt^\dagger) \beta &= \beta \vt \mb N_4 \vt^\dagger + \vt N_4 (\beta \vt \mb)^\dagger.
	\end{split}
\end{equation}
Letting $\vt=(1, 0)$ in either of these conditions, we find that $a = 0$.  
Next, substituting $\vt = (0, 1)$ into the second condition produces
\begin{equation}
	- b = 2 c |\bb_3|^2.
\end{equation}
We would have arrived at the same result had we substituted into the first condition
and used the fact $|\uu|^2 = 1$.  Now substituting $\vt = (1, \ss)$ into the second condition, we find
\begin{equation}
	- 2 \Re(\ss\bar{\qq}) - b |\ss|^2 = 2c |\ss|^2 |\bb_3|^2.
\end{equation}
Therefore, using the previous result and letting $\ss = 1$, $\ss = \ii$, $\ss = \jj$
and $\ss = \kk$, we see that all components of $\qq$ must vanish.  All other conditions
on $N_0$ are now automatically satisfied, leaving
\begin{equation}
	N_0 = \begin{pmatrix}
		0 & 0 \\ 0 & - 2c |\bb_3|^2
	\end{pmatrix}.
\end{equation}
Equipped with these $N_i$, we can now analyse the condition from Lemma \ref{lem:N2_111}:
\begin{equation}
	\Re(\vt N_0\vt^\dagger) \vt \mh = \tfrac12 \vt N_2 \vt^\dagger \vt + \vt N_3 \vt^\dagger \vt \mb 
+ \vt N_4 s^\dagger \vt \mp.
\end{equation}
Letting $\vt = (0, 1)$:
\begin{equation} \label{eq:QQQ_branch4_n+}
	- 2c |\bb_3|^2 \begin{pmatrix}
		\hh_3 & \bb_3(\hh_1 - \uu)\bb_3^{-1}
	\end{pmatrix} = -c\bb_3\uu\bar{\bb_3} \begin{pmatrix}
	0 & 1
	\end{pmatrix} - \Im(\bb_3\uu\rr) \begin{pmatrix}
		\bb_3 & 0
	\end{pmatrix} + \Im(\bb_3\rr) \begin{pmatrix}
		\bb_3\uu & 0
	\end{pmatrix}.
\end{equation}
Concentrating on the second component, we have
\begin{equation}
	-2 c |\bb_3|^2 \bb_3 (\hh_1 - \uu) \bb_3^{-1} = -c \bb_3 \uu \bar{\bb_3}.
\end{equation}
Using the fact $|\bb_3|^2 \bb_3 = \bar{\bb_3}$ and cancelling relevant terms leaves 
\begin{equation}
	2 \bb_3\hh_1\bar{\bb_3} = 0.
\end{equation}
Since, by assumption $\bb_3 \neq 0$, we get $\hh_1 = 0$.  The first component of 
\eqref{eq:QQQ_branch4_n+} gives us a prescription for $\hh_3$,
\begin{equation}
	-2 c |\bb_3|^2 \hh_3 = - 2 \Im(\bb_3\uu\rr) \bb_3 + 2 \Im(\bb_3\rr) \bb_3\uu,
\end{equation}
therefore, we can fully describe $\mh$ in terms of $\mb$, $\mp$, and $N_1$.
\\ \\
The final conditions to consider are those from the $[\bH, \Q, \Q]$ super-Jacobi identity for $N_1$, $N_3$
and $N_4$.  First, the $N_1$ condition tell us
\begin{equation}
	0 = c\bar{\hh_3} + \rr \overbar{\bb_3\uu\bb_3^{-1}} \qquad 0 = \Re(\hh_3\rr).
\end{equation}
Notice that the second constraint here is automatically satisfied by the first, since $c \neq 0$ for a non-vanishing $N_0$.
Substituting this expression for $\hh_3$ into the previous prescription, we find
\begin{equation} \label{eq:h4}
	|\rr|^2 |\bb_3|^2 \bb_3 \uu \bb_3^{-1} = [\Im(\bb_3\rr), \Im(\bb_3\uu\rr)] + \Re(\bb_3\uu\rr) \Im(\bb_3\rr)
	- \Re(\bb_3\rr) \Im(\bb_3\uu\rr).
\end{equation}
Now, the constraints that the $N_3$ and $N_4$ conditions produce are the ones given in
sub-branch \hyperlink{N2_n+_42iii}{$\mathsf{4.2.iii}$}:
\begin{equation}
	\Im(\bb_3\rr) = - [\hh_4, [\hh_4, \Im(\bb_3\rr)]] \quad \text{and} \quad 
	\Im(\bb_3\uu\rr) = - [\hh_4, [\hh_4, \Im(\bb_3\uu\rr)]].
\end{equation}
These tell us that $\Im(\bb_3\rr)$ and $\Im(\bb_3\uu\rr)$ are perpendicular to $\hh_4$ in
$\Im(\mathbb{H})$.  Therefore, the expression in \eqref{eq:h4} becomes
\begin{equation}
	0 = \Re(\bb_3\rr), \quad 0 = \Re(\bb_3\uu\rr) \quad \text{and} \quad
	|\rr|^2 \bb_3 \uu \bar{\bb_3} = [\Im(\bb_3\rr), \Im(\bb_3\uu\rr)].
\end{equation} 
Putting all of these constraints together, we have
\begin{equation}
	\begin{split}
	\mathcal{C}_{\hat{\n}_+,\,\mathsf{4.5.iii}} = \{
		\uu^2 &= -1, \quad
		\Im(\bb_3\rr) = - [\hh_4, [\hh_4, \Im(\bb_3\rr)]], \quad
		\Im(\bb_3\uu\rr) = - [\hh_4, [\hh_4, \Im(\bb_3\uu\rr)]], \\
		0 &= \Re(\bb_3\rr), \quad
		0 = \Re(\bb_3\uu\rr), \quad
		|\rr|^2 \bb_3 \uu \bar{\bb_3} = [\Im(\bb_3\rr), \Im(\bb_3\uu\rr)] \},
	\end{split}
\end{equation}
for non-vanishing \hypertarget{N2_n+_45iii}{matrices}
\begin{equation}
	\begin{split}
	\mathcal{M}_{\hat{\n}_+,\,\mathsf{4.5.iii}} = \Big\{
		\mb &= \begin{pmatrix}
			0 & 0 \\ \bb_3 & 0
		\end{pmatrix}, \quad
		\mp = \begin{pmatrix}
			0 & 0 \\ \bb_3\uu & 0
		\end{pmatrix}, \quad
		\mh = \begin{pmatrix}
			0 & 0 \\ -c^{-1} \bb_3\uu\bb_3^{-1} \bar{\rr} & \bb_3 \uu \bb_3^{-1}
		\end{pmatrix}, \quad
		N_0 = \begin{pmatrix}
		0 & 0 \\ 0 & - 2c |\bb_3|^2
		\end{pmatrix}, \\
		N_1 &= \begin{pmatrix}
			c & \rr \\ \bar{\rr} & d
		\end{pmatrix}, \quad
		N_2 = \begin{pmatrix}
		0 & 0 \\ 0 & -2c\bb_3\uu\bar{\bb_3}
		\end{pmatrix}, \quad 
		N_3 = \begin{pmatrix}
		0 & c \overbar{\bb_3\uu} \\ - c \bb_3\uu & - 2 \Im(\bb_3\uu\rr)
		\end{pmatrix}, \quad 
		N_4 = \begin{pmatrix}
		0 & - c\bar{\bb_3} \\ c\bb_3 & 2 \Im(\bb_3\rr)	
		\end{pmatrix} \Big\}.
		\end{split}
\end{equation}
To demonstrate that there are super-extensions in this sub-branch, we will first simplify this system by letting parameters vanish where possible.  In particular, $\rr$ and $d$ in $N_1$ may be set to zero.  This reduces $\mathcal{C}_{\hat{\n}_+,\,\mathsf{4.5.iii}}$ to contain only $\uu^2 = -1$.  Now we can use the endomorphisms of $\s_{\bar{1}}$ to impose $\bb_3 = \ii$, $\uu = \jj$ and $c=1$.  With these choices, the matrices become
\begin{equation}
	\begin{split}
	\mb &= \begin{pmatrix}
			0 & 0 \\ \ii & 0
		\end{pmatrix}, \quad
		\mp = \begin{pmatrix}
			0 & 0 \\ \kk & 0
		\end{pmatrix}, \quad
		\mh = \begin{pmatrix}
			0 & 0 \\ 0 & \jj
		\end{pmatrix}, \quad
		N_0 = \begin{pmatrix}
		0 & 0 \\ 0 & - 2
		\end{pmatrix}, \\
		N_1 &= \begin{pmatrix}
			1 & 0 \\ 0 & 0
		\end{pmatrix}, \quad
		N_2 = \begin{pmatrix}
		0 & 0 \\ 0 & 2\jj
		\end{pmatrix}, \quad 
		N_3 = \begin{pmatrix}
		0 & -\kk \\ -\kk & 0
		\end{pmatrix}, \quad 
		N_4 = \begin{pmatrix}
		0 & \ii \\ \ii & 0
		\end{pmatrix} .
		\end{split}
\end{equation}
As there are no restrictions on the parameter $d$, we may introduce it without affecting our other parameter choices; however, this is not the case for $\rr$.  There are several constraints in $\mathcal{C}_{\hat{\n}_+,\,\mathsf{4.5.iii}}$ involving $\rr$; therefore, we need to examine these constraints to determine whether new parameters must be chosen.  Interrogating 
\begin{equation}
		\Im(\bb_3\rr) = - [\hh_4, [\hh_4, \Im(\bb_3\rr)]] \quad \text{and} \quad 
		\Im(\bb_3\uu\rr) = - [\hh_4, [\hh_4, \Im(\bb_3\uu\rr)]]
\end{equation} 
with the parameter choices stated above, we find that $\rr$ must vanish.  In particular, due to $\hh_4 = \bb_3 \uu \bb_3^{-1}$ having unit length, we cannot replicate the analysis of sub-branch \hyperlink{N2_n+_42iii}{$\mathsf{4.2.iii}$}, where the magnitude of $\hh_4$ was necessarily either $+\tfrac12$ or $-\tfrac12$.  Thus, we cannot produce a super-extension in this sub-branch for which $\rr \neq 0$.  This simplifies the above $(\mathcal{M}, \mathcal{C})$, such that the remaining constraints are
\begin{equation}
	\mathcal{C}_{\hat{\n}_+,\,\mathsf{4.5.iii}} = \{ \uu^2 = -1 \},
\end{equation}
and the non-vanishing matrices are now 
\begin{equation}
	\begin{split}
	\mathcal{M}_{\hat{\n}_+,\,\mathsf{4.5.iii}} = \Big\{
		\mb &= \begin{pmatrix}
			0 & 0 \\ \bb_3 & 0
		\end{pmatrix}, \quad
		\mp = \begin{pmatrix}
			0 & 0 \\ \bb_3\uu & 0
		\end{pmatrix}, \quad
		\mh = \begin{pmatrix}
			0 & 0 \\ 0 & \bb_3 \uu \bb_3^{-1}
		\end{pmatrix}, \quad
		N_0 = \begin{pmatrix}
		0 & 0 \\ 0 & - 2c |\bb_3|^2
		\end{pmatrix}, \\
		N_1 &= \begin{pmatrix}
			c & 0 \\ 0 & d
		\end{pmatrix}, \quad
		N_2 = \begin{pmatrix}
		0 & 0 \\ 0 & -2c\bb_3\uu\bar{\bb_3}
		\end{pmatrix}, \quad 
		N_3 = \begin{pmatrix}
		0 & c \overbar{\bb_3\uu} \\ - c \bb_3\uu & 0
		\end{pmatrix}, \quad 
		N_4 = \begin{pmatrix}
		0 & - c\bar{\bb_3} \\ c\bb_3 & 0
		\end{pmatrix} \Big\}.
	\end{split}
\end{equation}
\paragraph{$\hat{\g}$}~\\ \\  Finally, substitute $\lambda = \eta = \varepsilon = 0$ and $\mu = -1$ into the second constraint in 
\eqref{eq:N2_branch4_condition} to investigate the $\hat{\g}$ case.  We
find
\begin{equation}
	[\vv, \hh_1] = -1,
\end{equation}
which, since $[\vv, \hh_1] \in \Im(\mathbb{H})$, is inconsistent.  Therefore, we cannot
get a super-extension of $\hat{\g}$ in this branch.
\subsection{Summary} \label{subsec:N2_summary}
Table \ref{tab:N2_classification} lists all the sub-branches of $\cS$ we found that contain $\N=2$ generalised Bargmann superalgebras.  Each Lie superalgebra in one of these branches is an $\N=2$ super-extension of one of the generalised Bargmann algebras given in Table \ref{tab:algebras}, taken from \cite{Figueroa-OFarrill:2017ycu}.  It is interesting that $\bZ$ only appears in
\begin{equation}
	[\B, \P] = \bZ \quad \text{and} \quad [\Q, \Q] = \bZ.
\end{equation}
Therefore, in all instances, $\bZ$ remains central after the super-extension.  In particular, this means that we may always find a kinematical Lie superalgebra (without a central-extension) by taking the quotient of our generalised Bargmann superalgebra $\s$ by the $\mathbb{R}$-span of $\bZ$, $\s / \langle \bZ \rangle$.
\begin{table}[h!]
  \centering
  \caption{Sub-branches of $\N=2$ generalised Bargmann superalgebras (with $[\Q,\Q]\neq 0$)}
  \label{tab:N2_classification}
  \setlength{\extrarowheight}{2pt}
  \rowcolors{2}{blue!10}{white}
    \begin{tabular}{l|l*{5}{|>{$}c<{$}}}\toprule
      \multicolumn{1}{c|}{(S)B} & \multicolumn{1}{c|}{$\k$}& \multicolumn{1}{c|}{$\mh$}& \multicolumn{1}{c|}{$\mz$}& \multicolumn{1}{c|}{$\mb$} & \multicolumn{1}{c|}{$\mp$} & \multicolumn{1}{c|}{$[\Q,\Q]$}\\
      \toprule
      \hyperlink{N2_a_1i}{$\mathsf{1.i}$} & \hyperlink{a}{$\hat{\a}$}  & \checkmark & & & &  \bZ  \\
      \hyperlink{N2_a_1ii}{$\mathsf{1.ii}$}  & \hyperlink{a}{$\hat{\a}$} & & & &  & \bH + \bZ \\
      \hyperlink{N2_a_21ii}{$\mathsf{2.1.ii}$} & \hyperlink{a}{$\hat{\a}$} & & & & \checkmark &  \bH  \\
      \hyperlink{N2_a_22i}{$\mathsf{2.2.i}$} & \hyperlink{a}{$\hat{\a}$} & \checkmark & & & \checkmark &  \bZ  \\
      \hyperlink{N2_a_23ii}{$\mathsf{2.3.i}$} & \hyperlink{a}{$\hat{\a}$}  & \checkmark & & & \checkmark &  \bZ + \B \\
      \hyperlink{N2_a_23i}{$\mathsf{2.3.ii}$} & \hyperlink{a}{$\hat{\a}$}  & & & & \checkmark & \bH + \bZ + \B \\
      \hyperlink{N2_a_45iii}{$\mathsf{4.5.iii}$} & \hyperlink{a}{$\hat{\a}$} & \checkmark & & \checkmark & \checkmark & \bH + \bZ + \B + \P \\
      \hyperlink{N2_ng_1i}{$\mathsf{1.i}$}  & \hyperlink{n-}{$\hat{\n}_-$} & \checkmark & & & &  \bZ  \\
      \hyperlink{N2_n-_22i}{$\mathsf{2.2.i}$} & \hyperlink{n-}{$\hat{\n}_-$}  & \checkmark & & & \checkmark &  \bZ  \\
      \hyperlink{N2_n-_23i}{$\mathsf{2.3.i}$} & \hyperlink{n-}{$\hat{\n}_-$}  & \checkmark & & & \checkmark &  \bZ + \P  \\
      \hyperlink{N2_n-_32i}{$\mathsf{3.2.i}$} & \hyperlink{n-}{$\hat{\n}_-$}  & \checkmark & & \checkmark & &  \bZ  \\ 
      \hyperlink{N2_n-_33i}{$\mathsf{3.3.i}$} & \hyperlink{n-}{$\hat{\n}_-$}  & \checkmark & & \checkmark & &  \bZ + \P \\
      \hyperlink{N2_ng_1i}{$\mathsf{1.i}$} & \hyperlink{n+}{$\hat{\n}_+$}  & \checkmark & & & &  \bZ  \\      
      \hyperlink{N2_ng_22i}{$\mathsf{2.2.i}$} & \hyperlink{n+}{$\hat{\n}_+$}  & \checkmark & & & \checkmark &  \bZ  \\ 
      \hyperlink{N2_n+_41i}{$\mathsf{4.1.i}$} & \hyperlink{n+}{$\hat{\n}_+$} & \checkmark & & \checkmark & \checkmark &  \bZ  \\      
       \hyperlink{N2_n+_42iii}{$\mathsf{4.2.iii}$} & \hyperlink{n+}{$\hat{\n}_+$} & \checkmark & & \checkmark & \checkmark &  \bZ + \B + \P  \\
      \hyperlink{N2_n+_45iii}{$\mathsf{4.5.iii}$} & \hyperlink{n+}{$\hat{\n}_+$}  & \checkmark & & \checkmark & \checkmark &  \bH + \bZ + \B + \P \\
      \hyperlink{N2_ng_1}{$\mathsf{1.i}$} & \hyperlink{g}{$\hat{\g}$} & \checkmark & & & &  \bZ  \\  
      \hyperlink{N2_ng_2i}{$\mathsf{2.2.i}$} & \hyperlink{g}{$\hat{\g}$}  & \checkmark & & & \checkmark &  \bZ  \\  
      \hyperlink{N2_g_32i}{$\mathsf{3.2.i}$} & \hyperlink{g}{$\hat{\g}$} & \checkmark & & \checkmark & &  \bZ  \\         
      \hyperlink{N2_g_33i}{$\mathsf{3.3.i}$} & \hyperlink{g}{$\hat{\g}$} & \checkmark & & \checkmark & &  \bZ + \P  \\  
       \hyperlink{N2_g_33ii}{$\mathsf{3.3.ii}$}  & \hyperlink{g}{$\hat{\g}$} & & & \checkmark & &  \bH + \bZ + \P  \\  
      \bottomrule
    \end{tabular}
    \caption*{The first column indicates the sub-branch of generalised Bargmann superalgebras,
    so that the reader may navigate back to find the conditions on the non-vanishing parameters
    of these superalgberas.  The second column then tells us the underlying generalised Bargmann algebra $\k$.
    The next four columns tells us which of the $\s_{\bar{0}}$ generators $\bH,\bZ, \B$, 
    and $\P$ act on $\Q$.  Recall, $\J$ necessarily acts on $\Q$, so we do not need to state this explicitly.
    The final column shows which $\s_{\bar{0}}$ generators occur in the $[\Q, \Q]$ bracket.
     }
\end{table}
\subsubsection{Unpacking the Notation}
Although the formalism employed in this classification was useful for our purposes, it may be unfamiliar to the reader.  Therefore, in this section, we will convert one of the $\N=2$ super-extensions of the Bargmann algebra in sub-branch \hyperlink{N2_g_33ii}{$\mathsf{3.3.ii}$} into a more standard notation.  The $\s_{\bar{0}}$ brackets have already been discussed in section \ref{subsec:FI_KLSAs}, so we will concentrate solely on the $[\s_{\bar{0}}, \s_{\bar{1}}]$ and $[\s_{\bar{1}}, \s_{\bar{1}}]$ brackets,
\begin{equation}
	[\sB(\beta), \sQ(\vt)] = \sQ(\beta \vt \mb) \quad \text{and} \quad [\sQ(\vt), \sQ(\vt)] = \Re(\vt N_0 \vt^\dagger) \sH + \Re(\vt N_1 \vt^\dagger) \sZ - \sP(\vt N_4 \vt^\dagger),
\end{equation}
where 
\begin{equation}
	\mb = \begin{pmatrix}
		0 & 0 \\ \bb_3 & 0
	\end{pmatrix} \quad
	N_0 = \begin{pmatrix}
		0 & 0 \\ 0 & -2c |\bb_3|^2
	\end{pmatrix} \quad
	N_1 = \begin{pmatrix}
		c & \rr \\ \bar{\rr} & d
	\end{pmatrix} \quad
	N_4 = \begin{pmatrix}
		0 & -c\bar{\bb_3} \\ c\bb_3 & \bb_3 \rr - \bar{\rr} \bar{\bb_3}
	\end{pmatrix}.
\end{equation}
Let $\{\sQ^1_a\}$ be a real basis for the first $\so(3)$ spinor module in $\s_{\bar{1}} = S^1 \oplus S^2$ where $a \in \{1, 2, 3, 4\}$ , and $\{\sQ^2_a\}$ be a basis for the second $\so(3)$ spinor module.  Letting $\vt = (\theta_1, \theta_2)$, and substituting the above matrices into the $[\s_{\bar{0}}, \s_{\bar{1}}]$ bracket, we get
\begin{equation}
		[\sB(\beta), \sQ^1(\theta_1)] = 0 \quad \text{and}  \quad [\sB(\beta), \sQ^2(\theta_2)] = \sQ^1(\beta \theta_2 \bb_3).
\end{equation}
Substituting $\vt = \vt'= (\theta_1, 0)$, $\vt = (\theta_1, 0)$ and $\vt'=(0,\theta_2)$, and $\vt=\vt'=(0,\theta_2)$ into the $[\s_{\bar{1}}, \s_{\bar{1}}]$ bracket we find
\begin{equation}
	\begin{split}
		[\sQ^1(\theta_1), \sQ^1(\theta_1)] &= c |\theta_1|^2 \sZ \\
		[\sQ^1(\theta_1), \sQ^2(\theta_2)] &= \Re(\theta_1\rr\bar{\theta_2})Z -\tfrac{c}{2} \sP(\theta_2\bb_3\bar{\theta_1}-\theta_1\bar{\bb_3}\bar{\theta_2}) \\ 
		[\sQ^2(\theta_2), \sQ^2(\theta_2)] &= -2c |\bb_3|^2 |\theta_2|^2 \sH + d |\theta_2|^2 \sZ - \sP(\theta_2 (\bb_3 \rr - \bar{\rr}\bar{\bb_3})\bar{\theta_2}).
	\end{split}
\end{equation} 
For the purposes of the present example, we will set the parameters of this super-extension as specified in \eqref{eq:notation_example}; therefore, we have $[\s_{\bar{0}}, \s_{\bar{1}}]$ brackets
\begin{equation}
		[\sB(\beta), \sQ^1(\theta_1)] = \sQ^2(\beta \theta_2 \ii) \quad \text{and}  \quad [\sB(\beta), \sQ^2(\theta_2)] = 0,
\end{equation}
and $[\s_{\bar{1}}, \s_{\bar{1}}]$ brackets
\begin{equation}
		[\sQ^1(\theta_1), \sQ^1(\theta_1)] = |\theta_1|^2 \sZ, \quad
		[\sQ^1(\theta_1), \sQ^2(\theta_2)] = -\tfrac{1}{2} \sP(\theta_2\ii\bar{\theta_1}+\theta_1\ii\bar{\theta_2}) \quad \text{and} \quad
		[\sQ^2(\theta_2), \sQ^2(\theta_2)] = |\theta_2|^2 \sH.
\end{equation}
Now, we can write
\begin{equation}
	[\bB_i, \bQ^2_a] = \sum_{b = 1}^4 \bQ^1_b \tensor{\beta}{_i^b_a} \quad [\bQ^1_a, \bQ^1_b] = \delta_{ab} \bZ
	\quad [\bQ^1_a, \bQ^2_b] = \sum_{i = 1}^3 \bP_i \Gamma_{ab}^i, \quad \quad [\bQ^2_a, \bQ^2_b] = \delta_{ab} \bH.
\end{equation}
Our brackets then produce the $\beta_i$ matrices
\begin{equation}
	\beta_1 = \begin{pmatrix}
		 -\mathbb{1} & 0 \\ 0 & \mathbb{1}
	\end{pmatrix} \quad
	\beta_2 = \begin{pmatrix}
		0 & - \sigma_1 \\ - \sigma_1 & 0
	\end{pmatrix} \quad
	\beta_3 = \begin{pmatrix}
		0 & \sigma_3 \\ \sigma_3 & 0
	\end{pmatrix},
\end{equation}
and the symmetric $\Gamma^i$ matrices
\begin{equation}
	\Gamma^1 = \begin{pmatrix}
		-\mathbb{1} & 0  \\ 0 & \mathbb{1}
	\end{pmatrix} \quad
	\Gamma^2 = \begin{pmatrix}
		0 & -\sigma_1 \\ -\sigma_1 & 0 
	\end{pmatrix} \quad
	\Gamma^3 = \begin{pmatrix}
		0 & \sigma_3 \\ \sigma_3 & 0
	\end{pmatrix},
\end{equation}
where $\sigma_1$ and $\sigma_3$ are the first and third Pauli matrix, respectively.  This $\N=2$ Bargmann superalgebra takes the same form as the $(2+1)$-dimensional Bargmann superalgebra utilised in \cite{Andringa:2013mma}.
\section{Discussion} \label{sec:discussion}
In this paper, we classified the $\N=1$ super-extensions of the generalised Bargmann algebras with three-dimensional spatial isotropy up to isomorphism.  We also presented the non-empty sub-branches of the variety $\cS$ describing the $\N=2$ super-extensions of the generalised Bargmann algebras.  To simplify this classification problem, we utilised a quaternionic formalism such that $\so(3)$ scalar modules were described by copies of $\mathbb{R}$, $\so(3)$ vector modules were described by copies of $\Im(\mathbb{H})$, and $\so(3)$ spinor modules were described by copies of $\mathbb{H}$.  We began by defining a universal generalised Bargmann algebra, which, under the appropriate setting of some parameters, may be reduced to the centrally-extended static kinematical Lie algebra $\hat{\a}$, the centrally-extended Newton-Hooke algebras $\hat{\n}_\pm$, or the Bargmann algebra $\hat{\g}$.  The most general form for the $[\s_{\bar{0}}, \s_{\bar{1}}]$ and $[\s_{\bar{1}}, \s_{\bar{1}}]$ bracket components were found before substituting them into the super-Jacobi identities and finding the constraints on the parameters for these maps.  Because of the formalism in use, solving these constraints amounted to some linear algebra over the quaternions, and paying attention to the allowed basis transformations $\G \subset \GL(\s_{\bar{0}}) \times \GL(\s_{\bar{1}})$.  Since we are only interested in supersymmetric extensions of these algebras, we limited ourselves to those branches which allow for non-vanishing $[\Q, \Q]$.  The results of the $\N=1$ and $\N=2$ analyses are in Tables \ref{tab:N1_classification} and \ref{tab:N2_classification}, respectively.  We found 9 isomorphism classes in the $\N=1$ case, and 22 non-empty sub-branches in the $\N=2$ case.  
\\ \\
These classifications have a few interesting features.  The $\N=1$ classification showed that if we centrally-extended a kinematical Lie algebra before finding its super-extensions, we will generally obtain very different results than if we super-extended the algebra before finding its central extensions.   It would be interesting to investigate whether there are any special properties of those generalised Bargmann superalgebras which can arise through both procedures. 
\\ \\
Although, not particularly interesting in the $\N=1$ case, the double complex structure found in which $\mb$ and $\mp$ act as the differentials on modules $\s_{\bar{1}}$ may be interesting to study for $\N \geq 2$.  Notice that this interpretation is possible due to the kinematical Lie algebras having $[\B, \B] = [\P, \P] = 0$; therefore, we will have a double complex on any Lie superalgebra with these brackets.  In particular, all non-relativistic and ultra-relativistic kinematical Lie algebras have these brackets: this structure will be present for all investigations into extended supersymmetry for Galilean and Carrollian spacetimes.  Thus a full understanding of this structure may shed some light on extended supersymmetry in these regimes. 
\\ \\
As with the classification of the $\N=1$ super-extensions of the non-centrally-extended kinematical Lie algebras in \cite{Figueroa-OFarrill:2019ucc}, both the $\N=1$ classification and $\N=2$ branch analysis presented here demonstrate that each generalised Bargmann algebra can have many possible super-extensions.  With the exception of the superalgebras of sub-branch \hyperlink{N2_a_21ii}{$\mathsf{2.1.ii}$} for the centrally-extended static KLA $\hat{\a}$, every super-extension has the central extension $\bZ$ in the $[\Q, \Q]$ component of the bracket.\footnote{Non-relativistic superalgebras with $[\Q, \Q] = \bZ$ have been considered in papers such as \cite{Meyer:2017zfg, CLARK198491}.} Therefore, understanding the significance of the supersymmetry containing this internal component appears to be vital for a full understanding of these generalised Bargmann superalgebras and their phenomenology. 
\\ \\
A clear next step in this project would be to determine the automorphism groups for each of the generalised Bargmann superalgebras to determine their admissible Lie super pairs as was done for the kinematical Lie superalgebras without central extension in \cite{Figueroa-OFarrill:2019ucc}.  With this information, we could then classify the possible generalised Bargmann superspaces in $(3+1)$-dimensions, and, \textit{"superising"} the work in \cite{Figueroa-OFarrill:2019sex}, find the geometric properties of these spaces, such as the invariant structures and their associated symmetries. 
\\ \\
In this paper, we restricted ourselves to only the first section in Table \ref{tab:ce_algebras}; therefore, another obvious extension to the current work would be to classify the super-extensions for the algebras in the other two sections.  One further direction of investigation would be to classify the generalised Bargmann superalgebras for different dimensions $D$ and higher $\N$.  Notice that the formalism introduced in section \ref{subsec:N2_setup} could easily be recycled for investigations into $D=3$ kinematical Lie superalgebras with $\N>2$.  Wanting to explore $\N$-extended super-extensions, we would only change the size of the $\s_{\bar{1}}$ vector space and the quaternionic matrices, ensuring we were working with $\s_{\bar{1}} = \bigoplus_{i=1}^{\N} S^i$, where $S^i$ is a copy of the four-dimensional real $\so(3)$ spinor module, and $\mh, \mz, \mb, \mp \in \Mat_{\N}(\mathbb{H})$. It may also be interesting to look at $D=2$ due to the connection with Chern-Simons theories (see \cite{deAzcarraga:2002xi, Gomis:2019nih}), and to determine the invariant tensors for the Lie superalgebras presented in this paper and try to map them down to the invariant tensors for the Lie superalgebras in one dimension lower.
\\ \\
In addition to introducing more $\so(3)$ spinor modules into the $\s_{\bar{1}}$ vector space as suggested above, we could look to extend the $\s_{\bar{0}}$ vector space.  Having demonstrated how one may incorporate the extra $\so(3)$ scalar generator $\bZ$ into the underlying vector space, we could introduce the additional generators to consider classifying Maxwell superalgebras, or graded conformal  Lie superalgebras, extending the work of \cite{Figueroa-OFarrill:2018ygf} in the latter case. \\
\begin{table}[h!]
  \centering
  \caption{Supergravity Algebras}
  \label{tab:SG_algebras}
  \setlength{\extrarowheight}{2pt}
  \rowcolors{2}{blue!10}{white}
    \begin{tabular}{l|l*{6}{|>{$}c<{$}}}\toprule
      \multicolumn{1}{c|}{S} & \multicolumn{1}{c|}{$\k$} &  \multicolumn{1}{c|}{(S)B} & \multicolumn{1}{c|}{$\bH$}& \multicolumn{1}{c|}{$\bZ$}& \multicolumn{1}{c|}{$\B$} & \multicolumn{1}{c|}{$\P$} & \multicolumn{1}{c|}{$[\Q,\Q]$}\\
      \toprule
      3 & \hyperlink{a}{$\hat{\a}$} & - & \checkmark & & & & \bH + \bZ \\
      - & \hyperlink{a}{$\hat{\a}$} & \hyperlink{N2_a_1ii}{\mathsf{1.ii}} & & & &  & \bH + \bZ \\
      - & \hyperlink{a}{$\hat{\a}$} & \hyperlink{N2_a_21ii}{\mathsf{2.1.ii}} & & & & \checkmark &  \bH  \\
      - & \hyperlink{a}{$\hat{\a}$} & \hyperlink{N2_a_23i}{\mathsf{2.3.i}} & \checkmark & & & \checkmark &  \bZ + \B \\
      - & \hyperlink{a}{$\hat{\a}$} & \hyperlink{N2_a_23ii}{\mathsf{2.3.ii}} & & & & \checkmark & \bH + \bZ + \B \\
      - & \hyperlink{a}{$\hat{\a}$} & \hyperlink{N2_a_45iii}{\mathsf{4.5.iii}} & \checkmark & & \checkmark & \checkmark & \bH + \bZ + \B + \P \\
      - & \hyperlink{n-}{$\hat{\n}_-$} & \hyperlink{N2_n-_23i}{\mathsf{2.3.i}} & \checkmark & & & \checkmark &  \bZ + \P  \\
      - & \hyperlink{n-}{$\hat{\n}_-$}& \hyperlink{N2_n-_33i}{\mathsf{3.3.i}} & \checkmark & & \checkmark & &  \bZ + \P \\
      - & \hyperlink{n+}{$\hat{\n}_+$} & \hyperlink{N2_n+_42iii}{\mathsf{4.2.iii}} & \checkmark & & \checkmark & \checkmark &  \bZ + \B + \P  \\
      - & \hyperlink{n+}{$\hat{\n}_+$} & \hyperlink{N2_n+_45iii}{\mathsf{4.5.iii}} & \checkmark & & \checkmark & \checkmark &  \bH + \bZ + \B + \P \\
      - & \hyperlink{g}{$\hat{\g}$} & \hyperlink{N2_g_33i}{\mathsf{3.3.i}} & \checkmark & & \checkmark & &  \bZ + \P  \\  
      - & \hyperlink{g}{$\hat{\g}$} & \hyperlink{N2_g_33ii}{\mathsf{3.3.ii}} & & & \checkmark & &  \bH + \bZ + \P  \\  
      \bottomrule
    \end{tabular}
    \caption*{The first column gives the unique identifier for the $\N=1$  generalised Bargmann superalgebras $\s$,
    and the second column tells us the underlying generalised Bargmann algebra $\k$.
    For the $\N=2$ superalgebras, the third column indicates the sub-branch the generalised Bargmann superalgebra comes from,
    so that the reader may navigate back to find the conditions on the non-vanishing parameters
    of the superalgebra.  The next four columns tells us which of the $\s_{\bar{0}}$ generators $\bH, \bZ, \B$, 
    and $\P$ act on $\Q$.  Recall, $\J$ necessarily acts on $\Q$, so we do not need to state this explicitly.
    The final column shows which $\s_{\bar{0}}$ generators occur in the $[\Q, \Q]$ bracket.
     }
\end{table}~\\
For future research, it is interesting that many of these superalgebras can be gauged to produce supergravity theories in $(3+1)$-dimensions.  Following the relativistic case, for a superalgebra to gauge to a non-trivial supergravity theory, we must have the commutator of two (local) supersymmetry transformations producing a (local) spacetime translation; therefore, $[\Q, \Q]$ must have $\bH$, $\P$, or $\bH$ and $\P$ on the right-hand-side.  A list of the sub-branches in which such generalised Bargmann superalgebras are found is given in Table \ref{tab:SG_algebras}.  
\\ \\
Another exciting application of the superalgebras classified here lies in possible holographic dualities.  There is an extensive literature investigating non-relativistic holography (see \cite{Hartong:2014oma, Christensen:2013rfa, Hartong:2015wxa, Hartong:2015zia} for a few examples), in which torsional Newton-Cartan (TNC) geometries are dual to Ho\v{r}ava-Lifshitz gravity.  On the Newton-Cartan side of this duality, we have some geometry that may be obtained by gauging a non-relativistic algebra \cite{Bergshoeff:2014uea}; therefore, we may consider taking any of the generalised Bargmann algebras as our starting point. To obtain the dual theory, we may use the following procedure.  Let $\phi$ be an isomorphism which exchanges the two $\so(3)$ scalar modules $\bH$ and $\bZ$, and exchanges the two $\so(3)$ vector modules $\B$ and $\P$.  To be more explicit, define $\phi$ as
\begin{equation}
	\begin{split}
	\phi &: \s \rightarrow \s'
	\end{split} \quad \text{such that} \quad 
	\begin{split}
		\phi(\bH) &= \bZ \\
		\phi(\bZ) &= \bH
	\end{split} \quad
	\begin{split}
		\phi(\bB_i) &= \bP_i \\
		\phi(\bP_i) &= \bB_i \\
	\end{split} \quad
		\phi(\bQ^A_a) = \bQ^A_a, \nonumber
\end{equation}
where $i$ runs over $\so(3)$ vector indices $\{1, 2, 3\}$, $a$ runs over $\so(3)$
spinor indices $\{1, 2, 3, 4\}$, and $A$ labels our spinor module $\{1, 2\}$.
Focussing solely on the generalised Bargmann superalgebras described in Table \ref{tab:holographic_algebras},
the $\N=1$ cases have two brackets in common:
\begin{equation}
	[\sB(\beta), \sP(\pi)] = \Re(\bar{\beta}\pi) Z \quad \text{and} \quad
		[\sQ(s), \sQ(s)] = |s|^2 \sZ. \nonumber
\end{equation}
For the $\N=2$ superalgebras, these two brackets are very similar:
\begin{equation}
	[\sB(\beta), \sP(\pi)] = \Re(\bar{\beta}\pi) \sZ \quad \text{and} \quad
		[\sQ(s), \sQ(s)] = \Re(s N_1 s^\dagger) \sZ, \nonumber
\end{equation}
where $N_1^\dagger = N_1$.  Notice that by sending $\bZ$ to $\bH$, the two $\s$
brackets including the mass generator $\bZ$ become $[\B, \P] = \bH$ and $[\Q, \Q] = \bH$;
thus, under this isomorphism, we obtain Carroll superalgebras from these 
generalised Bargmann superalgebras.\footnote{This duality was also recognised in the non-supersymmetric context in \cite{Duval:2014uoa}.}   By exchanging the $\so(3)$ vector modules $\B$ and $\P$ as well, we have the following interpretation for the generalised Bargmann algebras. The centrally-extended static kinematical Lie algebra $\hat{\a}$, with only $[\B, \P] = \bZ$, becomes the Carroll algebra.  The centrally-extended dS Galilean algebra $\hat{\n}_-$ becomes
\begin{equation}
	[\B, \P] = \bH, \quad [\bZ, \B] = -\B, \quad [\bZ, \P] = \P. \nonumber
\end{equation}
Interpreting $\bZ$ as a dilatation, these are the brackets for the Carroll Lifshitz algebra first discussed in \cite{Gibbons:2009me}. In fact, our $\N=1$ extension of this algebra is very similar to the superalgebra found to describe the Carroll superparticle in \cite{Bergshoeff:2015wma}.  The only difference being that, in our context, $[\bZ, \Q] = \Q$, whereas they have the dilatation acting trivially on $\Q$.  That we have such a bracket may not be too surprising given the two $\N=1$ Galilean superalgebras found in \cite{Figueroa-OFarrill:2019ucc}.  There, we found the Galilean superalgebra obtained through a contraction of the Poincaré superalgebra, in which the only bracket involving $\s_{\bar{1}}$ was $[\Q, \Q] = -\P$; but, we also found a super-extension of the Galilean algebra which included $[\bH, \Q] = \Q$.  Analogously, it may be that since the super-extension of the Carroll-Lifshitz algebra presented in \cite{Bergshoeff:2015wma} comes from a limit, this bracket is not present.  The centrally-extended AdS Galilean algebra $\hat{\n}_+$ and the Bargmann algebra become
\begin{gather}
		[\B, \P] = \bH, \quad [\bZ, \B] = \P, \quad [\bZ, \P] = -\B, \nonumber \\ \text{and} \nonumber\\
		[\B, \P] = \bH, \quad [\bZ, \B] = \P, \nonumber
\end{gather}
respectively.  For now, we do not have any interpretation for these Lie algebras; however, it may be of interest to explore which theories require an extension of the Carroll algebra by a generator that acts on the $\so(3)$ vector modules through a rotation (the $\hat{\n}_+$ case) or a nilpotent matrix (the $\hat{\g}$ case).  Using this isomorphism, we may systematically produce Carroll superalgebras corresponding to each of the generalised Bargmann superalgebras.  With these connections to non-relativistic holography, it would be interesting to use the methods applied in this classification to the generalised Lifshitz and Schrödinger algebras presented in \cite{Figueroa-OFarrill:2018ygf}.
\\ \\
If interested in extending this work to Lifshitz algebras, it is important to note that papers concentrating on supersymmetric Lifshitz field theories such as \cite{Chapman:2015wha, Arav:2019tqm} utilise a super-extension of the Lifshitz algebra in which the modules of $\s_{\bar{1}}$ transform as scalars under $\so(3)$. Therefore, the classification method highlighted in this paper would need to be modified to take this change of $\so(3)$ action into account.
\\ \\
\begin{table}[h!]
  \centering
  \caption{Generalised Bargmann superalgebras with $[\Q, \Q] = \bZ$}
  \label{tab:holographic_algebras}
  \setlength{\extrarowheight}{2pt}
  \rowcolors{2}{blue!10}{white}
    \begin{tabular}{l|l*{6}{|>{$}c<{$}}}\toprule
      \multicolumn{1}{c|}{S} & \multicolumn{1}{c|}{$\k$} &  \multicolumn{1}{c|}{(S)B} & \multicolumn{1}{c|}{$\bH$}& \multicolumn{1}{c|}{$\bZ$}& \multicolumn{1}{c|}{$\B$} & \multicolumn{1}{c|}{$\P$} & \multicolumn{1}{c|}{$[\Q,\Q]$}\\
      \toprule
      1 & \hyperlink{a}{$\hat{\a}$} & - & & & & & \bZ  \\    
      2 & \hyperlink{a}{$\hat{\a}$} & - & \checkmark & & & & \bZ  \\    
      - & \hyperlink{a}{$\hat{\a}$} & \hyperlink{N2_a_1i}{\mathsf{1.i}} & \checkmark & & & &  \bZ  \\
      - & \hyperlink{a}{$\hat{\a}$} & \hyperlink{N2_a_22i}{\mathsf{2.2.i}} & \checkmark & & & \checkmark &  \bZ  \\
      4 & \hyperlink{n-}{$\hat{\n}_-$} & - & & & & & \bZ  \\ 
      5 & \hyperlink{n-}{$\hat{\n}_-$} & - & \checkmark & & & & \bZ  \\       
      - & \hyperlink{n-}{$\hat{\n}_-$} & \hyperlink{N2_ng_1i}{\mathsf{1.i}} & \checkmark & & & &  \bZ  \\
      - & \hyperlink{n-}{$\hat{\n}_-$} & \hyperlink{N2_n-_22i}{\mathsf{2.2.i}} & \checkmark & & & \checkmark &  \bZ  \\
      - & \hyperlink{n-}{$\hat{\n}_-$} & \hyperlink{N2_n-_32i}{\mathsf{3.2.i}} & \checkmark & & \checkmark & &  \bZ  \\ 
      6 & \hyperlink{n+}{$\hat{\n}_+$} & - & & & & &  \bZ  \\
      7 & \hyperlink{n+}{$\hat{\n}_+$} & - & \checkmark & & & &  \bZ  \\
      - & \hyperlink{n+}{$\hat{\n}_+$} & \hyperlink{N2_ng_1i}{\mathsf{1.i}} & \checkmark & & & &  \bZ  \\      
      - & \hyperlink{n+}{$\hat{\n}_+$} & \hyperlink{N2_n+_22i}{\mathsf{2.2.i}} & \checkmark & & & \checkmark &  \bZ  \\
      - & \hyperlink{n+}{$\hat{\n}_+$} & \hyperlink{N2_n+_41i}{\mathsf{4.1.i}} & \checkmark & & \checkmark & \checkmark &  \bZ  \\      
      8 & \hyperlink{g}{$\hat{\g}$} & - & & & & & \bZ  \\  
      9 & \hyperlink{g}{$\hat{\g}$} & - & \checkmark & & & & \bZ  \\  
      - & \hyperlink{g}{$\hat{\g}$} & \hyperlink{N2_ng_1i}{\mathsf{1.i}} & \checkmark & & & &  \bZ  \\  
      - & \hyperlink{g}{$\hat{\g}$} & \hyperlink{N2_g_22i}{\mathsf{2.2.i}} & \checkmark & & & \checkmark &  \bZ  \\  
      - & \hyperlink{g}{$\hat{\g}$}& \hyperlink{N2_g_32i}{\mathsf{3.2.i}} & \checkmark & & \checkmark & &  \bZ  \\  
      \bottomrule
    \end{tabular}
    \caption*{The first column gives the unique identifier for each of the $\N=1$ generalised Bargmann superalgebra $\s$,
    and the second column tells us the underlying generalised Bargmann algebra $\k$.
    For the $\N=2$ superalgebras, the third column indicates the sub-branch the generalised Bargmann superalgebra comes from,
    so that the reader may navigate back to find the conditions on the non-vanishing parameters
    of the superalgebra.  The next four columns tells us which of the $\s_{\bar{0}}$ generators $\bH, \bZ, \B$, 
    and $\P$ act on $\Q$.  Recall, $\J$ necessarily acts on $\Q$, so we do not need to state this explicitly.
    The final column shows which $\s_{\bar{0}}$ generators occur in the $[\Q, \Q]$ bracket.
     }
\end{table}
\section*{Acknowledgements}
I am grateful to José Figueroa-O'Farrill for first proposing this project and his continual guidance throughout.  I would also like to thank Jan Rosseel for his incredibly helpful comments on a previous version of this paper.  Finally, I want to thank Jelle Hartong for bringing my attention to the supersymmetric Aristotelian literature.  This work is funded by the Engineering and Physical Sciences Research Council (EPSRC).
\printbibliography
\end{document}